\documentclass[11pt,a4paper]{article}                           

\usepackage{amsmath}
\usepackage{amssymb}
\usepackage[usenames]{color} 
\usepackage{multirow}
\usepackage{graphicx}
\usepackage{epstopdf}
\usepackage{array}
\usepackage{hhline}
\usepackage{cite}
\usepackage{microtype,colortbl}
\usepackage{color}
\usepackage[usenames,dvipsnames]{xcolor}
\usepackage{pstool}
\usepackage{tikz}
\usepackage{ytableau}
\usepackage{latexsym,stmaryrd}
\usepackage[bookmarks=false]{hyperref}
 \usepackage{lscape}
\usepackage{caption}
 \usepackage{subcaption}
 \usepackage{tensor}
 \usepackage{dsfont}
\usepackage{amsthm}
\usepackage{enumerate}
\usepackage{tkz-euclide}
\usepackage{pifont}

\usepackage{ifpdf}
\def\hybrid{\topmargin -20pt    \oddsidemargin 0pt
        \headheight 0pt \headsep 0pt
        \textwidth 6.25in       
        \textheight 9 in       
        \marginparwidth .875in
        \parskip 5pt plus 1pt 
          \jot = 1.5ex
   }
\hybrid
\numberwithin{equation}{section}
\numberwithin{table}{section}

\newcommand{\beq}{\begin{equation}}  \newcommand{\eeq}{\end{equation}}
\newcommand{\bal}{\begin{aligned}}   \newcommand{\eal}{\end{aligned}}
\newcommand{\bea}{\begin{eqnarray}}  \newcommand{\eea}{\end{eqnarray}}

\newcommand{\bmat}{\left(\begin{array}}
\newcommand{\emat}{\end{array}\right)}

\usepackage{cancel}

\newcommand{\be}{\begin{equation}}
\newcommand{\ee}{\end{equation}}

\newtheorem{lem}{Lemma}
\newtheorem{cor}{Corollary}
\newtheorem{thm}{Theorem}

\newtheorem{conjecture}{Conjecture}

\definecolor{Gray}{gray}{0.95}

\begin{document}
\baselineskip=14pt
\parskip 5pt plus 1pt 

\vspace*{3cm}
\begin{center}
	{\Large\bfseries Finiteness Theorems and Counting Conjectures\\[.3cm] 
    for the Flux Landscape}\\[.3cm]
	
	\vspace{1.5cm}
	{\bf Thomas W.~Grimm}\footnote{t.w.grimm@uu.nl},
	{\bf Jeroen Monnee}\footnote{j.monnee@uu.nl}
	
	{\small
		\vspace*{.5cm}
		Institute for Theoretical Physics, Utrecht University\\ Princetonplein 5, 3584 CC Utrecht, The Netherlands\\[3mm]
	}
\end{center}
\vspace{3cm}
\begin{abstract}\noindent

In this paper, we explore the string theory landscape obtained from type IIB and F-theory flux compactifications. We first give a comprehensive 
introduction to a number of mathematical finiteness theorems, indicate how they have been obtained, and clarify their implications for the structure of the locus of flux vacua. Subsequently, in order to address finer details of the locus of flux vacua, we propose three mathematically precise conjectures on the expected number of connected components, geometric complexity, and dimensionality of the vacuum locus. With the recent breakthroughs on the tameness of Hodge theory, we believe that they are attainable to rigorous mathematical tools and can be successfully addressed in the near future. 
The remainder of the paper is concerned with more technical aspects of the finiteness theorems. In particular, we investigate their local implications and explain how infinite tails of disconnected vacua approaching the boundaries of the moduli space are forbidden. To make this precise, we present new results on asymptotic expansions of Hodge inner products near arbitrary boundaries of the complex structure moduli space. 

\end{abstract}

\newpage

\tableofcontents
\setcounter{footnote}{0}

\newpage
\setcounter{section}{0}

\section{Introduction}

String theory is known to have a plethora of solutions around which effectively four-dimensional quantum field theories coupled to classical gravity can be determined. The space of such lower-dimensional effective theories is often referred to as the string theory landscape. With this understanding, one might then inquire which of these theories can possibly describe our Universe. On a more basic level one might wonder if the number of such theories, after appropriately identifying equivalent theories, is at all finite. If this is not the case, one should seriously question the predictive capabilities of string theory. These issues were addressed at length in the seminal works of Douglas et al.~\cite{Douglas:2003um,Ashok:2003gk,Denef:2004ze,Acharya:2006zw}, which led to the general expectation that the string landscape is, in an appropriate sense, finite. This expectation is further corroborated by efforts in the swampland program \cite{Hamada:2021yxy,Grimm:2021vpn}, which aims to identify the fundamental properties an effective theory coupled to gravity should satisfy in order to admit a UV-completion, see \cite{Palti:2019pca,vanBeest:2021lhn} for reviews. Concurrently, there have been some major developments in the fields of algebraic geometry and logic that have lifted this expectation to the level of a mathematical theorem, at least in specific settings. The aim of the present work is to provide a collection of finiteness results, coming from the fields of Hodge theory and tame geometry, in a way that is hopefully accessible to physicists. In particular, we hope to clarify what has/has not been shown and to give some insights and new perspectives on the various proofs. We then draw from this knowledge to put forward a number of structural conjectures about the landscape.

To prove something about the whole string landscape is a daunting task. Therefore, we will focus our attention on a particular corner of the string landscape, namely those four-dimensional low-energy effective theories that arise from flux compactifications of type IIB/F-theory \cite{Dasgupta:1999ss,Giddings:2001yu,Grimm:2004uq,Grimm:2010ks}, see \cite{Grana:2005jc,Douglas:2006es,Denef:2008wq} for reviews. These compactifications, viewed from the dual M-theory perspective \cite{Becker:1996gj}, are specified by a family of Calabi--Yau fourfolds varying in moduli, together with a background flux $G_4$. The moduli are generically stabilized at the critical points of the scalar potential induced by the flux, leading to a typically large landscape of flux vacua. Such vacua are of great phenomenological interest, as they feature spontaneous supersymmetry breaking down to $\mathcal{N}=1$ or even $\mathcal{N}=0$, and may eventually lead to de Sitter solutions \cite{Kachru:2003aw,Balasubramanian:2005zx} with a small cosmological constant \cite{Demirtas:2019sip,Alvarez-Garcia:2020pxd,Demirtas:2020ffz,Honma:2021klo,Demirtas:2021ote,Broeckel:2021uty,Bastian:2021hpc}. A crucial point is that the flux has to satisfy a number of consistency conditions, as the effective theory originates from a UV complete theory of quantum gravity. These conditions include a quantization condition and the so-called tadpole cancellation condition. Consequently, the central question is whether there exists only a finite number of fluxes and associated critical points that simultaneously satisfy these consistency conditions. To be clear, we will consider the issue of finiteness within a given family of Calabi--Yau fourfolds, varying in moduli. In particular, we will not address whether there exist only finitely many distinct topological classes of Calabi--Yau fourfolds, which is an interesting question on its own. 

In the context of IIB/F-theory flux compactification, initial evidence for this suggested finiteness was presented in the works \cite{Ashok:2003gk,Denef:2004ze,Acharya:2006zw}, which where later formalized in the mathematical works \cite{Douglas:2004zu,Douglas:2004kc,Douglas:2005df}. The underlying approach in these studies involved approximating the total number or index of flux vacua by integrating a suitable distribution of flux vacua over the moduli space, and showing that the latter is finite \cite{Eguchi:2005eh,Douglas:2006zj,Lu:2009aw}. From this distribution one could also obtain rough estimates for the total number of flux vacua, leading to the infamous number $10^{500}$. However, one critical limitation in their analysis was the relaxation of the quantization condition on the flux. Indeed, in order to give a complete proof of the finiteness of flux vacua, one expects that the quantization condition is crucial. 

Let us be a bit more specific on the kinds of vacua we will consider in this work. Importantly, we will focus on the stabilization of the complex structure moduli. In contrast, the K\"ahler moduli, whether stabilized or not, will not play an important role. Our analysis will involve two qualitatively different classes of vacua. Both classes correspond to the global minima of the flux-induced scalar potential and yield Minkowski vacua. In terms of the four-dimensional $\mathcal{N}=1$ supergravity formulation, both classes satisfy $D_iW_{\mathrm{flux}}=0$, where $W_{\mathrm{flux}}$ denotes the flux superpotential. The two classes are distinguished by whether they satisfy $W_{\mathrm{flux}}=0$ or $W_{\mathrm{flux}}\neq 0$, and are referred to as Hodge vacua and self-dual vacua, respectively. This is summarized in table \ref{tab:vacua}. Let us also emphasize that, for the purpose of this work, it is not necessary that all complex structure moduli are stabilized. As such, the vacuum locus may consist of various connected components of different dimensionality. 

\begin{table}[t!]
    \centering
    \begin{tabular}{|c|c|c|}
    \hline
      \rule[-.15cm]{0cm}{.55cm}   \textbf{definition} & \textbf{single solution} & \textbf{collection of all solutions} \\ \hline
       \rule[-.15cm]{0cm}{.55cm}  $D_iW_{\mathrm{flux}}=0,\,W_{\mathrm{flux}}=0$ & Hodge vacuum & Hodge locus\\ \hline
       \rule[-.15cm]{0cm}{.55cm}  $D_iW_{\mathrm{flux}}=0,\,W_{\mathrm{flux}}\neq 0$ & self-dual vacuum & self-dual locus \\ \hline
    \end{tabular}
    \caption{Summary of the two classes of vacua considered in this work, together with the terminology employed to describe a single vacuum and the full vacuum locus. Here $W_{\mathrm{flux}}$ denotes the flux-induced superpotential for the complex structure moduli.}
    \label{tab:vacua}
\end{table}

We now provide some more details on the finiteness results we will discuss, starting with the case of Hodge vacua. In the mathematics literature, primitive self-dual fluxes that additionally satisfy $W_{\mathrm{flux}}=0$ are a particular example of ``Hodge classes''. These are integral classes of type (2,2) in the Hodge decomposition of the primitive middle cohomology of the Calabi--Yau fourfold. One of the major milestones of Hodge theory is a theorem of Cattani, Deligne, and Kaplan which states that the locus of Hodge classes is a countable union of algebraic varieties \cite{CDK}. Interestingly, the same result can also be derived by assuming the Hodge conjecture to be true. For this reason, the result of Cattani, Deligne, and Kaplan is often viewed as the strongest evidence for the Hodge conjecture. Furthermore, if the flux satisfies the tadpole cancellation condition, meaning it has a bounded self-intersection, then the locus is in fact a \textit{finite} union of algebraic varieties. In particular, its number of connected components, which counts the number of Hodge vacua with possibly flat directions, is finite. 

Let us now turn our attention to generic self-dual flux vacua, for which initial finiteness results were presented in \cite{Schnellletter}, see also \cite{Grimm:2020cda}, for the case of a single complex structure modulus. In these works a detailed description of the Hodge norm of the $G_4$ flux was obtained by employing deep results in asymptotic Hodge theory, such as the one-variable $\mathrm{Sl}(2)$-orbit theorem of Cattani, Kaplan and Schmid \cite{CKS}. The finiteness of self-dual vacua in the general multi-variable case was proven recently in \cite{Bakker:2021uqw}, see also \cite{Grimm:2021vpn}. In contrast to the one-variable case, the proof of the multi-variable case is much more involved and relied heavily on recent advances in the field of tame geometry, such as the definability of the period map \cite{BKT}. 

The main technical result of the present work is to provide another perspective on the finiteness of self-dual vacua in the multi-variable case, without relying on methods from tame geometry. Instead, we generalize the analysis performed in \cite{Grimm:2020cda,Schnellletter} by considering the $\mathrm{Sl}(2)$-orbit theorem in its full multi-variable glory. In particular, we present general formulas for asymptotic Hodge inner products of arbitrary fluxes that include infinite towers of corrections. The derivation of these formulae utilizes a multi-variable generalization of the CKS recursion \cite{CKS}, see also \cite{Grimm:2021ikg}.
We then apply these results to prove the finiteness of self-dual flux vacua within a well-defined approximation that is often used in the study of asymptotic Hodge theory. This provides a good intuition for why finiteness is likely to persist, even when there are multiple moduli at play. Additionally, our improved expressions for Hodge inner products are of separate interest and may be used to generalize previous analyses to sub-leading orders. As an example, we derive an asymptotic formula for the central charge of D3-particles in the context of type IIB compactifications, which is valid near any boundary of the complex structure moduli space and generalizes the results of \cite{Bastian:2020egp}. 

Finally, in order to study more detailed features of the locus of flux vacua beyond just its finiteness, we outline a set of three concrete mathematical conjectures which may be addressed in the near future by combining techniques from asymptotic Hodge theory, (sharply) o-minimal geometry and the theory of unlikely intersections. The first two conjectures concern the enumeration of flux vacua, in particular Hodge vacua, as well as a candidate notion of geometric complexity, as developed by Binyamini and Novikov in \cite{binyamini2022}, for the locus of self-dual flux vacua. The third conjecture is a modified version of the tadpole conjecture of \cite{Bena:2020xrh}, adapted to the special class of Hodge vacua and is instead concerned with the dimensionality of the vacuum locus. In other words, it is related to the existence of a flat direction in the scalar potential. For related work on the tadpole conjecture, we refer the reader to \cite{Braun:2020jrx,Bena:2021wyr,Marchesano:2021gyv,Lust:2021xds,Plauschinn:2021hkp,Grana:2022dfw,Lust:2022mhk,Coudarchet:2023mmm,Braun:2023pzd}.

\subsubsection*{Outline of the paper}
The paper can be roughly divided into three different parts, which are organized as follows.
\begin{enumerate}[I.]
    \item Sections \ref{sec:flux_compactification} and \ref{sec:finiteness_intro}: this comprises the main physics content of the paper.
    \begin{enumerate}[i.]
        \item In section \ref{sec:flux_compactification} we provide a brief review of flux compactification in the language of F-theory. In particular, we recall the quantization and tadpole cancellation conditions that the flux should satisfy, and define the two classes of vacua that will be studied in the rest of the paper. 
        \item In section \ref{sec:finiteness_intro} we present a general discussion on the issue of finiteness of flux vacua and illustrate the main difficulties that arise. We then formulate and discuss a number of precise finiteness theorems, for both Hodge vacua and self-dual vacua, which have been established in the literature. 
    \end{enumerate}        
    \item Section \ref{sec:future_questions}: here we turn towards some future prospects and challenges that we believe to be worthy of further study. In particular, we present three concrete mathematical conjectures concerning the counting of Hodge vacua, the complexity of the landscape and the tadpole conjecture, and propose how these matters may be investigated using the methods of (sharply) o-minimal structures, unlikely intersection theory and Hodge theory. 
    \item Sections \ref{sec:asymp_Hodge_inner_products} and \ref{sec:self_dual_locus}: this comprises the main mathematical content of the paper. 
    \begin{enumerate}[i.]
        \item In section \ref{sec:asymp_Hodge_inner_products} we perform a general analysis of asymptotic Hodge inner products using the machinery of the multi-variable $\mathrm{Sl}(2)$-orbit theorem. Additionally, to exemplify possible applications of our general formulae we derive the following asymptotic expansion for the central charge of D3-particles with charge $q$ in type IIB compactifications
\begin{equation}
\label{eq:central_charge_intro}
    | \mathcal{Z} |= y^\ell\left|\frac{\langle q, \Omega_\infty\rangle_\infty}{||\Omega_\infty||_\infty}+\sum_{k=1}^\infty\sum_{s\leq k-1} y^{-k+\frac{1}{2}s}\frac{\langle f_k^s q, \Omega_\infty\rangle_\infty}{||\Omega_\infty||_\infty}\right|\,,
\end{equation}
which is valid in the region where a single saxion $y$ is large, and present the full multi-variable generalization in equation \eqref{eq:central_charge_nilpotent}. The meaning of the various objects appearing in \eqref{eq:central_charge_intro} is explained in detail in section \ref{subsec:nilpotent_orbit_expansion}.
        \item In section \ref{sec:self_dual_locus} we present additional evidence for the finiteness of self-dual flux vacua, by employing the methods of asymptotic Hodge theory introduced in section \ref{sec:asymp_Hodge_inner_products}. 
    \end{enumerate}
\end{enumerate}
For a first reading, we suggest the reader to focus on section \ref{sec:finiteness_intro} (and section \ref{sec:flux_compactification}, if necessary), as this contains all the main results that will be discussed in this work. The reader who is interested in outstanding questions on the structure of the vacuum locus and suggestions for future endeavours in o-minimal geometry and Hodge theory, formulated as a set of three concrete mathematical conjectures, is highly encouraged to read section \ref{sec:future_questions}. Those who would like to delve deeper into some aspects of multi-variable asymptotic Hodge theory, as well as their usage in the computation of asymptotic Hodge inner products, or the proof of some of the finiteness theorems, are invited to read sections \ref{sec:asymp_Hodge_inner_products} and \ref{sec:self_dual_locus}. Some additional details as well as some illustrative examples are collected in appendices \ref{app:additional_proofs} and \ref{app:Hodge_norms}. Finally, for the brave readers who already have some familiarity with (mixed) Hodge theory, we have included a reformulated version of the classic proof of the finiteness of Hodge classes in appendix \ref{app:Hodge_locus}. 

\section{F-theory Flux Compactification}
\label{sec:flux_compactification}

In this section we provide a brief review of F-theory flux compactification. For further details we refer the reader to \cite{Grana:2005jc,Douglas:2006es,Denef:2008wq}. We establish our notation and conventions but present no new results. The reader familiar with the topic can safely skip this section. 

\subsubsection*{Low-energy effective theory}

It is well known that compactification of F-theory on a Calabi--Yau fourfold $Y_4$, elliptically fibered over a base $B_3$, yields an effective four-dimensional $\mathcal{N}=1$ supergravity theory at low energies. In the absence of fluxes, the resulting theory contains a (typically large) number of massless fields/moduli. Throughout this work, we will be concerned with only a subset of the spectrum of the low-energy effective theory. To be precise, we will consider the complex scalar fields $z^i$, $i=1,\ldots, h^{3,1}(Y_4)$, that correspond to the complex structure deformations of the Calabi--Yau fourfold. In the orientifold or weak-coupling limit, these deformations collectively describe the complex structure deformations of the Calabi-Yau threefold $Y_3$ that is a double cover of $B_3$, as well as the deformations of the D7-branes and the type IIB axio-dilaton $\tau$. 
Besides the complex structure moduli, the low-energy effective theory contains other massless fields as well, most notably the $h^{1,1}(Y_4)$ K\"ahler moduli that correspond to K\"ahler deformations of the Calabi--Yau fourfold. In the F-theory limit, one of these K\"ahler moduli, playing the role of the volume modulus of the elliptic fibre, should be sent to an appropriate limit, while the remaining $h^{1,1}(Y_4)-1$ K\"ahler moduli may be stabilized by various methods. Typical methods involve stabilization through non-perturbative $g_s$ corrections to the superpotential as in the KKLT scenario \cite{Kachru:2003aw}, or through perturbative $\alpha'$ corrections to the K\"ahler potential as in the Large Volume Scenario \cite{Balasubramanian:2005zx}. For the purpose of this work, however, the K\"ahler moduli, whether stabilized or not, will not play an important role. 

\subsubsection*{Fluxes}
In the absence of fluxes, there is no energetic obstruction to performing a complex structure deformation of the underlying Calabi--Yau manifold. In the effective theory, this manifests itself in the fact that the complex structure moduli $z^i$ are massless fields. In the presence of fluxes, however, this is no longer the case. Recall that a flux in F-theory corresponds to a harmonic four-form $G_4$ on the Calabi--Yau fourfold. Furthermore, Dirac quantization imposes that $G_4$ is integral, hence it can be viewed an element of the integral middle de Rham cohomology $H^4(Y_4,\mathbb{Z})$.\footnote{To be precise, it is the quantity $G_4-\frac{p_1(Y_4)}{4}$, where $p_1(Y_4)$ denotes the first Pontryagin class of $Y_4$. This shift will not change any of the arguments made in this paper, hence we assume for simplicity that $G_4$ itself is integral.} In particular, this means that integrals of $G_4$ over closed 4-cycles are integers. This quantization condition will play a central role in all of the finiteness results we will describe. It is therefore important to highlight that the quantization condition arises from the fact that we are considering a low-energy effective theory of \textit{quantum} gravity.

\subsubsection*{Tadpole cancellation condition}
Besides the quantization condition, there is another condition that is imposed on the four-form flux $G_4$. This condition originates from the fact that on the compact Calabi--Yau $Y_4$ the total D3-brane charge (or equivalently M2-brane charge) has to vanish. Since the $G_4$ flux itself also induces this charge, this result in the tadpole cancellation condition
\cite{Sethi:1996es}
\begin{equation}\label{eq:tadpole_condition}
    N_{\text{D3}}+\frac{1}{2}\int_{Y_4}G_4\wedge G_4 = \frac{\chi(Y_4)}{24}\,,
\end{equation}
where $\chi(Y_4)$ denotes the Euler characteristic of the Calabi--Yau $Y_4$ and $N_{\text{D3}}$ denotes the net number of spacetime-filling $\left(\text{D3}-\overline{\text{D3}}\right)$ branes. In the remainder of this work, we will view the tadpole cancellation condition as giving an upper bound on the self-intersection of $G_4$ and write it as
\begin{equation}\label{eq:tadpole_bound}
    \int_{Y_4}G_4\wedge G_4 \leq L\,,
\end{equation}
for some fixed integer $L$ which we will refer to as the tadpole bound. We stress that the assumption of compactness of $Y_4$ is crucial in deriving \eqref{eq:tadpole_condition}, and is motivated by the need for gravity (i.e.~a finite lower-dimensional Planck mass) in the effective theory. Therefore, in contrast to the quantization condition, the tadpole cancellation condition arises from the fact that we are considering a low-energy effective theory of quantum \textit{gravity}. 

\subsection{Moduli stabilization: Hodge theory formulation}

\subsubsection*{Scalar potential}

The presence of a non-trivial four-form flux induces a scalar potential in the low-energy effective theory given by \cite{Haack:2001jz,Grimm:2010ks} 
\begin{equation}
\label{eq:scalar_potential}
    V_{\mathrm{flux}}(z) = \frac{1}{\mathcal{V}_b^2}\int_{Y_4}\left(G_4\wedge\star\,G_4-G_4\wedge G_4 \right)\,,
\end{equation}
where $\mathcal{V}_b$ denotes the volume of the base $B_3$ and $\star$ denotes the Hodge star operator on $Y_4$, which is itself a function of the complex structure moduli. The first term in \eqref{eq:scalar_potential} corresponds to the integrated kinetic energy of the M-theory 3-form gauge field, while the second term originates from the integrated Bianchi identity for $G_4$. Since the potential depends on the complex structure moduli via the Hodge star operator, there will be energetically favoured combinations of $z^i$ and $G_4$ for which the potential \eqref{eq:scalar_potential} is minimized. Such configurations are referred to as \textit{flux vacua}.

\subsubsection*{Self-dual vacua}

The scalar potential \eqref{eq:scalar_potential} is positive semi-definite and attains a global minimum whenever the four-form flux is \textit{self-dual}, i.e.
\begin{equation}
    \label{eq:self_dual}
    G_4=\star\,G_4\,.
\end{equation}
A self-dual vacuum corresponds to a Minkowski vacuum, since $V=0$. We stress that one should regard the condition $G_4=\star\,G_4$ as a condition in cohomology. To elucidate the self-duality condition \eqref{eq:self_dual}, we recall that the middle cohomology of $Y_4$ admits a Hodge decomposition
\begin{equation}\label{eq:Hodge_decomp}
    H^4\left(Y_4,\mathbb{C}\right) = H^{4,0}\oplus H^{3,1}\oplus H^{2,2}\oplus H^{1,3}\oplus H^{0,4}\,,
\end{equation}
into harmonic $(p,q)$-forms. One can show that the self-duality condition \eqref{eq:self_dual} implies that $G_4$ has no $(3,1)$ component. In other words, $G_4$ has a decomposition (recall that $G_4$ is real)
\begin{equation}
    G_4 = \left(G_4\right)^{4,0}+\left(G_4\right)^{2,2}+\left(G_4\right)^{0,4}\,.
\end{equation}
The self-duality condition therefore comprises $h^{3,1}$ complex equations for the $h^{3,1}$ complex structure moduli and hence one expects that a generic choice of $G_4$ stabilizes all moduli. It is, however, not at all obvious whether this holds true if $G_4$ is constrained by the tadpole cancellation condition \eqref{eq:tadpole_condition}. In fact, it was recently suggested that indeed this naive expectation may fail when $h^{3,1}$ becomes sufficiently large \cite{Bena:2020xrh}, leading to the so-called tadpole conjecture. See also \cite{Braun:2020jrx,Bena:2021wyr,Marchesano:2021gyv,Lust:2021xds,Plauschinn:2021hkp,Grana:2022dfw,Lust:2022mhk,Coudarchet:2023mmm,Braun:2023pzd} for related works. 

\subsubsection*{Hodge vacua}
A self-dual vacuum will be referred to as a \textit{Hodge vacuum} if, in addition, $G_4$ only has a (2,2)-component and is primitive. The latter means that 
\begin{equation}\label{eq:primitive}
    J\wedge G_4=0\,,
\end{equation}
where $J$ denotes the K\"ahler $(1,1)$-form on $Y_4$. In mathematics, cohomology classes of this type are referred to as Hodge classes. As will be elaborated upon in section \ref{sec:finiteness_intro}, such classes play a very special role in Hodge theory. 

\subsection{Moduli stabilization: superpotential formulation}

A possibly more familiar formalism to describe the vacua of four-dimensional $\mathcal{N}=1$ supergravity theories is the superpotential formalism. Although we will not use this language much for the rest of this work, we include it here for the reader's convenience.

\subsubsection*{Scalar potential}

In any four-dimensional $\mathcal{N}=1$ supergravity theory, the F-term contribution to the scalar potential can be written as
\begin{equation}
\label{eq:scalar_potential_W}
    V = e^K\left(G^{I\bar{J}}D_I W \overline{D_J W}-3|W|^2\right)\,,\qquad D_IW = \left(\partial_I+\partial_I K\right) W\,,
\end{equation}
where $K$ is a K\"ahler potential that determines a K\"ahler metric $G_{I\bar{J}}$ and $W$ is the holomorphic superpotential.

\subsubsection*{F-theory realization}

In the context of F-theory compactification, the indices $I,\bar{J}$ in \eqref{eq:scalar_potential_W} run over both the complex structure moduli and the K\"ahler moduli. To clarify the relation between the scalar potentials \eqref{eq:scalar_potential} and \eqref{eq:scalar_potential_W} we need to specify the K\"ahler potential and superpotential.
\begin{itemize}
\item The K\"ahler potential $K$ is given by
\begin{equation}
\label{eq:Kahler_potential}
    K = -2\log\mathcal{V}_b-\log \int_{Y_4}\Omega\wedge\overline{\Omega}\,.
\end{equation}
The first term is the tree-level K\"ahler potential for the complex coordinates $T^\alpha$ that depend on the K\"ahler moduli. The second term is the K\"ahler potential for the complex structure moduli, depending on the holomorphic $(4,0)$-form $\Omega(z)$. The tree-level K\"ahler potential enjoys the no-scale property
\begin{equation}\label{eq:no_scale}
    G^{\alpha\bar{\beta}}\partial_\alpha K \overline{\partial_\beta K}=3\,.
\end{equation}
It is important to stress, however, that $K$ receives both perturbative corrections, coming e.g.~from $\alpha'$ corrections to the ten-dimensional IIB supergravity action, as well as non-perturbative corrections coming from worldsheet instantons. These corrections will generically break the no-scale structure \eqref{eq:no_scale} of the K\"ahler potential. 
\item The superpotential $W$ is given by the flux-induced superpotential $W_{\mathrm{flux}}$, where \cite{Gukov:1999ya,Haack:2001jz}
\begin{equation}
\label{eq:superpotential}
    W_{\mathrm{flux}}(z) = \int_{Y_4}G_4\wedge \Omega(z)\,.
\end{equation}
In contrast to the K\"ahler potential $K$, the superpotential $W$ is perturbatively exact and only receives non-perturbative corrections coming e.g.~from Euclidean D3-brane instantons and gaugino condensation. Note that, since our discussion is restricted to the perturbative level, $W$ does not depend on the K\"ahler moduli. In particular, we have
\begin{equation}\label{eq:DW_Kahler}
    D_\alpha W_{\mathrm{flux}} = \left(\partial_\alpha K\right)W_{\mathrm{flux}}\,,
\end{equation}
where again $\alpha$ runs over the complex coordinates involving the K\"ahler moduli. 
\end{itemize}
Combining the no-scale condition \eqref{eq:no_scale} together with the simplification \eqref{eq:DW_Kahler}, the scalar potential reduces to
\begin{equation}
\label{eq:scalar_potential_noscale}
    V = e^K G^{i\bar{\jmath}}D_i W_{\mathrm{flux}}\overline{D_jW_{\mathrm{flux}}}\,,
\end{equation}
where $i,\bar{\jmath}$ run over the complex structure moduli only. In particular, note that the $-3|W|^2$ term has dropped out. As a result, the scalar potential \eqref{eq:scalar_potential_noscale} is positive semi-definite and can be seen to be equivalent to \eqref{eq:scalar_potential}. 

\subsubsection*{Vacua}

To stabilize the complex structure moduli, it then remains to solve the condition $D_iW_{\mathrm{flux}}=0$. This is equivalent to the condition that $G_4$ has no $(3,1)$ nor $(1,3)$ component in the Hodge decomposition \eqref{eq:Hodge_decomp}. If $G_4$ is primitive, as we will assume throughout this work, this is in turn equivalent to the self-duality condition \eqref{eq:self_dual}. If additionally $W_{\mathrm{flux}}=0$ then $G_4$ also has no $(4,0)$ and $(0,4)$ components, so $G_4$ is purely of type $(2,2)$. In particular, in this case the vacuum corresponds to a Hodge vacuum. This is summarized in table \ref{tab:vacuum_conditions}.

\begin{table}[t!]
\centering
\begin{tabular}{|c|c|c|}
\hline
\rule[-.15cm]{0cm}{.55cm}                   & \quad Hodge decomposition of $G_4$ \quad  & superpotential \\ \hline
\rule[-.15cm]{0cm}{.55cm} Hodge vacuum     & $(2,2)$ & $D_i W_{\mathrm{flux}}=W_{\mathrm{flux}}=0$               \\ \hline
\rule[-.15cm]{0cm}{.55cm} self-dual vacuum &  $(4,0)+(2,2)+(0,4)$ & $D_i W_{\mathrm{flux}}=0\,,W_{\mathrm{flux}}\neq 0$\\
\hline
\end{tabular}
\caption{Overview of the conditions for a Hodge/self-dual flux vacuum in terms of the Hodge decomposition of $G_4$, see \eqref{eq:Hodge_decomp}, and the flux-induced superpotential $W_{\mathrm{flux}}$.}
\label{tab:vacuum_conditions}
\end{table}

\section{Finiteness of Flux Vacua}
\label{sec:finiteness_intro}

In section \ref{sec:flux_compactification} we have reviewed the conditions on the four-form flux $G_4$ and the complex structure moduli $z^i$ that determine the locus of self-dual flux vacua. In the remainder of this work, we will be interested in gaining a more detailed understanding of what this locus looks like. In particular, our aim is to ascertain whether it consists of a finite number of points (or, more precisely, a finite number of connected components). The purpose of this section is two-fold. First, we provide a general discussion to emphasize the main non-trivial aspects of the problem, illustrated with a simple example of a rigid $Y_3\times T^2$ compactification. Second, we formulate the problem within the broader framework of Hodge theory and present a number of exact finiteness theorems that have been established in the literature. The subsequent sections will delve into a more detailed examination of these theorems and their proofs.  

\subsection{Why finiteness is non-trivial}
\label{subsec:finiteness_intro}

\subsubsection*{Infinite tails of vacua?}
First, let us emphasize again that we are investigating the finiteness of vacua within a fixed topological class of Calabi--Yau fourfolds, but varying in complex structure moduli. In this setting, we recall from section \ref{sec:flux_compactification} that a self-dual flux vacuum consists of a pair $(z^i,G_4)$, where $z^i$ are the complex structure moduli and $G_4$ is the four-form flux, satisfying three conditions 
\begin{equation}
\label{eq:vacuum_conditions}
\boxed{ \quad 
    G_4\in H^4\left(Y_4,\mathbb{Z}\right)\,,\qquad G_4=\star\,G_4\,,\qquad \int_{Y_4}G_4\wedge G_4 \leq L\,\quad }
\end{equation}
where $\star$ is to be evaluated at $z^i$, and $L$ is some positive integer that reflects the tadpole bound. Note that, for some choices of the flux, it may happen that not all $z^i$ are stabilized, meaning that the scalar potential has flat directions, in which case we count each connected component of the higher-dimensional vacuum locus as a single vacuum. The question, then, is how many solutions to \eqref{eq:vacuum_conditions} exist as one varies over all possible choices of $G_4$. Naively, it appears that $G_4$ varies over an infinite lattice. However, upon combining the self-duality condition and the tadpole condition, one finds the relation
\begin{equation}\label{eq:tadpole_self-dual}
    \int_{Y_4}G_4\wedge\star\,G_4 \leq L\,.
\end{equation}
At a non-singular point in the moduli space, the left-hand side of \eqref{eq:tadpole_self-dual} is a manifestly positive-definite quadratic form in the fluxes. Therefore, at a fixed point $z^i$ the constraint \eqref{eq:tadpole_self-dual} restricts the fluxes to lie in the interior of some ellipsoid inside the flux lattice, whose exact shape depends on the chosen value of the moduli. Clearly such a region contains only finitely many discrete lattice points and hence finitely many self-dual flux vacua. Furthermore, this remains true as long as one varies the moduli $z^i$ over a compact subset of the moduli space. 

However, it is not at all obvious what happens as the moduli vary over an unbounded set, as is typically the case in the context of Calabi--Yau compactifications. In other words, one might find an accumulation of vacua as one approaches a boundary of the moduli space. Along such limits the Hodge star operator may degenerate, causing some directions of the ellipsoid to become arbitrarily large and thus include arbitrarily many lattice points. In order to address the fate of these potentially infinite tails of vacua, one has to deal with the following two major roadblocks:
\begin{itemize}
    \item \textbf{Hodge star behaviour}: It is necessary to understand all possible ways in which the Hodge star can degenerate as one approaches an arbitrary boundary in the moduli space of any Calabi--Yau fourfold, in particular with an arbitrary number of moduli. 
    \item \textbf{Path-dependence}: When there are multiple moduli at play, the degeneration of the Hodge star is highly dependent on how one approaches a given boundary in the moduli space. 
\end{itemize}
The possible degenerations of the Hodge star are well-studied in the field of asymptotic Hodge theory, as will be reviewed in the next sections. The issue of path-dependence is, however, a bit more subtle. For Hodge vacua this issue can in fact be dealt with using just Hodge theoretic techniques. Essentially, one applies a clever inductive reasoning to range over all possible hierarchies between the moduli. In section \ref{sec:self_dual_locus} we employ a new strategy to tackle this issue, which is valid within a certain approximation that will be made precise. However, in order to address the fate of self-dual vacua in full generality, these techniques are likely to be insufficient. Recently, these issues were overcome by incorporating deep results in o-minimal geometry on the tameness of Hodge theory \cite{BKT, Bakker:2021uqw}. 

\subsubsection*{An example: rigid $Y_3\times T^2$}

So far our discussion has been rather abstract. In order to illustrate some of the points we have made above, let us consider a simple example. The point of the example will be to highlight the possible presence of infinite tails of vacua and to give an idea why such tails nevertheless cannot appear. However, we stress that, due to its simplicity, the example will not give an adequate indication for the complexity of the general problem. In particular, the issue of path-dependence will not play a role here.

\begin{figure}[t!]
		\begin{tikzpicture}[scale=0.8]
			
			
			\draw[step=1cm, gray, very thin] (-4.1,0) grid (4.1,7.1);
			\draw[thick,->] (-4,0) -- (4,0) node[anchor =  west] {$\mathrm{Re}\;\tau$};
			\draw[thick,->] (0,0) -- (0,7) node[anchor = south ] {$\mathrm{Im}\;\tau$};   
			
			\draw[very thick] (2,0) arc (180:360: -2);
			\draw[dashed] (4,0) arc (180:360: -2);
			\draw[dashed] (0,0) arc (180:360: -2);
			\draw[dashed,thin] (-2,0) arc (180:270: -2);
			\draw[dashed,thin] (4,2) arc (270:360: -2);
			
			\draw[very thick] (1,1.73) -- (1,7);
			\draw[very thick] (-1,1.73) -- (-1,7);
			
			\draw[dashed] (3,1.73) -- (3,7);
			\draw[dashed] (-3,1.73) -- (-3,7);
			
			\fill[gray,opacity=0.3] (1,1.73) arc(60:120:2) -- (-1,7) -- (1,7) -- (1,1.73);
			
			\fill[gray,opacity=0.1] (3,1.73) arc(60:120:2) -- (1,7) -- (3,7) -- (3,1.73);
			\fill[gray,opacity=0.1] (-1,1.73) arc(60:120:2) -- (-3,7) -- (-1,7) -- (-1,1.73);
			\fill[gray,opacity=0.1] (1,1.73) arc(60:120:2) arc(60:0:2) arc(0:-60:-2);
			
			\foreach \x in {-2,-1,0,1,2}
			\draw (2*\x,1pt) -- (2*\x,-1pt) node[anchor=north]{$\x$};
			
			\filldraw[red!60] (0,4) circle (4pt);
			
		\end{tikzpicture}
		\qquad        
		\raisebox{0.4cm}{
			\begin{tikzpicture}[scale=0.7]
				
				
				\draw[color=red!60, fill=red!5,very thick] (0,0) ellipse (4.0/3.0 and 3.0);
				
				\draw[step=1cm, gray, very thin] (-4.2,-4.2) grid (4.2,4.2);
				\draw[thick,->] (-4,0) -- (4,0) node[anchor=west]{$v_1$};
				\draw[thick,->] (0,-4) -- (0,4) node[anchor=south]{$v_2$};
				
				\filldraw[black] (0,0) circle (2pt);
				\filldraw[black] (1,0) circle (2pt);
				\filldraw[black] (0,1) circle (2pt);
				\filldraw[black] (1,1) circle (2pt);
				\filldraw[black] (-1,-1) circle (2pt);
				\filldraw[black] (-1,0) circle (2pt);
				\filldraw[black] (0,-1) circle (2pt);
				\filldraw[black] (1,-1) circle (2pt);
				\filldraw[black] (-1,1) circle (2pt);
				\filldraw[black] (0,-2) circle (2pt);
				\filldraw[black] (0,2) circle (2pt);
				
			\end{tikzpicture}
		}
		\caption{A geometric depiction of the tadpole bound \eqref{eq:tadpole_torus} for the two-torus. Left: a fundamental domain for the Teichm\"uller parameter $\tau$. Right: the corresponding region inside the flux lattice where the tadpole bound is satisfied. For simplicity, we have considered a point with $\mathrm{Re}\,\tau=0$.}
		\label{fig:tadpole_torus}		
\end{figure}
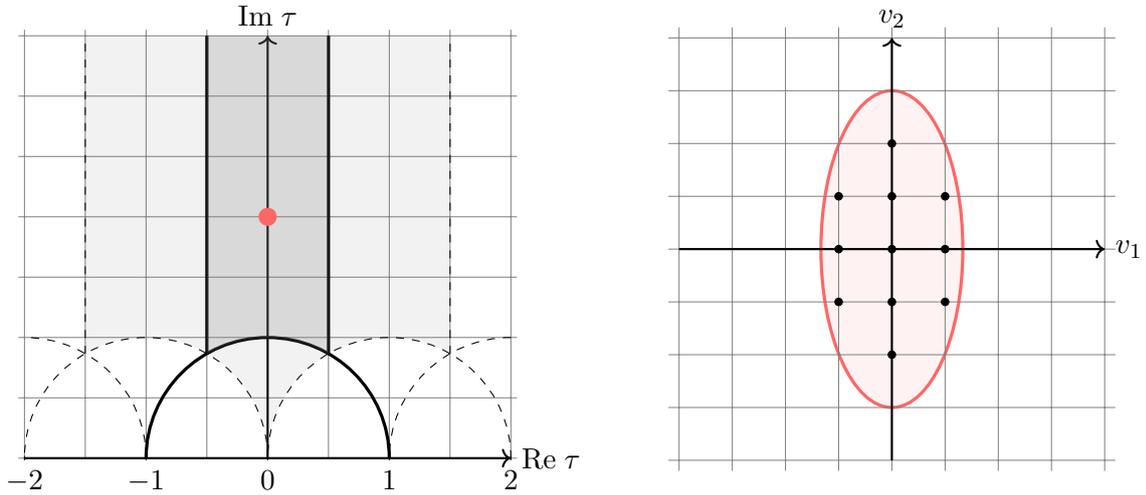

We take $Y_4$ to be a direct product
\begin{equation}
	Y_4 = Y_3\times T^2\,,
\end{equation}
with $Y_3$ a rigid Calabi--Yau threefold (i.e.~having no moduli) and $T^2$ a two-torus, whose complex structure modulus will be denoted by $\tau$, with $\mathrm{Im}\,\tau>0$. We consider a one-form flux on the torus 
\begin{equation}
	v\in H^1(T^2,\mathbb{Z}[i])\,,\qquad v=\begin{pmatrix} v_1\\v_2\end{pmatrix}\,,
\end{equation}
where $v_1,v_2\in\mathbb{Z}[i]$ are Gaussian integers. The vector representation of $v$ is taken with respect to the standard basis of 1-cycles on the torus, in terms of which the period vector is given simply by $(1,\tau)$. Then one readily computes
\begin{equation}\label{eq:tadpole_torus}
	||v||^2 = \int_{T^2} v\wedge\star\,\bar{v} = |v_1|^2\,\mathrm{Im}\,\tau+\frac{|v_2-v_1\mathrm{Re}\,\tau|^2}{\mathrm{Im}\,\tau}\,.
\end{equation}
As expected, for a fixed value of $\tau$ a region inside the flux lattice of bounded $||v||$ corresponds to the interior of an ellipsoid. Furthermore, the semi-major and semi-minor axes of the ellipsoid scale as $\mathrm{Im}\,\tau$ and $1/\mathrm{Im}\,\tau$, respectively. The situation is illustrated in figure \ref{fig:tadpole_torus}.\footnote{It should be noted that not necessarily all fluxes choices depicted in figure \ref{fig:tadpole_torus} satisfy the vacuum conditions.} Indeed, as one approaches the weak-coupling point $\mathrm{Im}\,\tau\rightarrow\infty$, corresponding to the boundary of the moduli space, one of the axes of the ellipsoid blows up, while the other shrinks. Therefore, by letting $\mathrm{Im}\,\tau$ become arbitrarily large, it appears that one can reach an infinite amount of different fluxes and thus an infinite number of vacua. 

The crucial point, however, is that when $\mathrm{Im}\,\tau$ becomes too large, it becomes impossible to satisfy both the self-duality condition and the tadpole condition. This can be seen as follows. Since the fluxes are quantized, the quantity $|v_1|$ cannot become arbitrarily small. Therefore, as $\mathrm{Im}\,\tau$ increases, at some point one must set $v_1=0$ in order to satisfy the tadpole bound $||v||^2<L$. At this point, one is left with
\begin{equation}
    ||v||^2 = \frac{|v_2|^2}{\mathrm{Im}\,\tau}\,.
\end{equation}
It appears that $|v_2|$ can become arbitrarily large, without $||v||$ exceeding the tadpole bound. However, at this point we should recall the self-duality condition\footnote{To be precise, the analogous condition is that $v$ is imaginary anti self-dual, i.e.
\begin{equation}
    \star\,v = -i v\,.
\end{equation}}, which can be solved explicitly to give
\begin{equation}
    v_1 = \frac{v_2}{\bar{\tau}}\,. 
\end{equation}
Indeed, one immediately sees that if $v_1=0$, then the only solution to the self-duality condition is that also $v_2=0$. In other words, beyond some critical value of $\mathrm{Im}\,\tau$, the only possible vacuum is the trivial one, hence no infinite tails of vacua can occur. Furthermore, the critical value is around $\mathrm{Im}\,\tau = L$. The situation is depicted in figure \ref{fig:tadpole_torus_2}.

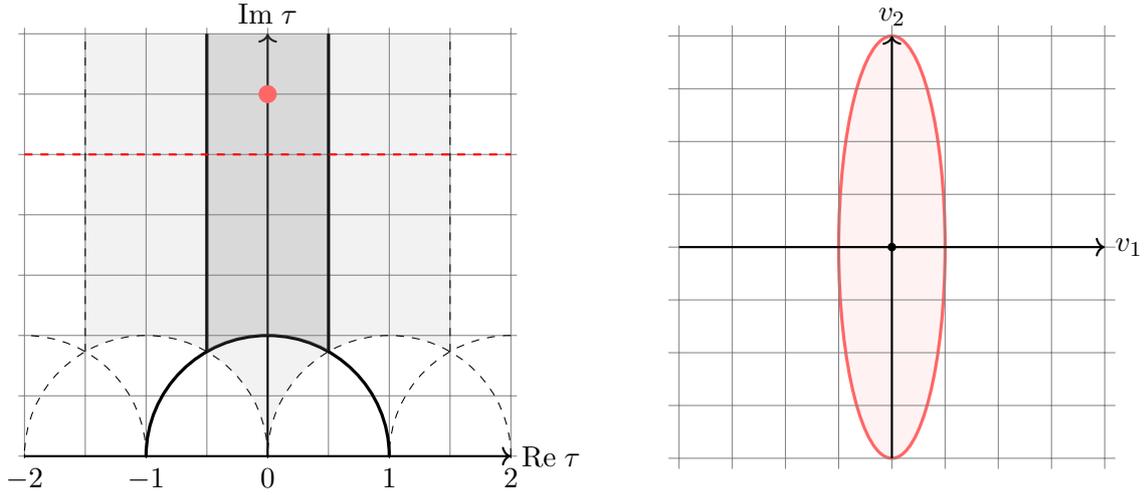
\begin{figure}[t!]
		\begin{tikzpicture}[scale=0.8]
			
			
			\draw[step=1cm, gray, very thin] (-4.1,0) grid (4.1,7.1);
			\draw[thick,->] (-4,0) -- (4,0) node[anchor =  west] {$\mathrm{Re}\;\tau$};
			\draw[thick,->] (0,0) -- (0,7) node[anchor = south ] {$\mathrm{Im}\;\tau$};   
			
			\draw[very thick] (2,0) arc (180:360: -2);
			\draw[dashed] (4,0) arc (180:360: -2);
			\draw[dashed] (0,0) arc (180:360: -2);
			\draw[dashed,thin] (-2,0) arc (180:270: -2);
			\draw[dashed,thin] (4,2) arc (270:360: -2);
			
			\draw[very thick] (1,1.73) -- (1,7);
			\draw[very thick] (-1,1.73) -- (-1,7);
			
			\draw[dashed] (3,1.73) -- (3,7);
			\draw[dashed] (-3,1.73) -- (-3,7);
			
			\fill[gray,opacity=0.3] (1,1.73) arc(60:120:2) -- (-1,7) -- (1,7) -- (1,1.73);
			
			\fill[gray,opacity=0.1] (3,1.73) arc(60:120:2) -- (1,7) -- (3,7) -- (3,1.73);
			\fill[gray,opacity=0.1] (-1,1.73) arc(60:120:2) -- (-3,7) -- (-1,7) -- (-1,1.73);
			\fill[gray,opacity=0.1] (1,1.73) arc(60:120:2) arc(60:0:2) arc(0:-60:-2);
			
			\foreach \x in {-2,-1,0,1,2}
			\draw (2*\x,1pt) -- (2*\x,-1pt) node[anchor=north]{$\x$};
			
			\filldraw[red!60] (0,6) circle (4pt);

            \draw[dashed, red, thick] (-4,5)--(4,5);
			
		\end{tikzpicture}
		\qquad        
		\raisebox{0.4cm}{
			\begin{tikzpicture}[scale=0.7]
				
				
				\draw[color=red!60, fill=red!5,very thick] (0,0) ellipse (1.0 and 4.0);
				
				\draw[step=1cm, gray, very thin] (-4.2,-4.2) grid (4.2,4.2);
				\draw[thick,->] (-4,0) -- (4,0) node[anchor=west]{$v_1$};
				\draw[thick,->] (0,-4) -- (0,4) node[anchor=south]{$v_2$};
				
				\filldraw[black] (0,0) circle (2pt);
				
			\end{tikzpicture}
		}
		\caption{The same setup is depicted as in figure \ref{fig:tadpole_torus}, but now $\mathrm{Im}\,\tau$ has exceeded the value of $L$, shown with the red dashed line. Correspondingly, there are no non-trivial self-dual fluxes beyond this point.}
		\label{fig:tadpole_torus_2}		
\end{figure}

\subsection{Finiteness theorems: global}
\label{subsec:finiteness_theorems_global}
Having discussed some general features of the problem of finiteness, let us now turn to a concrete description of the known results. This will first be done from a global point of view, meaning we focus on properties such as algebraicity and definability. We introduce the locus of Hodge classes and the locus of self-dual classes using the language of variations of Hodge structures. We briefly recall the important definitions, but refer the reader to \cite{Grimm:2018cpv,Grimm:2021ckh} for a more detailed introduction. 

\subsubsection{Hodge theory}
To state the results in full generality, we will consider the setting of an abstract variation of Hodge structure. For the convenience of the reader, we have summarized the main ingredients and their F-theory realization in table \ref{tab:F-theory}. 

\subsubsection*{Hodge structure}
As our starting point, we let $H_{\mathbb{Z}}$ be a vector space over $\mathbb{Z}$, generalizing the (primitive) flux lattice of the $G_4$ flux in F-theory. Then a Hodge structure on $H_{\mathbb{Z}}$ is a decomposition of its complexification $H_{\mathbb{C}}=H_{\mathbb{Z}}\otimes\mathbb{C}$ into $D+1$ complex subspaces
\begin{equation}\label{eq:Hodge_decomp_general}
    H_{\mathbb{C}} = H^{D,0}\oplus \cdots \oplus H^{0,D}=\bigoplus_{p+q=D}H^{p,q}\,,
\end{equation}
satisfying $H^{p,q}=\overline{H^{q,p}}$ with respect to complex conjugation. The integer $D$ is referred to as the weight of the Hodge structure. We speak of a variation of Hodge structure when the decomposition \eqref{eq:Hodge_decomp_general} varies over some parameter space $\mathcal{M}$ in a particular way which will be specified in a moment. For example, in the F-theory setting the parameter space $\mathcal{M}$ corresponds to the complex structure moduli space of the underlying Calabi--Yau fourfold. Since a variation of the complex structure changes the notion of what we call holomorphic and anti-holomorphic, this induces a variation of the decomposition \eqref{eq:Hodge_decomp_general}. 

More abstractly, one can think of a variation of Hodge structure as being defined in terms of the Hodge bundle
\begin{equation}\label{eq:Hodge_bundle}
    E\rightarrow \mathcal{M}\,,
\end{equation}
The fibres of the bundle \eqref{eq:Hodge_bundle} are the vector space $H_{\mathbb{C}}$, and the fibration encodes the variation of the $(p,q)$-decomposition of $H_{\mathbb{C}}$ as one moves in the base space $\mathcal{M}$. Locally, one may think of points in $E$ as a pair $(z^i,v)$, with $z^i\in\mathcal{M}$ and $v\in H_{\mathbb{C}}$. 

\subsubsection*{Hodge filtration}
The Hodge decomposition \eqref{eq:Hodge_decomp_general} can equivalently be expressed in terms of a so-called Hodge filtration. This is a decreasing filtration of vector spaces
\begin{equation}
    0\subseteq F^D\subseteq F^{D-1}\subseteq \cdots \subseteq F^0=H_{\mathbb{C}}\,,
\end{equation}
such that $H_{\mathbb{C}}=F^p\oplus \overline{F^{D-p+1}}$. One can pass between the two formulations by using the relations
\begin{equation}
\label{eq:decomp_filtration}
    H^{p,q} = F^p\cap\overline{F}^q\,,\qquad F^p=\bigoplus_{k=p}^D H^{k, D-k}\,.
\end{equation}
The properties of a variation of Hodge structure are neatly encoded in terms of the Hodge filtration. Indeed, given a set of local coordinates $z^i$ on $\mathcal{M}$, the filtration must satisfy the following conditions
\begin{equation}
\label{eq:horizontality}
    \frac{\partial F^p}{\partial z^i}\subseteq F^{p-1}\,,\qquad \frac{\partial F^p}{\partial\bar{z}^i}\subseteq F^p\,.
\end{equation}
The former condition implies that when taking a holomorphic derivative of a vector in $F^p$, the resulting vector ends up at most one step down in the filtration. The latter condition means that the Hodge filtration varies \textit{holomorphically} as a function of the moduli. This is in contrast to the Hodge decomposition $H^{p,q}$, for which only $H^{D,0}$ varies holomorphically while the rest of the components do not. 

\subsubsection*{Polarization}
We are interested in the case of a variation of polarized Hodge structure. This means that $H_{\mathbb{C}}$ is endowed with a $(-1)^D$-symmetric bilinear form
\begin{equation}
    (\cdot,\cdot):\,H_{\mathbb{C}}\times H_{\mathbb{C}}\rightarrow \mathbb{C}\,,
\end{equation}
satisfying the following polarization conditions with respect to the decomposition \eqref{eq:Hodge_decomp_general}
\begin{align}
    &\text{(i)}:\qquad \left(H^{p,q},H^{r,s}\right)=0\,,\qquad \text{unless $(p,q)=(s,r)$\,,}\\
    &\text{(ii)}:\quad \hspace{0.6cm}i^{p-q}\left(v,\bar{v}\right)>0\,,\qquad \text{for $v\in H^{p,q}$ and $v\neq 0$\,.}
\end{align}
We will often refer to $(\cdot,\cdot)$ as the intersection form. For future convenience, we also introduce the notation
\begin{equation}\label{eq:H_bounded}
    H_{\mathbb{Z}}(L):=\left\{v\in H_{\mathbb{Z}}\,:\,(v,v)\leq L\right\}\,,
\end{equation}
for the set of integral vectors whose self-intersection is bounded by a given real number $L$.

\subsubsection*{Symmetry group/algebra}
Let us write $G_{\mathbb{R}}$ for the real automorphism group of the pairing $\left(\cdot,\cdot\right)$, and denote its algebra by $\mathfrak{g}_{\mathbb{R}}$. This means that
\begin{align}
    &g\in G_{\mathbb{R}}:\qquad \left(gv, gw\right) = (v,w)\,,\\
    &X\in\mathfrak{g}_{\mathbb{R}}:\qquad \left(Xv, w\right)+\left(v,Xw\right)=0\,,
\end{align}
for all $v,w\in H$.

\subsubsection*{Weil operator}
Finally, we introduce the Weil operator $C\in G_{\mathbb{R}}$, defined to act on the various components of the Hodge decomposition as 
\begin{equation}
    Cv = i^{p-q}v\,,\qquad v\in H^{p,q}\,.
\end{equation}
In general, the Weil operator satisfies $C^2=(-1)^D$, hence its eigenvalues are $\pm 1$ when $D$ is even, and $\pm i$ when $D$ is odd. Correspondingly, we employ the following terminology for its eigenvectors:
\begin{itemize}
    \item \textbf{(anti) self-dual:} $Cv=\pm v$\,,
    \item \textbf{imaginary (anti) self-dual:} $Cv=\pm i v$.
\end{itemize}
The main relevance of the Weil operator is that it induces a natural inner product on $H_\mathbb{C}$ that is compatible with the Hodge decomposition. Indeed, as a result of the second polarization condition, one finds that
\begin{equation}\label{eq:Hodge_inner_product}
    \langle v,w\rangle:=\left(v, C\bar{w}\right)\,,\qquad ||v||^2:=\langle v,v\rangle\,,
\end{equation}
respectively define an inner product and a norm on $H_{\mathbb{C}}$. Furthermore, as a consequence of the first polarization condition, the Hodge decomposition \eqref{eq:Hodge_decomp_general} is orthogonal with respect to this Hodge inner product. 

\subsubsection*{Calabi--Yau fourfold realization}
For the reader who is mostly interested in the Calabi--Yau fourfold setting, which is the setting relevant for studying F-theory flux vacua, we have summarized the corresponding realization of the various Hodge-theoretic objects in table \ref{tab:F-theory}.

\begin{table}[h!]
    \centering
    \begin{tabular}{|c|c|} \hline
         \rule[-.15cm]{0cm}{.55cm}  &  F-theory setting\\ \hline
        \rule[-.2cm]{0cm}{.7cm} $H_{\mathbb{Z}}$ & \hspace*{.3cm} primitive middle cohomology $H^4_{\mathrm{prim}}\left(Y_4,\mathbb{Z}\right)$ \hspace*{.4cm} \\
        \rule[-.2cm]{0cm}{.7cm} pairing $\left(v,w\right)$ & $\int_{Y_4}v\wedge w$\\
        \rule[-.2cm]{0cm}{.7cm} Weil operator $C$ & Hodge star $\star$\\
        \rule[-.2cm]{0cm}{.7cm} Hodge inner product  $\langle v, w\rangle$ & $\int_{Y_4} v\wedge\star\,\bar{w}$\\
        \rule[-.3cm]{0cm}{.8cm} symmetry group $G_{\mathbb{R}}$ & $\mathrm{SO}\big(2+h_{\text{prim}}^{2,2}, 2h^{3,1}\big)$\\
        \hline
    \end{tabular}
    \caption{Realization of the various Hodge-theoretic objects in the Calabi--Yau fourfold setting, relevant for the study of F-theory flux compactifications.}
    \label{tab:F-theory}
\end{table} 

\subsubsection{Locus of Hodge classes}
\label{subsubsec:locus_Hodge}
Recall from the discussion in section \ref{sec:flux_compactification} that a self-dual flux vacuum is called a Hodge vacuum if the flux $G_4$ only has a (2,2)-component. In other words,
\begin{equation}
    G_4\in H^4(Y_4,\mathbb{Z})\cap H^{2,2}\,.
\end{equation}
Classes of this type are so special that they have a name: they are referred to as \textit{Hodge classes}. More generally, given a variation of Hodge structure of even weight $2k$, a Hodge class is an integral class of type $(k,k)$. In view of the tadpole condition, it is natural to consider the subset of Hodge classes whose self-intersection is bounded, for which we recall the notation \eqref{eq:H_bounded}. The set of all Hodge classes with self-intersection bounded by $L$ defines a subspace of the Hodge bundle $E$ which will be denoted by 
\begin{equation}\label{eq:Hodge-locus}
    E_{\mathrm{Hodge}}(L) = \{(z^i, v)\in E\,|\,v\in H_{\mathbb{Z}}(L)\cap H^{k,k}  \}\,,\qquad D=2k\,.
\end{equation}
We will refer to $E_{\mathrm{Hodge}}(L)$ as the \textit{locus of bounded Hodge classes}. The full locus of Hodge classes is then the countable union of $E_{\mathrm{Hodge}}(L)$ over all integers $L$ and is denoted simply by $E_{\mathrm{Hodge}}$. It is relatively easy to see that $E_{\mathrm{Hodge}}$ defines a complex-analytic subspace of $E$. There are two ways to see this:
\begin{enumerate}
    \item \textbf{Superpotential:} In the F-theory setting, a Hodge vacuum is alternatively defined by the equations $\partial_i W = W=0$, which are holomorphic in the complex structure moduli.
    \item \textbf{Hodge filtration:} More generally, it follows from the relation \eqref{eq:decomp_filtration} that
    \begin{equation}
        H_{\mathbb{Z}}\cap H^{k,k} = H_{\mathbb{Z}}\cap F^k\,.
    \end{equation}
    Note that the reality condition is crucial here. By definition of a variation of Hodge structure, the filtration $F^p$ depends holomorphically on the moduli. 
\end{enumerate}
The fact that the locus of Hodge classes is complex-analytic is already quite special, as this property is not retained for generic self-dual vacua, as will be explained later. At the same time, due to the additional condition $W=0$, the locus is defined by $h^{3,1}+1$ generically independent equations, hence one expects solutions to be relatively rare. Said differently, in order for a vacuum to exist, something special must occur in order for some of the equations to become dependent. The special thing that needs to happen is captured by the following striking theorem of Cattani, Deligne and Kaplan.

\begin{thm}[Cattani, Deligne, Kaplan \cite{CDK}]
\label{thm:CDK}
$E_{\mathrm{Hodge}}(L)$ is an algebraic variety, finite over $\mathcal{M}$.    
\end{thm}
By the phrase `finite over $\mathcal{M}$' it is meant that restriction of the projection $p:E\rightarrow\mathcal{M}$ to $E_{\mathrm{Hodge}}(L)$ has finite fibers. In other words, for each $z\in\mathcal{M}$ the fiber over $z$ consists of finitely many points. Furthermore, the algebraicity of $E_{\mathrm{Hodge}}(L)$ means that it can each be represented by a finite set of algebraic equations in $E$, i.e.~\textit{polynomials} in the moduli and the fluxes. In other words, it is of the form
\begin{equation}
    P_i(x_1,\ldots, x_k)=0\,,
\end{equation}
for some polynomials $P_i$. It should be stressed that this is truly remarkable, as the superpotential $W$ itself is typically a complicated transcendental function in the moduli. Nevertheless, the locus where $\partial_i W=W=0$ enjoys a comparatively simple description. This can be made very explicit in concrete examples, and we refer the reader to the upcoming work \cite{Grimm_vdHeisteeg_algebraicity_to_appear} where this is investigated in detail. 

For the purpose of the present work, the crucial observation is that the algebraicity of $E_{\mathrm{Hodge}}(L)$ automatically implies the finiteness of Hodge vacua. Indeed, it is clear that the zero-set of a finite collection of polynomials has only finitely many connected components. This should be contrasted with the full locus of Hodge classes $E_{\mathrm{Hodge}}$, which is only a countable union of algebraic varieties and hence does not have such a finiteness property.\footnote{See however \cite{baldi2022distribution} for recent refinements of this statement.} In this regard, it is interesting to point out that when the variation of Hodge structure under consideration comes from a family of smooth projective varieties, the same conclusion follows from the famous Hodge conjecture. However, the Hodge conjecture does not predict the stronger statement that $E_{\mathrm{Hodge}}(L)$ is algebraic. In other words, it does not predict the finiteness of Hodge vacua. It is therefore rather curious that the string-theoretic setting imposes the additional crucial constraint, namely the tadpole condition, to exactly ensure finiteness.

For the interested reader, let us give a very rough idea of how one would approach a proof Theorem \ref{thm:CDK}, following the original work of Cattani, Deligne, and Kaplan. In particular, we focus on how one would reduce this to a local statement, which will then be discussed in more detail in section \ref{subsec:finiteness_theorems_local} and appendix \ref{app:Hodge_locus}. The reduction is performed by employing a comparison theorem which connects algebraic geometry and analytic geometry known as Chow's theorem, which states that any closed analytic subspace of a complex projective space is algebraic.\footnote{This now falls within the broader domain of so-called GAGA results, which encompasses various types of comparison results between algebraic and analytic geometry in terms of comparisons of categories of sheaves. Here GAGA stands for \textit{G\'eometrie Alg\'ebrique et G\'eom\'etrie Analytique}.} Very roughly, this means that if some closed analytic subspace is well-behaved enough in the asymptotics, then it is in fact algebraic. Indeed, we have seen that the Hodge locus is complex-analytic on $\mathcal{M}$. Furthermore, it is well-known that $\mathcal{M}$ is quasi-projective, so that its closure can be embedded in a complex projective space \cite{Viehweg}. The strategy, then, is to show that the closure of the Hodge locus in $\overline{\mathcal{M}}$ is analytic as well and to then apply Chow's theorem to establish the desired algebraicity. Hence, one reduces the question to a study of the Hodge locus locally at the divisor $\overline{\mathcal{M}}\setminus \mathcal{M}$, which brings one into the realm of degenerations of Hodge structures and asymptotic Hodge theory. Physically, this means one is studying the structure of Hodge vacua as one approaches the boundary of the moduli space, which, following our initial discussion in section \ref{subsec:finiteness_intro}, is exactly the question we are interested in.

Finally, let us mention a generalization of Theorem \ref{thm:CDK} by Schnell, who introduced the ``extended locus of Hodge classes'' \cite{schnell2014extended}. The rough goal was construct a natural compactification of the Hodge locus to also incorporate so-called ``limit Hodge classes''. These are, as the name suggests, integral classes that become Hodge in an appropriate limit and should therefore lie on the boundary of the Hodge locus. 

\subsubsection{Locus of self-dual classes}
As soon as one moves towards generic self-dual flux vacua, the situation becomes more complicated. Indeed, since the $G_4$ flux is now allowed to have also $(4,0)$ and $(0,4)$ components, it no longer corresponds to a Hodge class. In a similar fashion as before, let us denote by
\begin{equation}
    E_{\text{self-dual}}(L) = \{(z^i,v)\in E\,:\,v\in H_{\mathbb{Z}}(L),\,C(z)v = v \}\,,
\end{equation}
the set of all integral self-dual fluxes with a bounded self-intersection. We will refer to $E_{\text{self-dual}}$ as the \textit{locus of bounded self-dual classes}. In contrast to the locus of Hodge classes, the locus of self-dual classes is a priori only a \textit{real}-analytic subspace of $E$. Again, one can see this by noting that a generic self-dual vacuum is defined by the equation $D_iW_{\mathrm{flux}}=0$, which now involves the real K\"ahler potential $K$. Nevertheless, in analogy with the algebraicity of the locus of bounded Hodge classes, it was shown in \cite{Bakker:2021uqw} that the locus of bounded self-dual classes has a lot more structure than one might at first expect, as captured in the following
\begin{thm}[Bakker, Grimm, Schnell, Tsimerman \cite{Bakker:2021uqw}]
\label{thm:finiteness_self_dual}
    The set $E_{\rm {self\text{-}dual}}(L)$ is a definable in the o-minimal structure $\mathbb{R}_{\mathrm{an,exp}}$. Furthermore, it is a closed, real-analytic subspace of $E$, finite over $\mathcal{M}$.
\end{thm}
Let us briefly elaborate on the phrase `definable in the o-minimal structure $\mathbb{R}_{\mathrm{an,exp}}$'. For a more detailed explanation we refer the reader to \cite{Grimm:2021vpn}. Roughly, this means that the locus of bounded self-dual classes can be described by a finite set of polynomial equations and inequalities that involve not only the moduli and fluxes, but also any restricted analytic function and real exponential function of the moduli. More precisely, the o-minimal structure $\mathbb{R}_{\mathrm{an,exp}}$ is generated (through finite products, unions, intersections and projections) by sets of the form
\begin{align}
    P(x_1,\ldots, x_k, f_1,\ldots, f_m, e^{x_1},\ldots, e^{x_k})=0\,,
\end{align}
where the $f_i$ are restricted analytic functions and $P$ is a polynomial.  

Importantly for our purposes, the fact that the locus of bounded self-dual classes is definable implies that it also has an inherent finiteness property, which we now explain briefly in three steps. 
\begin{enumerate}
    \item The fact that the restriction of $p:E\rightarrow\mathcal{M}$ to $E_{\mathrm{self-dual}}(L)$ has finite fibers means that for each point $z\in\mathcal{M}$, its preimage under this map consists of a finite number of points. In other words, for fixed $z$ the size of the fiber $p^{-1}(z)$ is bounded. This is, of course, not enough to prove finiteness completely, since $z$ itself ranges over an infinite set.
    \item Due to the special properties of definable functions, one can show that in fact the size of the fiber is uniformly bounded. Hence, there exists an integer $N_{\mathrm{max}}$ such that
    \begin{equation}
        |p^{-1}(z)|\leq N_{\mathrm{max}}\,,
    \end{equation}
    for all $z\in\mathcal{M}$.
    \item Finally, one can show that since the map $p$ is itself definable, the set
    \begin{equation}
        \{z\in\mathcal{M}: |p^{-1}(z)|\leq N_{\mathrm{max}}\}\,,
    \end{equation}
    is definable as well. In particular, it cannot contain infinitely many discrete points. 
\end{enumerate}
As a final remark, let us mention that Theorem \ref{thm:finiteness_self_dual} actually implies Theorem \ref{thm:CDK}, namely that the locus of bounded Hodge classes is algebraic. This was shown in \cite{BKT} by Bakker, Klingler and Tsimerman using the so-called definable Chow theorem of Peterzil and Starchenko \cite{Peterzil:2009}. The latter is an alternative version of Chow's theorem adapted to the setting of o-minimal geometry and roughly states that a complex-analytic set which is also definable is in fact algebraic. Recalling that the locus of Hodge classes is clearly complex-analytic, one recovers Theorem \ref{thm:CDK}.

\subsection{Finiteness theorems: local}
\label{subsec:finiteness_theorems_local}
In this section we discuss some local manifestations of the finiteness theorems presented in section \ref{subsec:finiteness_theorems_global}. Arguably, when it comes to developing further intuition for the finiteness of vacua, the local analysis is more illuminating. Indeed, in section \ref{subsec:finiteness_intro} it was argued that, as far as finiteness is concerned, the main question is whether it is possible for vacua to accumulate near the boundaries of the moduli space. Furthermore, in section \ref{subsubsec:locus_Hodge} we gave a rough idea of how the proof of the theorem of Cattani, Deligne, and Kaplan heavily relies on a local analysis near the boundaries of the moduli space. This brings us into the realm of asymptotic Hodge theory.

\subsubsection{Asymptotic Hodge theory (1)}
\label{subsubsec:asymp_Hodge_theory_1}
Since we are interested in a local description of $\mathcal{M}$ in the near-boundary regime, we may assume that $\mathcal{M}$ is given by the direct product of $r$ punctured disks $\Delta^*$ and $m-r$ disks $\Delta$, where $m$ denotes the complex dimension of $\mathcal{M}$ and $r$ denotes the number of coordinates that approach the boundary. Without loss of generality, we may and will assume that $m=r$. We choose local coordinates $z^i$ on the punctured disks such that the punctures are located at $z^i=0$, corresponding to the locations of singular divisors in the moduli space. Furthermore, we denote by 
\begin{equation}
    t^i = \frac{1}{2\pi i}\log z^i\,,
\end{equation}
the corresponding coordinates on the universal covering space of $(\Delta^*)^m$. The $t^i$ coordinates each take value in the complex upper half-plane $\mathbb{H}$, and the singularities are located at $\mathrm{Im}\,t^i\rightarrow\infty$. In the following, we will decompose $t^i$ into its real and imaginary parts as
\begin{equation}
    t^i = x^i + i y^i\,,
\end{equation}
with $x^i$ and $y^i$ corresponding to the axions and saxions, respectively. Note that, due to the periodic nature of the axionic coordinates, a fundamental domain of the $x^ i$ is the bounded interval $[0,1]$. The two descriptions of the near-boundary regime of $\mathcal{M}$ are illustrated in figure \ref{fig:disc}.

\begin{figure}[t!]
	\begin{center}
		\includegraphics[width=0.8\textwidth]{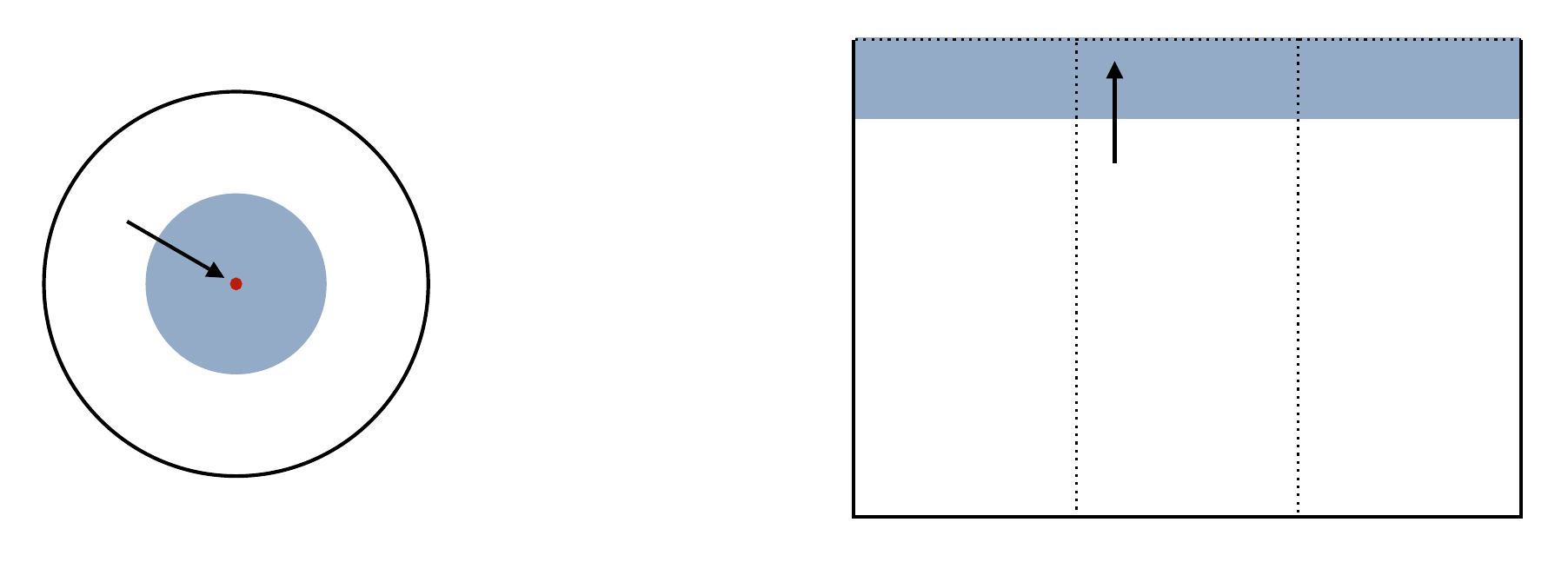}
		\vspace*{-1cm}
	\end{center}
	\begin{picture}(0,0)
	\put(70,82){\rotatebox{-31}{\small $z \rightarrow 0$}}
	\put(307,88){\small $y \rightarrow \infty$}
	\put(10,110){(a)}
	\put(200,110){(b)}
	\end{picture}
	
	\caption{Two local descriptions of the near-boundary regime of the moduli space $\mathcal{M}$. Figure (a): Poincar\'e disc with the singularity located at $z=0$. Figure (b): upper half-plane, with the singularity located at $y\rightarrow\infty$.\label{fig:disc}}
	
\end{figure}
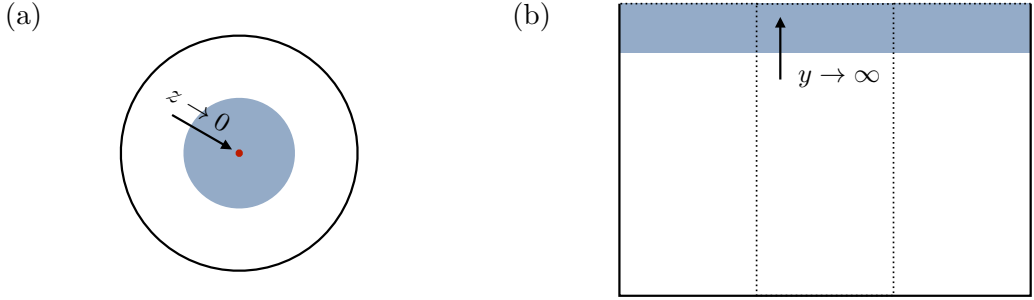

\subsubsection*{Monodromy}
Of vital importance is the local monodromy behaviour of the variation of Hodge structure when encircling the singularity. This is obtained by sending $z^i\mapsto z^i e^{2\pi i }$ or equivalently $t^i\mapsto t^i+1$ and asking how the Hodge filtration transforms under this map. There are in total $m$ monodromy operators $T_i\in G_{\mathbb{R}}$, which act on the Hodge filtration as
\begin{equation}
\label{eq:Hodge_monodromy}
    T_i F^p(t^i) = F^p(t^i+1)\,.
\end{equation}
To be precise, by the action of $T_i$ on a given filtration we simply mean the action of $T_i$, as a matrix, on the vectors that span that filtration. After an appropriate coordinate redefinition, the monodromy operators may be taken to be unipotent and of the form
\begin{equation}
\label{eq:log_monodromy}
    T_i = e^{N_i}\,,\qquad [N_i,N_j]=0\,,
\end{equation}
where $N_i\in\mathfrak{g}_{\mathbb{R}}$ are commuting nilpotent operators, whose nilpotency degree lies between $0$ and the weight $D$ of the Hodge structure. The log-monodromy matrices $N_i$ play a central role in the study of asymptotic Hodge theory, as we explain in the following.

\subsubsection*{Nilpotent orbit approximation}
From the preceding discussions of (asymptotic) Hodge theory, we would like to highlight two important features of the Hodge filtration $F^p$. Namely, (1) it is holomorphic, recall equation \eqref{eq:horizontality}, and (2) it undergoes a monodromy transformation when encircling a singularity in the moduli space, recall equation \eqref{eq:Hodge_monodromy}. Intuitively, the simplest types of Hodge filtrations that exhibit these features are of the form
\begin{equation}
\label{eq:nilpotent_orbit}
    F_{\mathrm{nil}}^p = e^{t^i N_i}F_0^p\,,
\end{equation}
where $F_0^p$ is some moduli-independent filtration.\footnote{Of course, there are conditions that should be placed on $F_0^p$ to ensure that $F_{\mathrm{nil}}^p$ is a proper polarized variation of Hodge structure. Notably, the first condition in \eqref{eq:horizontality} restricts how the log-monodromy matrices $N_i$ can act on $F_0^p$. However, it should be stressed that generically $F_0^p$ itself does not constitute a polarized Hodge filtration.} Hodge filtrations of the form \eqref{eq:nilpotent_orbit} are referred to as ``nilpotent orbits'', since they correspond  to the orbit of some fixed filtration under the action of the nilpotent operators $N_i$. One of the striking results of asymptotic Hodge theory, due to Schmid \cite{schmid}, is that \textit{any} polarized variation of Hodge structure asymptotes to a nilpotent orbit as one approaches a singularity in the moduli space. In other words, in the regime  where some $\mathrm{Im}\,t^i\gg 1$, for $i=1,\ldots, r$, one has
\begin{equation}
\label{eq:nilpotent_orbit_theorem}
    F^p(t,\zeta)\approx F_{\mathrm{nil}}^p(t,\zeta) = e^{t^i N_i} F_0^p(\zeta)\,,
\end{equation}
with the corrections being exponentially small in $\mathrm{Im}\,t^i$.\footnote{The more precise statement is that, in terms of a natural notion of distance $d(\cdot,\cdot)$ on the space of all polarized Hodge filtrations, one has
\begin{equation}
    d\left(F, F_{\mathrm{nil}}\right)\leq K \sum_{j=1}^r \left(y^i\right)^\beta e^{-2\pi y^i}\,,\qquad y^i\gg 1\,,
\end{equation}
for some constants $K,\beta$. In other words, in the regime $y^i\gg 1$ the two filtrations are exponentially close in this distance.} Here $\zeta$ denotes those remaining $m-r$ moduli which are not sent to the boundary, sometimes referred to as ``spectator moduli''. As mentioned earlier, we will assume without loss of generality that $m=r$ and will therefore ignore such spectator moduli.

The result \eqref{eq:nilpotent_orbit_theorem}, known as the nilpotent orbit theorem, is an incredibly powerful tool to study the properties of general variations of Hodge structure. For example, one might first attempt to prove a given statement for the case that the variation of Hodge structure in question is described exactly by a nilpotent orbit. Then, one may study whether the result survives upon the inclusion of exponential corrections. This is exactly the strategy that is employed in some of the mentioned finiteness proofs. Indeed, one may first study the self-duality condition for the fluxes using the approximate Weil operator $C_{\mathrm{nil}}$ associated to $F_{\mathrm{nil}}$, as will be demonstrated in section \ref{sec:self_dual_locus}. Importantly, using the second main result of asymptotic Hodge theory, the $\mathrm{Sl}(2)$-orbit theorem, it is possible to characterize $C_{\mathrm{nil}}$ in complete generality. This will be explained in detail in section \ref{subsec:nilpotent_orbit_expansion}. 

\subsubsection{Finiteness of Hodge classes}

We can now formulate local versions of the finiteness theorems discussed in section \ref{subsec:finiteness_theorems_global}. In this section, we focus on the case of Hodge vacua. Our goal is to consider a sequence of such vacua that approaches the boundary of $\mathcal{M}$ and ask whether this sequence can take on infinitely many values. To this end, we state the following

\begin{thm}[{{\cite[Theorem 3.3]{CDK}}}]
\label{thm:finiteness_Hodge_loci_local}
    Let $t^i(n)\in\mathbb{H}^m$
    be a sequence of points such that $x^i(n)$ is bounded and $y^i(n)\rightarrow\infty$ as $n\rightarrow\infty$. Suppose furthermore that
    \begin{equation*}
        v(n)\in  H_{\mathbb{Z}}(L)\cap H^{k,k}\,,\qquad D=2k\,,
    \end{equation*}
    is a sequence of integral bounded Hodge classes. Then $v(n)$ can only take on finitely many values. 
\end{thm}
Here we stress that the Hodge decomposition $H^{k,k}$ is itself a function of the moduli. However, in order not to clutter the notation we will often omit this dependence. The upshot of Theorem \ref{thm:finiteness_Hodge_loci_local} is that it is indeed impossible to have an accumulation of Hodge vacua near the boundary of $\mathcal{M}$. In appendix \ref{app:Hodge_locus} we will describe the proof of Theorem \ref{thm:finiteness_Hodge_loci_local} in some detail. 

\subsubsection{Finiteness of self-dual classes}
\label{subsubsec:finiteness_self-dual_classes}
Finally, let us come to the finiteness of self-dual vacua. In contrast to Theorem \ref{thm:finiteness_Hodge_loci_local}, there has not yet appeared a fully general directly local proof for the finiteness of self-dual flux vacua. Nevertheless, the following statement clearly follows as a corollary of the global statement given in Theorem \ref{thm:finiteness_self_dual}.
\begin{cor}
\label{cor:finiteness_self-dual_local}
    Let $t^i(n)\in\mathbb{H}^m$
    be a sequence of points such that $x^i(n)$ is bounded and $y^i(n)\rightarrow\infty$ as $n\rightarrow\infty$. Suppose furthermore that $v(n)\in H_{\mathbb{Z}}(L)$ is a sequence of integral fluxes with bounded self-intersection, such that
    \begin{equation}
        C(t(n))v(n) = v(n)\,,
    \end{equation}
    for all $n$. Then $v(n)$ can only take on finitely many values. 
\end{cor}
An independent proof of Corollary \ref{cor:finiteness_self-dual_local} was given in \cite{Grimm:2020cda,Schnellletter} for the case of a single variable using methods from asymptotic Hodge theory. In section \ref{sec:self_dual_locus} we will extend these methods to the multi-variable setting in order to give some intuition for the finiteness of self-dual vacua in the general case, without using results from o-minimality. To be precise, we will provide a proof within the nilpotent orbit approximation. To be absolutely clear, we will prove the following
\begin{thm}
\label{thm:finiteness_selfdual_nilpotent}
    Let $t^i(n)\in\mathbb{H}^m$
    be a sequence of points such that $x^i(n)$ is bounded and $y^i(n)\rightarrow\infty$ as $n\rightarrow\infty$. Suppose furthermore that $v(n)\in H_{\mathbb{Z}}(L)$ is a sequence of integral fluxes with bounded self-intersection, such that
    \begin{equation}
        C_{\mathrm{nil}}(t(n))v(n) = v(n)\,,
    \end{equation}
    for all $n$. Then $v(n)$ can only take on finitely many values. 
\end{thm}
In particular, note the replacement of the general Weil operator $C$ by its nilpotent orbit approximation $C_{\mathrm{nil}}$. Of course, this will, therefore, not quite constitute a full independent proof of Corollary \ref{cor:finiteness_self-dual_local}. Nevertheless, the discussion will provide some valuable intuition for the asymptotic behaviour of vacua and will include some new insights into the asymptotic form of the Weil operator and generic Hodge inner products, which may be of independent interest for some readers.

\subsection{Summary}
We close this section by providing the reader with an overview of the various theorems we have discussed, see figure \ref{fig:theorems_overview}. Let us also highlight the variety of strategies that are employed in the proofs of these various theorems. For Hodge vacua, both in the single-variable and multi-variable case, the proof relies heavily on the machinery of mixed Hodge structures, as is explained in appendix \ref{app:Hodge_locus}. Instead, our analysis of the self-dual vacua in the nilpotent orbit approximation makes use of the asymptotic expansion of the Weil operator, as is described in sections \ref{sec:asymp_Hodge_inner_products} and \ref{sec:self_dual_locus}. Finally, for the general proof of the finiteness of self-dual flux vacua the recent advances in o-minimal geometry have played an essential role. 

\begin{figure}[h!]
    \centering
    
\scalebox{0.9}{\begin{tikzpicture}

\node[draw, text width=6cm, align=center] at (0,0) {self-dual,\\ multi-variable \cite{Bakker:2021uqw}};

\draw[-implies,double equal sign distance] (3.5,0) -- (5,0);

\draw[-implies,double equal sign distance] (3.5,-5) -- (5,-5);

\draw[-implies,double equal sign distance] (0,-1) -- (0,-4);

\draw[-implies,double equal sign distance] (8.5,-3) -- (8.5,-4);

\draw[-implies,double equal sign distance] (3.5,-0.5) -- (5,-2);

\node[draw, text width=6cm, align=center] at (0,-5) {Hodge,\\ multi-variable \cite{CDK}\\ appendix \ref{subsubsec:proof_Hodge_locus_multi}};

\node[draw, text width=5cm,align=center] at (8.5,0) {self-dual,\\ multi-variable\\ (nilpotent orbit approx.)\\ section \ref{sec:self_dual_locus}};

\node[draw, text width=5cm,align=center] at (8.5,-2) {self-dual,\\ single-variable \cite{Grimm:2020cda,Schnellletter}};

\node[draw, text width=5cm,align=center] at (8.5,-5) {Hodge,\\ single-variable \cite{CDK}\\ appendix \ref{subsubsec:proof_Hodge_locus_single}};

\node[] at (4,2) {\Large Web of Finiteness Theorems};
 
\end{tikzpicture}}
    \caption{An overview of the various finiteness theorems discussed in this work, including the implications between them. }
    \label{fig:theorems_overview}
\end{figure}
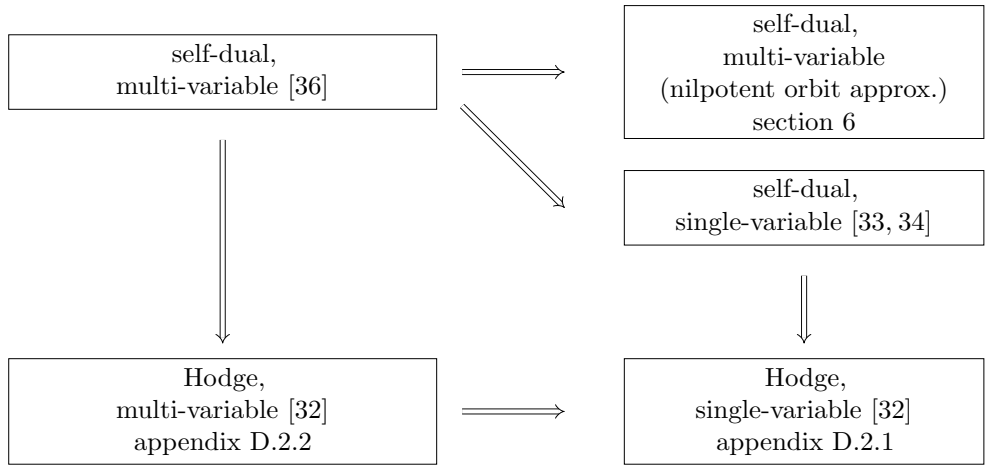

\section{Conjectures about the Flux Landscape}
\label{sec:future_questions}

In the preceding sections we have focused our attention on relatively rudimentary properties of the flux landscape, in particular with regards to its finiteness. In this section we would like to point out some additional questions that could feasibly be addressed in the near-future, whose answers would further elucidate more precise features of the flux landscape, and formulate them into precise mathematical conjectures. These conjectures would pose interesting challenges which can likely be tackled by the application and development of techniques in asymptotic Hodge theory and o-minimality. 

\subsection{Recounting flux vacua}

Having established that the number of self-dual flux vacua is finite, a natural follow-up question would be: how many are there? The early works of Douglas et al.~\cite{Ashok:2003gk,Denef:2004ze} suggest that such numbers could be very large, giving rough estimates of the order $10^{500}$ to $10^{272,000}$, see also \cite{Taylor:2015xtz}. At the same time, it has also been pointed out that these analyses have their shortcomings. In particular, it is possible that the smearing approximation used to effectively ignore the quantization condition significantly affects the precise counting of vacua. It is a challenging task to establish robust mathematical counting results. 

One might ask if this problem becomes attainable for the case of Hodge vacua. Here one faces the fact that the approximations of \cite{Ashok:2003gk,Denef:2004ze} are likely even less reliable. 
As discussed also in section \ref{subsubsec:locus_Hodge}, a Hodge vacuum is expected to be relatively rare. The main reason for this is the fact that a Hodge vacuum has to satisfy $h^{3,1}+1$ equations for only $h^{3,1}$ variables, hence the system is overdetermined. Importantly, after solving the $D_i W_{\mathrm{flux}}=0$ equations for the complex structure moduli in terms of the fluxes and inserting the result into the remaining $W_{\mathrm{flux}}=0$ equation, one is left with a highly transcendental equation for the fluxes. This transcendentality originates from the fact that the flux-induced superpotential is expressed in terms of the periods of the Calabi--Yau fourfold. The crucial point is that, due to the quantization condition, this highly transcendental equation needs to be solved over the integers, hence its solutions are expected to be rare. Indeed, in the context of o-minimal geometry, some intuition for this is provided by the celebrated counting theorem of Pila and Wilkie \cite{Pila:2006}. Very roughly speaking, the Pila--Wilkie theorem states that there are very few rational points on the transcendental part of a definable set. More precisely, the number of such points grows slower than any positive power of their multiplicative height.\footnote{For an integral flux $v=(v_1,\ldots, v_k)\in H_{\mathbb{Z}}$, its multiplicative height is simply $\mathrm{max} |v_i|$.}\textsuperscript{,}\footnote{In \cite{barroero2013counting} this theorem was applied to provide bounds on the number of lattice points in the fibers of definable families.}

In an earlier version of this work it was conjectured, based on the above considerations, that for those variations of polarized Hodge structure which are ``sufficiently transcendental'' (dictated by a property called the ``level'' \cite{baldi2022distribution}), the number of connected components in the locus of Hodge classes with a fixed self-intersection $L$ should grow sub-polynomially in $L$. In particular, this would imply that the number of $W_{\mathrm{flux}}=0$ vacua in F-theory grows much slower than expected. However, it was recently shown in an explicit example investigated in \cite{Grimm_vdHeisteeg_algebraicity_to_appear} that this conclusion is not quite correct, but for a very interesting reason. Namely, it can happen that a collection of Hodge classes actually lies on a higher-dimensional locus where additional Hodge \textit{tensors} appear, see appendix \ref{app:Hodge-tensors} for a basic introduction to Hodge tensors. The important point is that, because of the presence of these additional Hodge tensors, the restriction of the variation of Hodge structure to this higher-dimensional locus typically has a reduced level and thus becomes ``less transcendental'', such that the original logic based on the Pila--Wilkie counting theorem may not apply. As discussed in \cite{Grimm_vdHeisteeg_algebraicity_to_appear} this reduction in transcendentality on these loci indicates the presence of an underlying symmetry in the compactification manifold. In order to take into account these subtle matters, we therefore propose the following refined version of the counting conjecture.

\begin{conjecture}
    \label{conjecture:scaling_Hodge_vacua}
    Consider a variation of polarized Hodge structure $E\rightarrow\mathcal{M}$ of even weight $D=2k$. Fix a positive integer $L$ and consider the locus of Hodge classes with a fixed self-intersection $L$,
    \begin{equation}
        \hat{E}_{\mathrm{Hodge}}(L) = \{(z^i,v)\in E: v\in H^{k,k}\cap H_{\mathbb{Z}}\,,(v,v)= L\}\,.
    \end{equation}
  Furthermore, denote by $\hat{E}_{\mathrm{Hodge}}^{\mathrm{iso}}(L)$ the subset of points $(z^i,v)\in\hat{E}_{\mathrm{Hodge}}(L)$ for which $z^i$ are isolated points in the locus of Hodge tensors, see appendix \ref{app:Hodge-tensors}.
  
  We claim that if the level of the variation of Hodge structure is at least 3, then the number points in $\hat{E}^{\mathrm{iso}}_{\mathrm{Hodge}}(L)$ grows sub-polynomially in $L$. More precisely, for every $\epsilon>0$ there exists a $C>0$, such that
    \beq
       \# \hat{E}^{\mathrm{iso}}_{\mathrm{Hodge}}(L) < C L^\epsilon\ ,  
    \eeq
    where $\# \hat{E}^{\mathrm{iso}}_{\mathrm{Hodge}}(L)$ is the number points in $\hat{E}^{\mathrm{iso}}_{\mathrm{Hodge}}(L)$ and $C$ is independent of $L$.
\end{conjecture}
Some remarks are in order. First, we note that in \cite{baldi2022distribution} a related conjecture has been proposed. The latter states that, under similar conditions, the number of points in $\hat{E}^{\mathrm{iso}}_{\mathrm{Hodge}}$ is in fact finite, \textit{without fixing the self-intersection}. It is important to stress that while a similar statement for higher-dimensional loci has, rather strikingly, been proven in \cite{baldi2022distribution}, the statement for isolated points, which is the case of interest for us, is still a wide open problem. 

Second, let us briefly elaborate on the notion of the `level' of a variation of Hodge structure. The precise definition is somewhat technical and is explained in \cite{baldi2022distribution}. Roughly speaking, it is related to the length of the Hodge filtration and serves as a measure of its `complexity'. However, it should not be confused with the weight $D$ of the Hodge structure. For example, while the Hodge structure on the middle cohomology of a K3 surface is of weight $D=2$, its level is in fact equal to one. As another example, while one generically expects that the middle cohomology of a Calabi--Yau fourfold has level equal to four, one can show that for special cases such as $Y_3\times T^2$ or $\mathrm{K3}\times \mathrm{K3}$ the level is again equal to one. In particular, Conjecture \ref{conjecture:scaling_Hodge_vacua} does not apply to these cases. 

To elaborate on this point, consider the weak-coupling limit corresponding to type IIB orientifold compactifications, in which case one effectively reduces to a direct product $Y_4=Y_3\times T^2$ and hence the level reduces to one. In this setting, known scans of vacua in one-parameter and two-parameter Calabi--Yau manifolds, defined as hypersurfaces in weighted projective space, indicate that the number of vacua with $W_{\mathrm{flux}}=0$ in fact scales polynomially in $L$ \cite{DeWolfe:2004ns,Giryavets:2004zr,Conlon:2004ds}. This is confirmed by the recent work \cite{Plauschinn:2023hjw} in which a complete counting of vacua, including $W_{\mathrm{flux}}=0$ vacua, was performed for the mirror octic. To be clear, this is not in contradiction with Conjecture \ref{conjecture:scaling_Hodge_vacua}, due to the reduction in the level in the weak coupling limit. We believe, however, that this counting is actually not representative for the number of exact Hodge vacua in the non-perturbative setting of F-theory. Indeed, the observed polynomial scaling in the type IIB setting should be viewed as an artifact of truncating the axio-dilaton dependence to the polynomial, i.e.~algebraic, level. To emphasize this point, recall that the axio-dilaton $\tau$ can trivially be solved for in terms of the $F_3$ and $H_3$ fluxes as
\begin{equation}
\label{eq:solution_tau_IIB}
    \bar{\tau} = \frac{\int \Omega\wedge F_3}{\int \Omega\wedge H_3}\,.
\end{equation}
In contrast, as soon as one includes exponential corrections in $\tau$ it is clear that this is no longer so straightforward and we expect that the transcendental nature of the equations greatly restricts the number exact Hodge vacua.\footnote{Of course, there can also be perturbative corrections which break the simple relation \eqref{eq:solution_tau_IIB}, but these do not affect the transcendentality of the equations.} Put shortly, one should perform the counting of $W_{\mathrm{flux}}=0$ vacua in the full F-theory setting, which, in particular, requires a non-trivial elliptic fibration. Mathematically, this is captured by the condition that the level of the variation of Hodge structure should be at least three. A further motivation for this comes from the recent work \cite{baldi2022distribution}, in which it was shown that, when the level is at least three, the locus of Hodge classes corresponds to an atypical intersection, reflecting the fact that it is expected to occur only rarely.

Finally, let us mention some recent developments in mathematics concerning the issues of algebraicity and transcendentality in a Hodge-theoretic context. From a more number-theoretic point of view, a Hodge vacuum effectively requires that some of the $h^{3,1}+1$ equations are no longer algebraically independent. It is a long-standing question when there exist algebraic relations among transcendental numbers, which lies at the heart of the Schanuel conjecture. More concretely, given a collection of complex numbers $\alpha_1,\ldots,\alpha_n$ which are algebraically independent over $\mathbb{Q}$, the Schanuel conjecture gives a bound on the number of algebraic relations among the numbers $\alpha_1,\ldots, \alpha_n, e^{\alpha_1},\ldots, e^{\alpha_n}$. A functional analogue of this question, where one is considering algebraic relations between $f_1(x),\ldots, f_n(x), e^{f_1(x)},\ldots, e^{f_n(x)}$, is addressed by the Ax--Schanuel theorem \cite{Ax:1971}, which has also been generalized for certain transcendental functions besides the exponential function. Recently, techniques from o-minimal geometry and the theory of atypical/unlikely intersections have lead to great developments in this field as well as a proof of the Ax--Schanuel conjecture in the Hodge-theoretic setting \cite{klingler2017hodge,bakker2017axschanuel}. Very roughly speaking, the latter relates the appearance of an atypical intersection, meaning the existence of additional algebraic relations among e.g.~the periods, to a reduction of the so-called Mumford--Tate group. In a similar spirit, the recent work \cite{baldi2022distribution} has elucidated further properties of the Hodge locus using the theory of unlikely intersections. It would be very interesting to further investigate these techniques in the context of F-theory flux compactifications and ascertain whether they could lead to improved quantitative results on the counting of Hodge vacua and possibly prove or disprove Conjecture \ref{conjecture:scaling_Hodge_vacua}. Whether these techniques could also be applied to study self-dual vacua is not so clear.     

\subsection{Complexity of the flux landscape}

Another exciting avenue to explore with regards to the counting of flux vacua is using a certain notion of complexity that has recently been developed in the context of sharp o-minimality, which moreover may be applicable to study both Hodge vacua and self-dual vacua. The basic idea of sharp o-minimality, introduced by Binyamini and Novikov \cite{binyamini2022,binyamini2022sharply}, is to endow definable sets, and thereby definable functions, with some additional positive integers $(F,D)$, called the ``format'' $F$ and ``degree'' $D$ , that reflect the inherent geometric complexity of that set/function. This is in analogy with the degree of a polynomial, which clearly gives the number of zeroes of said polynomial over the complex numbers, but can also be used to give bounds on the number of its zeroes over the real numbers.\footnote{More generally, this falls under Khovanskii's theory of fewnomials \cite{Khovanskii:1980}.} Roughly speaking sharply o-minimal structure are defined in such a way that the functions arising in these structures have similar bounds on their number of zeros \cite{binyamini2022sharply}. Recently, the concept of sharp o-minimality has been explored in a variety of quantum mechanical systems in order to assign a well-defined notion of complexity to various physical observables \cite{Grimm:2023xqy}, see also \cite{Grimm:2021vpn,Douglas:2022ynw, Douglas:2023fcg}. It is natural to ask if a similar strategy can be applied to assign a complexity to e.g.~the F-theory flux scalar potential, which may then provide a new method of estimating the number of flux vacua. In this regard, we propose the following
\begin{conjecture}
    \label{conjecture:complexity}
    We conjecture that the locus of self-dual flux vacua is definable in a sharply o-minimal structure. Furthermore, we expect that its associated sharp complexity $(F,D)$ depends on the tadpole bound $L$ and the number of moduli $h^{3,1}$ in the following way: 
    \begin{equation}
        D=\mathrm{poly}(L)\,,\qquad F=\mathcal{O}(h^{3,1})\,.
    \end{equation}
\end{conjecture}
Our expectation for the scaling of $D$ and $F$ is rather conservative, and is motivated by the form of the Ashok--Douglas index density \cite{Ashok:2003gk,Denef:2004ze}. Indeed, the latter grows as $L^{h^{3,1}}$, while generically the number of zeroes of functions that are definable in a sharply o-minimal structure depends polynomially on $D$ and exponentially on $F$. Since the sharp complexity $(F,D)$ only gives upper bounds on the number of such zeroes, it could also be the case that already for self-dual vacua, the scaling is in fact more restricted. Certainly, this is expected for the special class of Hodge vacua, as captured by Conjecture \ref{conjecture:scaling_Hodge_vacua}. 

Nevertheless, we stress that the statement of Conjecture \ref{conjecture:complexity} is highly non-trivial. Indeed, while Theorem \ref{thm:finiteness_self_dual} establishes that the locus of self-dual flux vacua is definable in the o-minimal structure $\mathbb{R}_{\mathrm{an},\mathrm{exp}}$, it has been shown that this structure is not sharply o-minimal. Roughly speaking, a generic restricted analytic function does not have a well-defined notion of complexity, because one has too much freedom in specifying the coefficients in its series expansion. Nevertheless, it is currently conjectured \cite{binyamini2022}, that period integrals are in fact definable in a sharply o-minimal structure, meaning that they actually live in a much smaller o-minimal structure than $\mathbb{R}_{\mathrm{an},\mathrm{exp}}$. This would, in particular, imply a positive answer to the first part of Conjecture \ref{conjecture:complexity}. Lastly, let us mention the recent work \cite{binyamini2022wilkies} in which a proof was given for Wilkie's conjecture \cite{Pila:2006} when restricting to certain sharply o-minimal structures. Together with Conjecture \ref{conjecture:complexity}, the latter suggests that the scaling in Conjecture \ref{conjecture:scaling_Hodge_vacua} may be even more restricted by replacing the sub-polynomial scaling with a logarithmic scaling in $L$. It would be very interesting to investigate this further. 

\subsection{A generalized tadpole conjecture for the Hodge locus}

In the previous points we have focused on counting the number of flux vacua or, more precisely, the number of connected components of the vacuum locus. A related question concerns the dimension of the various connected components, in particular whether it can be zero. In other words, one might ask whether all complex structure moduli can always be stabilized for a suitable choice of flux. When one is only solving the vacuum conditions, it is reasonable to expect that this can indeed be achieved, since one imposes at least $h^{3,1}$ complex conditions for the same number of complex variables. However, it is not obvious whether this can be done whilst also imposing the tadpole condition. Indeed, the tadpole conjecture postulates that one cannot stabilize a large number of complex structure moduli within the tadpole bound, i.e.~when $h^{3,1}$ is much larger than all other Hodge numbers \cite{Bena:2020xrh}. More concretely, it states that for large $h^{3,1}$ and all moduli stabilized, one has
\begin{equation}
    \frac{1}{2}\int_{Y_4}G_4\wedge G_4 > \alpha h^{3,1}\,,
\end{equation}
with $\alpha>1/3$. Recalling that $\frac{\chi(Y_4)}{24}\sim \frac{1}{4}h^{3,1}$, this implies that for large $h^{3,1}$ the tadpole grows too quickly to be contained within the tadpole bound. We refer the reader to \cite{Braun:2020jrx,Bena:2021wyr,Marchesano:2021gyv,Lust:2021xds,Plauschinn:2021hkp,Grana:2022dfw,Lust:2022mhk,Coudarchet:2023mmm,Braun:2023pzd} for related works on the tadpole conjecture. 

Let us attempt to formulate a version of the tadpole conjecture in a more mathematical fashion. Let $v\in H_{\mathbb{Z}}$ be an integral class, playing the role of the flux, and denote by 
\begin{equation}
    (v,v) = L\,,
\end{equation}
its self-intersection. The spirit of the tadpole conjecture is that when the flux $v$ defines a vacuum in which all moduli are stabilized, one necessarily has $L>\mathcal{O}(1)\cdot\mathrm{dim}\,\mathcal{M}$, where we recall that $\mathcal{M}$ denotes the complex structure moduli space. Conversely, if $\mathrm{dim}\,\mathcal{M}>\mathcal{O}(1)\cdot L$, then it must be that not all moduli are stabilized. The latter statement can be formalized as follows. Generically, the vacuum locus consists of several connected components, each having a well-defined notion of dimension.\footnote{See also \cite{Becker:2022hse} for a related discussion.} For the locus of Hodge classes this is immediately clear, since it is algebraic. For the locus of self-dual classes this follows from its definability in an o-minimal structure, since a natural notion of dimension is provided by the cell decomposition \cite{dries_1998}. Within this locus, some components may correspond to points, having dimension zero, while other components may correspond to higher-dimensional loci, having strictly positive dimension. The statement that not all moduli are stabilized then means that all components of the vacuum locus of with a fixed self-intersection $L$ have strictly positive dimension. For the class of Hodge vacua, such a special feature of the vacuum locus appears to be more plausible. Thus we are lead to the following
\begin{conjecture}
\label{conjecture:tadpole}
    Consider a variation of polarized Hodge structure $E\rightarrow\mathcal{M}$ of weight $D$. Fix a positive integer $L$ and recall the notation
    \begin{equation}
        E_{\mathrm{Hodge}}(L) = \{(z^i,v)\in E: v\in H^{k,k}\cap H_{\mathbb{Z}}\,,(v,v)\leq L\}\,,\qquad D=2k\,,
    \end{equation}
    for the locus of Hodge classes with self-intersection bounded by $L$. We conjecture that for certain positive constants $C_1, C_2$, which are independent of $L$ and $\mathrm{dim}\,\mathcal{M}$ (but may depend on other details of the variation of Hodge structure, such as the weight $D$), the following holds: if
    \begin{equation}   
    \label{eq:conj3_conditions}\mathrm{dim}\,\mathcal{M}>C_1\,\qquad \text{and}\qquad \mathrm{dim}\,\mathcal{M}>C_2\cdot L\,,
    \end{equation}
    then every connected component of $E_{\mathrm{Hodge}}(L)$ has strictly positive dimension.\footnote{Note that since $L\geq 1$ for non-trivial fluxes, the two conditions in \eqref{eq:conj3_conditions} reduce to a single condition whenever $C_2\geq C_1$.} Furthermore, when the variation of Hodge structure comes from the middle cohomology of a family of Calabi--Yau fourfolds, we expect that the constant $C_2$ is of order one. 
\end{conjecture}

\begin{figure}[t]
    \centering
    \begin{tikzpicture}[scale=0.6]
\draw[->, thick] (1.5, 0) to[out=10, in=160] (10, -0.5) to [out=60, in=210] (13.5, 3)
    to[out=140,in=10] (4,4) to[out=250, in=90] cycle;

\draw[line width=1.5pt, color=blue!60] (3,1.5) to[in=220, out=30] (12,3);

\fill[, color=red!60] (10,2.05) circle[radius=3pt];

\draw[dashed] (4,2.5) to[in=270, out=90] (3,6);

\draw[dashed] (10,2.5) to[in=280,out=90] (11,6);

\draw[->,thick] (10,6.5) -- (10,9.5);
\draw[->,thick] (10,6.5) -- (13,6.5);
\draw[thin] (11,6.5) -- (11,9.5);
\draw[thin] (12,6.5) -- (12,9.5);
\draw[thin] (10,7.5) -- (13,7.5);
\draw[thin] (10,8.5) -- (13,8.5);
\draw[->, line width=1.5pt, color=red!60] (10,6.5) -- (12,9.5);
\node at (12,9.8) {$v_2$};

\draw[->,thick] (2,6.5) -- (5,6.5);
\draw[->,thick] (2,6.5) -- (2,9.5);
\draw[thin] (3,6.5) -- (3,9.5);
\draw[thin] (4,6.5) -- (4,9.5);
\draw[thin] (2,7.5) -- (5,7.5);
\draw[thin] (2,8.5) -- (5,8.5);
\draw[->, line width=1.5pt, color=blue!60] (2,6.5) -- (5,7.5);
\node at (5.3,7.5) {$v_1$};

\node at (5,9.5) {$H_{\mathbb{Z}}$};
\node at (13,9.5) {$H_{\mathbb{Z}}$};

\node at (12,1) {$\mathcal{M}$};
\end{tikzpicture}
    \caption{Schematic illustration of the Hodge bundle, in which the flux lattice $H_{\mathbb{Z}}$ is fibered over the moduli space $\mathcal{M}$. In blue and in red we have depicted two components of the vacuum locus with different dimensionality, corresponding to two choices of flux $v_1$ and $v_2$, respectively.}
    \label{fig:Hodge_bundle}
\end{figure}
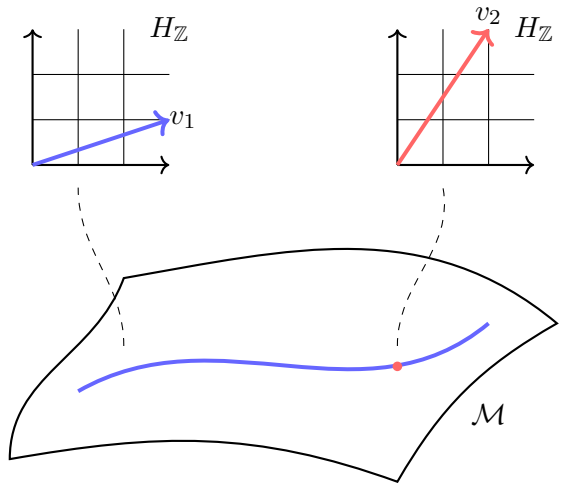
On the one hand, the statement of Conjecture \ref{conjecture:tadpole} is more general than the original tadpole conjecture of \cite{Bena:2020xrh}, as it is formulated for a general variation of Hodge structure. On the other hand, it should be emphasized that, in the specific setting of Calabi--Yau fourfold compactifications, the statement of Conjecture \ref{conjecture:tadpole} is weaker than the original tadpole conjecture, for a number of reasons. Firstly, Conjecture \ref{conjecture:tadpole} is formulated for Hodge vacua only, corresponding to vacua with $W_{\mathrm{flux}}=0$, while the original tadpole conjecture  applies to all self-dual vacua. Additionally, in the formulation of Conjecture \ref{conjecture:tadpole} there is no restriction on how many moduli are left unstabilized, as long as there is at least one. Finally, the exact values of the constants $C_1$ and $C_2$ are left undetermined. Especially for the physical application of studying the landscape of fully stabilized Hodge vacua, it is of utmost importance to quantify the exact values of $C_1,C_2$.

In figure \ref{fig:Hodge_bundle} we have illustrated two possible components of the locus of Hodge classes to exemplify the statement of Conjecture \ref{conjecture:tadpole}, for the case of a two-dimensional flux lattice $H_{\mathbb{Z}}$ and a real two-dimensional moduli space $\mathcal{M}$, so that $\mathrm{dim}\,E=4$. Suppose that $v_1$ is a choice of flux with sufficiently small tadpole $L_1$ so that Conjecture \ref{conjecture:tadpole} applies. Then the  vacuum locus corresponding to $v_1$ inside the full Hodge bundle may, for example, be a one-dimensional curve. Thus, this component of $E_{\mathrm{Hodge}}(L_1)$ has positive dimension. If instead $v_2$ is another choice of flux, with tadpole $L_2$, for which the corresponding vacuum is simply a point, then $L_2$ must be sufficiently large, in particular $L_2>L_1$. Note that it could additionally happen that this point lies on the component of the vacuum locus corresponding to $v_1$, as indicated in figure \ref{fig:Hodge_bundle}. In order to disentangle the two components, one should always consider the vacuum locus within the full Hodge bundle. To conclude, we believe a positive or negative answer to Conjecture \ref{conjecture:tadpole} would be an important step towards proving or disproving the tadpole conjecture. Although the conjecture remains rather speculative, it is conceivable that at least for Hodge classes a definite answer can be given in the near future. 

\section{Asymptotic Hodge Inner Products}
\label{sec:asymp_Hodge_inner_products}

In this section we provide some additional  material on asymptotic Hodge theory. This includes some results which have not yet appeared in the physics literature, which in fact comprise the core of the multi-variable $\mathrm{Sl}(2)$-orbit theorem of Cattani, Kaplan, and Schmid. In particular, we discuss the multi-variable nilpotent orbit expansion and show how this can be used to obtain general formulae for asymptotic Hodge inner products which include infinite series of sub-leading corrections. For the purpose of the present work, the main application of these results will be presented in section \ref{sec:self_dual_locus}, where the detailed properties of the nilpotent orbit expansion play a central role in the proof of Theorem \ref{thm:finiteness_selfdual_nilpotent}. However, the range of possible applications for these results goes far beyond just the finiteness proof. Indeed, as an example we present a general asymptotic formula for the central charge of D3-particles in type IIB  compactifications on Calabi--Yau threefolds in section \ref{subsec:inner_prod}. As such, this section may be of independent interest to some readers.

\subsection{Asymptotic Hodge theory and the nilpotent orbit expansion}
\label{subsec:nilpotent_orbit_expansion}
In the remainder of this section, we will assume the underlying variation of Hodge structure to be given exactly by a nilpotent orbit, in other words
\begin{equation}
\label{eq:nilpotent_orbit_2}
    F^p(t) = F^p_{\mathrm{nil}}(t) = e^{t^iN_i}F_0^p\,,
\end{equation}
recall also the discussion in section \ref{subsubsec:asymp_Hodge_theory_1}. In the following, we sometimes drop the subscript `nil' to avoid cluttering the notation. In order to tackle the finiteness of self-dual flux vacua, it is clearly necessary to understand the properties of the Weil operator $C_{\mathrm{nil}}$ associated to $F^p_{\mathrm{nil}}$, as it plays a central role in the vacuum conditions. In particular, it is necessary to know how $C_{\mathrm{nil}}$ can degenerate in the limit $\mathrm{Im}\,t^i\rightarrow\infty$. Unfortunately, the characterization \eqref{eq:nilpotent_orbit_2} of the nilpotent orbit is not immediately useful in this regard. The reason is that $F_0^p$ itself generically does not define a polarized Hodge structure, hence there is no Weil operator associated to it. 

Nevertheless, there exists a completely general procedure which characterizes $F^p_{\mathrm{nil}}$ and, consequently, $C_{\mathrm{nil}}$. The procedure lies at the heart of the proof of the $\mathrm{Sl}(2)$-orbit theorem of Cattani, Kaplan, and Schmid \cite{CKS}, and was referred to as a ``bulk reconstruction''  procedure in \cite{Grimm:2021ikg}. The idea is that, to each boundary in the moduli space, in particular to a given nilpotent orbit, one can naturally associate a set of so-called ``boundary data''
\begin{equation}
\label{eq:boundary_data}
\boxed{\rule[-.3cm]{0cm}{.8cm} \quad \left\{F^p_{\infty}, N_{(i)}^+, N_{(i)}^0, N_{(i)}^-, \delta_{(i)}\right\}\,,\qquad i=1,\ldots, m\,,\quad }
\end{equation}
consisting of
\begin{itemize}
    \item a boundary Hodge structure $F^p_\infty$\,,
    \item a collection of $m$ real $\mathfrak{sl}(2)$-triples $\{N_{(i)}^+, N_{(i)}^0, N_{(i)}^-\}$ and
    \item a collection of $m$ ``phase operators'' $\delta_{(i)}\in\mathfrak{g}_{\mathbb{R}}$. 
\end{itemize}
This will be explained in more detail shortly. The important point is that $F_\infty^p$ \textit{does} define a polarized Hodge structure, and therefore has an associated Weil operator $C_\infty$. Furthermore, given a set of boundary data, there exists a moduli-dependent $G_{\mathbb{R}}$-valued function $h$ with which the original nilpotent orbit as well as the Weil operator can be ``reconstructed'' via the relations
\begin{equation}
\label{eq:bulk-reconstruction}
    F^p_{\mathrm{nil}} = h F^p_{\infty}\,,\qquad C_{\mathrm{nil}} = h C_\infty h^{-1}\,.
\end{equation}
The operator $h$ is given via a completely algorithmic manner in terms of the $\mathfrak{sl}(2)$-triples and phase operators that comprise the boundary data. Together with the fact that all the possible boundary data can be classified \cite{robles_2016,Kerr2017}, this therefore provides a complete characterization of nilpotent orbits. In the following two subsections, we will provide some additional background on the boundary data, as well as the general form and properties of the function $h$.

\subsubsection{Asymptotic Hodge theory (2): Boundary data}
In this subsection we describe some of the essential properties of the boundary data \eqref{eq:boundary_data}. For the purpose of this work, it will not be necessary to understand exactly how the boundary data can be obtained, or classified, in general. Instead, it will be sufficient to use the existence of this data as well as their properties and role in the bulk reconstruction procedure. The interested reader may consult \cite{Grimm:2018cpv,Grimm:2021ckh} for further details, as well as appendix \ref{app:MHS}.

\subsubsection*{$\mathfrak{sl}(2,\mathbb{R})$-decomposition}
One of the central results of asymptotic Hodge theory is that, given a variation of Hodge structure on a product of $m$ punctured disks, there is a procedure to construct a set of $m$ commuting $\mathfrak{sl}(2,\mathbb{R})$-triples that is naturally associated to the limit. In other words, each boundary in the moduli space has, in an appropriate sense, an emergent $\mathfrak{sl}(2,\mathbb{R})^m$ symmetry. To be precise, each boundary actually has multiple of such emergent symmetries, depending on the hierarchy among the moduli that become large. To this end, we introduce a growth sector
\begin{equation}
\label{eq:growth_sector}
    R_{12\cdots m} = \{(t^1,\ldots, t^m)\in\mathbb{H}^m: \mathrm{Im}\,t^1>\mathrm{Im}\,t^2>\cdots > \mathrm{Im}\,t^m>1\,, \mathrm{Re}\,t^i \in [0,1]\}\,,
\end{equation}
and will, without loss of generality, restrict the remainder of our discussion to this particular growth sector. Pictorially, this means that, within this growth sector, one is always closest to the $z_1=0$ singularity, followed by the $z_2=0$ singularity, et cetera. Of course, by a reordering of the coordinates one can always restrict to this case. However, it is important to stress that in practical applications, when computing the explicit generators of the $\mathfrak{sl}(2,\mathbb{R})$-algebras, one will get different results in the different sectors.

For each $i=1,\ldots, m$, the corresponding $\mathfrak{sl}(2,\mathbb{R})$-triple will be denoted by a set of three real operators $N^+_i, N^0_i,N^-_i\in\mathfrak{g}_{\mathbb{R}}$, satisfying the usual commutation relations
\begin{equation}
    [N^0_i, N^\pm_j] = \pm 2 N^\pm_i\,\delta_{ij}\,,\qquad [N^+_i,N^-_j]=N^0_i\,\delta_{ij}\,.
\end{equation}
Furthermore, we define
\begin{equation}
\label{eq:sl2_operators_sum}
    N^\bullet_{(i)} = N^\bullet_1+\cdots+N^\bullet_i\,,\qquad \bullet=+,0,-\,.
\end{equation}
The operators $N^0_{(i)}$ (which are also mutually commuting) will be of particular importance in the rest of the discussion. This is because they induce a decomposition of the vector space $H_{\mathbb{R}}$ in terms of weights with respect to each $N^0_{(i)}$. Indeed, for a given vector $v\in H_{\mathbb{R}}$ its weight-decomposition will be denoted by\footnote{A word of caution: it is of course also possible to use $N^0_i$ to define a weight-decomposition. This is simply a matter of convention, which differs across different works. For the present work, we find this choice to be most convenient and natural.}
\begin{equation}
\label{eq:sl2_decomp_vector}
    v = \sum_{\ell} v_\ell\,,\qquad \ell=(\ell_1,\ldots, \ell_m)\,,\qquad N^0_{(i)}v_\ell = \ell_i\,v_\ell\,.
\end{equation}
Here the values of the $\ell_i$ run at most from $-D$ to $D$. In a similar fashion, the adjoint representation of the $\mathfrak{sl}(2,\mathbb{R})$-triples on the algebra $\mathfrak{g}_{\mathbb{R}}$ induces a weight-decomposition of an operator $\mathcal{O}\in\mathfrak{g}_\mathbb{R}$ as 
\begin{equation}
\label{eq:sl2_decomp_operator}
    \mathcal{O}=\sum_{s} \mathcal{O}^{s}\,,\qquad s=(s_1,\ldots, s_m)\,,\qquad [N^0_{(i)},\mathcal{O}^s] = s_i\,\mathcal{O}^s\,.
\end{equation}
Here the values of the $s_i$ run at most from $-2D$ to $2D$.

\subsubsection*{Phase operators}
From a computational perspective, the phase operators $\delta_{(i)}$ are, arguably, the most important part of the boundary data. This is because the form of the $\delta_{(i)}$ dictates how complicated the resulting expression for the map $h$ becomes. In particular, if all $\delta_{(i)}$ vanish the procedure essentially trivializes. The construction of the $\delta_{(i)}$ associated to a given nilpotent orbit is somewhat involved and is described in e.g.~\cite{Grimm:2021ckh}. Roughly speaking, the presence of the $\delta_{(i)}$ is intertwined with the reality of the $\mathfrak{sl}(2)$-triples. In order to ensure this reality, it is generically necessary to perform certain rotations, generated by the $\delta_{(i)}$, on the limiting filtration $F_0^p$ in order to remove complex phase factors. Hence the name ``phase operators''. Some further details are also described in appendix \ref{app:MHS}.

The phase operators satisfy two properties which play a crucial role in the proof of Theorem \ref{thm:finiteness_selfdual_nilpotent}, namely\footnote{To be precise, the phase operators $\delta_{(i)}$ should satisfy the following condition with respect to the so-called Deligne splitting
\begin{equation}
\label{eq:propertes_delta_Deligne}
    \delta_{(i)} = \sum_{p,q>0} \left[\delta_{(i)}\right]_{-p,-q}\,,\qquad \left[\delta_{(i)}\right]_{-p,-q}\left(\tilde{I}^{r,s}_{(i)}\right)\subseteq \tilde{I}^{r-p,s-q}_{(i)}\,,
\end{equation}
see also appendix \ref{app:MHS}. This condition implies \eqref{eq:properties_delta_2}. Alternatively, one may also formulate \eqref{eq:propertes_delta_Deligne} by introducing a so-called charge operator as was done in \cite{Grimm:2021ikg}.}
\begin{align}
\label{eq:properties_delta_1}
   (1):&\qquad [N_{(j)}^-, \delta_{(i)}] = 0\,,\qquad j\leq i\,,\\
\label{eq:properties_delta_2}
   (2):&\qquad \delta_{(i)}^{s^i}=0\,,\quad s^i=\left(s_1^i,\ldots, s_m^i\right)\,,\quad \text{if $s_i^i>-2$}\,.
\end{align}
The first property states that each $\delta_{(i)}$ is a lowest-weight operator with respect to $N^-_{(1)},\ldots, N_{(i)}^-$, while the second property imposes the additional restriction that its weight with respect to $N_{(i)}^0$ is less than or equal to $-2$.

\subsubsection*{Boundary Hodge structure}
Another important result of asymptotic Hodge theory is that it is possible to assign a sensible ``boundary Hodge structure'' to the puncture of the polydisc, which will be denoted by $H^{p,q}_{\infty}$. Correspondingly, the associated boundary Hodge filtration will be denoted by $F^p_{\infty}$, and its Weil operator by $C_\infty$. Again, it will not be necessary to understand the full details of how this is constructed, for which we refer the reader to \cite{Grimm:2021ckh} and appendix \ref{app:MHS}. However, an important point we would like to stress is that the boundary Hodge structure has a well-defined inner product, namely the one induced by $C_\infty$, which is moreover coordinate-independent. 

Let us now explain the sense in which one should think of $F^p_\infty$ as the boundary Hodge structure. To this end, we introduce the following real operator
\begin{equation}
    e(y) = \prod_{i=1}^m\left(\frac{y_i}{y_{i+1}}\right)^{\frac{1}{2}N^0_{(i)}}\,,\qquad y_{m+1}\equiv 1\,, 
\end{equation}
which takes values in $G_{\mathbb{R}}$. Then one can show that (assuming the axions remain bounded)
\begin{equation}
\label{eq:Fsharp_F_relation_limit}
    \lim_{y_1,\ldots, y_m\rightarrow\infty} e(y)F^p = F^p_{\infty}\,.
\end{equation}
There are two ways to interpret this property. On the one hand, one can use \eqref{eq:Fsharp_F_relation_limit} to investigate the properties of an element in $F^p$ as one approaches the boundary. This is the perspective we will take in the remainder of this section. On the other hand, equation \eqref{eq:Fsharp_F_relation_limit} roughly implies that $e(y)^{-1}F^p_{\infty}$ is the `leading approximation' to $F^p$. One might wonder whether it is possible to then compute the sub-leading corrections in order to give a detailed description of elements in $F^p$. This can indeed be done, as will be explained in the next section, and will be of central importance to analyse the fate of self-dual vacua.

\subsubsection{Asymptotic Hodge theory (3): Multi-variable bulk reconstruction}

Let us now turn to the computation of the operator $h$, by which the full nilpotent orbit $F_{\mathrm{nil}}$ can be recovered from just the boundary data. First, one may factor out the axion dependence and write
\begin{equation}
\boxed{\rule[-.5cm]{0cm}{1.2cm} \quad
    F_{\mathrm{nil}}= \mathrm{exp}\left[\sum_{i=1}^m x_i N_i\right]\cdot h(y_1,\ldots, y_m)\cdot F_{\infty}\,,\quad}
\end{equation}
with the map $h(y_1,\ldots, y_m)$ depending only on the saxions. In the case of a single variable ($m=1$) this map has been constructed explicitly for all possible boundary data that can arise for Calabi--Yau threefolds in \cite{Grimm:2021ikg}, where the procedure has been referred to as a ``bulk reconstruction'', in analogy with a similar procedure in holography. We also refer the reader to \cite{Grimm:2020cda}, in which such a notion of ``moduli space holography'' was first proposed. Moving on, the multi-variable bulk reconstruction can be viewed as a clever recursive application of the one-variable procedure. For the details of the exact construction of the map $h(y_1,\ldots, y_m)$ we refer the reader to the proof of the $\mathrm{Sl}(2)$-orbit theorem of \cite{CKS}, as well as the upcoming doctoral thesis of the second author. In the following, we will simply state the important properties.

\subsubsection*{Properties of $h(y_1,\ldots, y_m)$}
Thus, let us provide some details on the map $h(y_1,\ldots,y_m)$. First off, due to the inductive nature of its construction, it takes the form of a product
\begin{equation}
    h(y_1,\ldots, y_m) = \prod_{i=m}^1 h_i\left(\frac{y_i}{y_{i+1}}; \frac{y_1}{y_i},\ldots, \frac{y_{i-1}}{y_i}\right)\,,\qquad y_{m+1}\equiv 1\,.
\end{equation}
By the notation on the right-hand side it is meant that each factor $h_i$ should be viewed as a function of $\frac{y_i}{y_{i+1}}$, while the ratio's $\frac{y_1}{y_i},\ldots, \frac{y_{i-1}}{y_i}$ are kept fixed. These functions are each of the form
\begin{equation}
    h_i =  g_{i}\cdot\left(\frac{y_i}{y_{i+1}}\right)^{-\frac{1}{2}N^0_{(i)}}\,,\qquad g_{i}=1+\sum_{k_i=1}^\infty g_{i,k_i}\left(\frac{y_i}{y_{i+1}}\right)^{-k_i}\,.
\end{equation}
In particular, each $h_i$ admits a series expansion in $y_i/y_{i+1}$. The expansion coefficients $g_{i,k_i}$ are then the most important and complicated part of the expression. In the one-variable case, they can be expressed as moduli-independent universal Lie polynomials involving certain projections of the (single) phase operator $\delta$ and raising operator $N^+$ onto specific eigenspaces of the grading operators $N_{(i)}^0$ and $Q_{\infty}$. This is explained in detail in \cite{Grimm:2021ikg}. A major complication of the multi-variable case is that the expansion coefficients $g_{i,k_i}$ become moduli-dependent. Specifically, one has
\begin{equation}
    g_{i,k_i}=g_{i,k_i}\left(\frac{y_{i-1}}{y_i},\ldots,\frac{y_{1}}{y_2}\right)\,.
\end{equation}
In other words, for $i>1$, each $g_{i,k_i}$ is a function of all the previous $y_j$, with $j\leq i$. More precisely, each $g_{i,k_i}$ itself admits a power series expansion in $y_{i-1}/y_i,\ldots, y_1/y_2$. The origin of this additional moduli-dependence can roughly be understood as follows. In order to compute $h_1$, one should apply the one-variable bulk reconstruction procedure using the data $\delta_{(1)}$ and $N_{(1)}^+$. However, due to the nature of the inductive construction, in order to compute $h_2$, one should then apply the one-variable bulk reconstruction procedure using the data $\delta_{(2)}$ and $h_1 N_{(2)}^+ h_1^{-1}$. In other words, one should first `translate' the operator $N_{(2)}^+$ using the adjoint action of $h_1$, which introduces an additional moduli dependence. The same principle then extends in the natural way to all $h_i$. The final result for the map $h(y_1,\ldots, y_m)$ that is obtained via this recursive procedure should therefore be seen as a very non-trivial collection of nested power series. 

\subsubsection*{A useful rewriting}
Let us present a useful rewriting of the map $h(y_1,\ldots, y_m)$. In fact, for our purposes it will be more important to consider the form of the inverse operator $h^{-1}$, which admits a similar expansion
\begin{equation}
    h^{-1} =h_1^{-1}\cdots h_m^{-1}\,,
\end{equation}
where
\begin{equation}
\label{eq:hinv_before_commutation}
    h_i^{-1} = \left(\frac{y_{i}}{y_{i+1}}\right)^{\frac{1}{2}N^0_{(i)}}\left[\sum_{k_i=0}^\infty \left(\frac{y_{i}}{y_{i+1}}\right)^{-k_i} f_{i,k_i}\right]\,,\qquad f_{i,0}\equiv 1\,.
\end{equation}
Of course, the expansion coefficients $f_{i,k_i}$ can be straightforwardly related to the $g_{i,k_i}$ via an order-by-order inversion. For practical purposes, it will be useful to commute the factor left of the square brackets \eqref{eq:hinv_before_commutation} to the right by employing the weight-decomposition \eqref{eq:sl2_decomp_operator} of the $f_{i,k_i}$ coefficients. This yields the following expression
\begin{equation}
\label{eq:hinv}
\boxed{ \rule[-.7cm]{0cm}{1.6cm}
\quad     h^{-1} = \sum_{k_1,\ldots, k_m=0}^\infty \sum_{s^1,\ldots, s^m}\prod_{i=1}^m \left[\left(\frac{y_i}{y_{i+1}}\right)^{-k_i+\frac{1}{2}s_i^{(m)}}f^{s^i}_{i,k_i}\right] h^{-1}_{\mathrm{Sl}(2)}\,,\quad }
\end{equation}
where we introduced the notation
\begin{equation}
    s_i^{(m)} = s_i^{i}+\cdots +s_i^{m}\,,
\end{equation}
and identified
\begin{equation}
    h^{-1}_{\mathrm{Sl}(2)} = e(y)=\prod_{i=1}^m \left(\frac{y_i}{y_{i+1}}\right)^{\frac{1}{2}N_{(i)}^0}\,,
\end{equation}
which corresponds to the $\mathrm{Sl}(2)$-orbit approximation of the inverse period map. It is usually denoted by $e(y)$ and we will use this notation as well. 

\subsubsection*{Properties of $f_{i,k_i}$}
We end our discussion by stating two important properties of the expansion functions $f_{i,k_i}$. Both of these properties will feature prominently in the finiteness proof of self-dual vacua, as will be explained in section \ref{subsec:self-dual_proof}. Recalling the notation introduced in \eqref{eq:sl2_decomp_operator}, the first property is
\begin{equation}
\label{eq:expansion_coeff_weights}
    f_{i,k_i}^{s^i} = 0\,,\qquad s^i=(s_1^i,\ldots, s_m^i)\,,\qquad \text{if $s_i^i> k_i-1$ or $s_j^i\neq 0$ for $j>i$}\,.
\end{equation}
In words, this means that the last non-trivial weight of each $f_{i,k_i}$ is given by $s_i^i$, and that the value of this weight is restricted by the order $k_i$ at which this coefficient appears in the expansion. In particular, this implies that the second sum in \eqref{eq:hinv} only runs over $s_i^i\leq k_i-1$. The attentive reader may note that the properties \eqref{eq:expansion_coeff_weights} are somewhat similar to those of the phase operators $\delta_{(i)}$, see in particular \eqref{eq:properties_delta_2}. In fact, the former are a partial consequence of the latter, illustrating the importance of the properties of the phase operators within the bulk reconstruction procedure.

While the first property \eqref{eq:expansion_coeff_weights} restricts the possible weights of the expansion functions, the second property gives a bound on the scaling of the expansion functions in terms of their weights
    \begin{equation}  
    \label{eq:f_bounds_new}
    f_{i,k_i}^{s^i}\prec \prod_{j=1}^{i-1} \left(\frac{y_j}{y_{j+1}}\right)^{-s_j^i}\,,\qquad 2\leq i\leq m\,.
    \end{equation}
For the case $i=1$ the coefficients are simply constant matrices. Here the notation $\prec$ means the following: for two functions $f,g$ we write $f\prec g$ when $f$ is bounded by a constant multiple of $g$. The relation \eqref{eq:f_bounds_new} effectively describes the scaling of the various $f_{i,k_i}$ corrections in terms of their $\mathrm{sl}(2)$-weights. We would like to mention that, to our knowledge, the bound \eqref{eq:f_bounds_new} has not been written down explicitly before, although it does follow directly from the proof of the $\mathrm{Sl}(2)$-orbit theorem of \cite{CKS}. We have therefore presented a proof of \eqref{eq:f_bounds_new} in Appendix \ref{app:additional_proofs}. In simple terms, one needs to use the details of how the expansion coefficients are derived from the boundary data \eqref{eq:boundary_data} and, in particular, use some of the restrictions on the $\mathrm{sl}(2)$-weights of the phase operators $\delta_{(i)}$.

\subsubsection{Asymptotic Hodge inner products}
\label{subsec:inner_prod}

With the result \eqref{eq:hinv} for the nilpotent orbit approximation of the period map at hand, it is possible to evaluate general Hodge inner products in great generality. Indeed, given two elements $v,w\in H$, the nilpotent orbit approximation of their Hodge inner product can be expressed as
\begin{equation}
    \langle v, w\rangle_{\mathrm{nil}} = \left(v, C_{\mathrm{nil}} \bar{w}\right) = \left(v, h C_\infty h^{-1}\bar{w}\right) = \left(h^{-1}, C_\infty h^{-1}\bar{w}\right) = \langle h^{-1}v, h^{-1}w\rangle_\infty\,.
\end{equation}
In particular, all the moduli-dependence of the Hodge inner product is captured by the factors of $h^{-1}$. In general, the expressions that result from inserting \eqref{eq:hinv} can become rather involved, in particular due to the fact that many cross-terms can emerge. 

\subsubsection*{Central charge of D3-particles}
There is at least one case, however, in which the resulting expressions simplify nicely. This is the case where one is considering generic Hodge inner products between the period vector $\mathbf{\Pi}$, or equivalently the $(D,0)$-form $\Omega$, and an element of a definite $\mathrm{sl}(2)$-weight, denoted by $q$ in the following. This is of particular relevance when studying BPS states that arise from D3-branes wrapping a particular class of three-cycles in a Calabi--Yau threefold, in the context of type IIB string theory compactifications. Indeed, for a given BPS state, parametrized by a charge vector $q\in H^3(Y_3,\mathbb{Z})$, its mass is given by
\begin{equation}
\label{eq:central_charge}
    |\mathcal{Z}| = \frac{|\langle q,\Omega\rangle|}{||\Omega||}\,,
\end{equation}
where $\mathcal{Z}$ denotes the central charge of the BPS state, and $\Omega$ the holomorphic $(3,0)$-form on $Y_3$. Let us evaluate \eqref{eq:central_charge} in the nilpotent orbit approximation. First, since $\Omega\in H^{3,0}$, we may apply the relation \eqref{eq:bulk-reconstruction} and write
\begin{equation}
    h^{-1}\Omega = f\cdot \Omega_\infty\,,
\end{equation}
for some moduli-independent element $\Omega_\infty\in H^{3,0}_\infty$. Since \eqref{eq:bulk-reconstruction} is a vector-space identity, there is of course the freedom of an overall (possibly moduli-dependent) scaling, which we parametrize by the function $f$. Importantly, in the expression \eqref{eq:central_charge} for the central charge, this overall factor will drop out. Indeed, one finds
\begin{equation}
    |\mathcal{Z}| = \frac{|\langle h^{-1}q, h^{-1}\Omega\rangle_\infty|}{||h^{-1}\Omega||_\infty} = \frac{|\langle h^{-1}q, \Omega_\infty\rangle_\infty|}{||\Omega_\infty||_\infty}\,.
\end{equation}
Assuming then, for simplicity, that $q$ has a definite $\mathrm{sl}(2)$ weight given by $\ell$, and applying the result \eqref{eq:hinv} one finds 
\begin{equation}
\label{eq:central_charge_nilpotent}
    | \mathcal{Z} |= \prod_{i=1}^m \left(\frac{y_i}{y_{i+1}}\right)^{\frac{1}{2}\ell_i}\left|\sum_{k_1,\ldots, k_m=0}^\infty\sum_{s^1,\ldots, s^m} \left[\prod_{i=1}^m \left(\frac{y_i}{y_{i+1}} \right)^{-k_i+\frac{1}{2}s_i^{(m)}}\right] \frac{\langle f_{1,k_1}^{s^1}\cdots f_{m,k_m}^{s^m} q, \Omega_\infty\rangle_\infty}{||\Omega_\infty||_\infty}\right|
\end{equation}
where, for simplicity, we have set the axions $x^i$ to zero. Note that the overall prefactor in \eqref{eq:central_charge_nilpotent} comes from evaluating the action of $h^{-1}_{\mathrm{Sl}(2)}$ on the weight-eigenvector $q$. Focusing for the moment on the leading term in \eqref{eq:central_charge_nilpotent}, one may write
\begin{equation}
    |\mathcal{Z}| = \prod_{i=1}^m \left(\frac{y_i}{y_{i+1}}\right)^{\frac{1}{2}\ell_i}\cdot \left|\frac{\langle q, \Omega_\infty\rangle_\infty}{||\Omega_\infty||_\infty}+\mathrm{corrections}\right|\,.
\end{equation}
In particular, the asymptotic behaviour of the mass of the BPS state, given by $|\mathcal{Z}|$, is determined by weight of the charge vector $q$. This first-order expression has been used in \cite{Bastian:2020egp}, see also \cite{Gendler:2020dfp}, to compute the charge-to-mass ratio for this class of BPS states, under the assumption that the pairing $\langle q,\Omega_\infty\rangle$ is non-vanishing. Physically, this assumption can be interpreted as the statement that the asymptotic coupling of the BPS state to the graviphoton is non-zero. However, if the pairing $\langle q,\Omega_\infty\rangle$ does vanish, then it becomes necessary to consider the correction terms coming from the nilpotent orbit expansion in order to properly characterize the scaling of the mass. Notably, depending on the choice of $q$, as well as the type of boundary that one is expanding around, it may happen that a state which naively appears to become massive in fact becomes massless as one approaches the boundary. It would be interesting to generalize the analysis of \cite{Bastian:2020egp} to include the sub-leading corrections using the multi-variable bulk reconstruction procedure.

\section{The Asymptotic Self-dual Locus}
\label{sec:self_dual_locus}
In section \ref{sec:finiteness_intro} we have presented a number of finiteness theorems for both Hodge vacua and self-dual vacua, from both a global and a local perspective. The aim of this section is to apply the results of section \ref{sec:asymp_Hodge_inner_products} to prove Theorem \ref{thm:finiteness_selfdual_nilpotent}, and address the finiteness of self-dual vacua in the nilpotent orbit approximation. Before discussing the general proof, we first restrict to a simple one-variable setting in section \ref{subsec:self-dual_locus_example} in order to exemplify some features of the vacuum locus using the abstract machinery introduced in section \ref{sec:asymp_Hodge_inner_products}. In section \ref{subsec:self-dual_proof}, we present the full  proof of Theorem \ref{thm:finiteness_selfdual_nilpotent}.

\subsection{Example: one-variable $\mathrm{Sl}(2)$-orbit}
\label{subsec:self-dual_locus_example}
Before delving into the detailed proof of the finiteness result, let us first consider a very simple case in which there is just a single modulus, so $m=1$, and all the expansion coefficients (except the leading ones) in the nilpotent orbit expansion \eqref{eq:hinv} vanish. Effectively, this case will correspond to a generalization of the example discussed in section \ref{subsec:finiteness_intro}, though written in more abstract language. Mathematically, the stated assumptions imply that the variation of Hodge structure under consideration is given by a one-variable $\mathrm{Sl}(2)$-orbit
\begin{equation}
\label{eq:sl2_orbit}
    F^p = e^{xN}y^{-\frac{1}{2}N^0}F_{\infty}^p\,.
\end{equation}
We would like to investigate the set of points in the moduli space where a given $v\in H_{\mathbb{Z}}$ is self-dual and $v$ has a bounded Hodge norm, as imposed by the tadpole condition. In this simple setting, it is straightforward to evaluate these two conditions explicitly.
\subsubsection*{Tadpole constraint}
The Hodge norm of $v$ is given by
\begin{equation}
    ||v||^2 = \sum_{\ell} y^\ell \,||\hat{v}_\ell||^2_\infty\,,\qquad \hat{v}=e^{-xN}v\,,
\end{equation}
where we recall the notation \eqref{eq:sl2_decomp_vector} for the weight-decomposition. We are interested in the properties of vacua close to the boundary, i.e.~for large values of the saxion $y$. Clearly, for sufficiently large $y$, it is necessary to impose that $\hat{v}_\ell=0$ for all $\ell>0$, in order for $||v||^2$ to not exceed the tadpole bound. In other words, beyond some critical value of $y$, the flux is only allowed to have non-positive weights with respect to the $\mathfrak{sl}(2)$ grading operator $N^0$. 

\subsubsection*{Self-duality condition}
Using the fact that $C_\infty$ interchanges the $+\ell$ and $-\ell$ eigenspaces of $N^0$, it is straightforward to see that the self-duality condition, projected onto a weight $\ell$ component, can be written as
\begin{equation}
\label{eq:self-duality_sl2_one-variable}
    y^{\frac{1}{2}\ell} \hat{v}_\ell = y^{-\frac{1}{2}\ell} C_\infty \hat{v}_{-\ell}\,.
\end{equation}
In particular, whenever $\hat{v}_\ell=0$ for $\ell>0$ the self-duality condition imposes that additionally $\hat{v}_\ell=0$ for $\ell<0$ (note that $C_\infty$ is invertible). Therefore, for sufficiently large $y$ it must be that $\hat{v}$ only has an $\ell=0$ component, so $\hat{v}=\hat{v}_0$. In particular, the tadpole condition reduces to
\begin{equation}
    ||\hat{v}_0||_\infty^2 \leq L\,.
\end{equation}
Recalling the fact that $\hat{v}=e^{-x N}v$, that the axion $x$ is bounded and that the flux $v$ is integral, it is clear that there are only finitely many choices of $v$ which satisfy the tadpole condition. Hence, there are finitely many vacua. 

\subsubsection*{The vacuum locus}
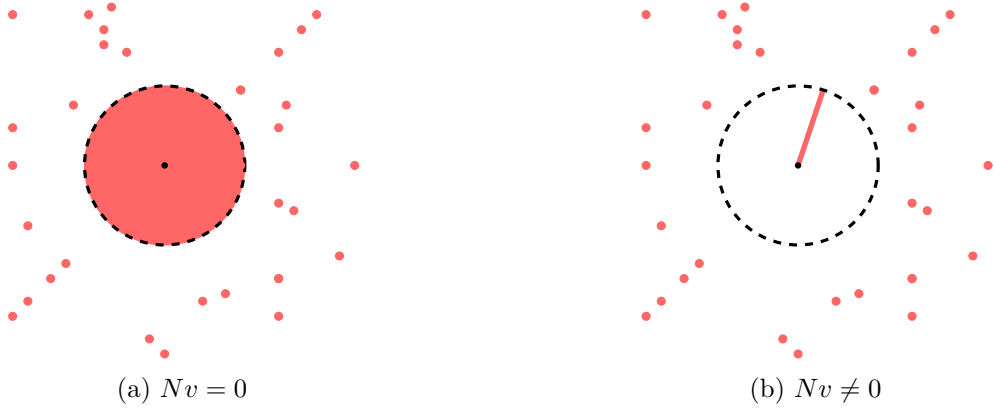
\begin{figure}[t]
\centering
\begin{subfigure}[b]{0.5\textwidth}
\centering
\begin{tikzpicture}
    \filldraw[color=red!60] (0,0) circle (30pt);    
    \filldraw[black] (0,0) circle (1pt);
    \draw[color=black, dashed, very thick] (0,0) circle (30pt);

    \foreach \i\j in {1.5 / 1.5, -0.7 / 2.1, 2.3 / -1.2, -1.8 / -1.8, 2.0 / 2.0,
 1.0 / 1.0, -1.2 / 0.8, 1.7 / -0.6, -0.6 / -0.6, 1.8 / 1.8,
 -0.2 / -2.3, 0.8 / -1.7, -2.0 / 0.0, 1.0 / 0.0, 0.0 / -2.5,
 -1.0 / 2.0, 2.5 / 0.0, -1.3 / -1.3, 1.5 / -0.5, -0.5 / 1.5,
 1.5 / -1.5, -2.0 / -2.0, 1.5 / 0.5,
 -0.8 / 1.8, 1.5 / -2.0, 1.5 / -1.5, -0.5 / -0.5, -1.8 / -0.8,
 1.0 / 1.0, 0.5 / -1.8, -0.8 / 0.0, -2.0 / 2.0, 0.5 / -0.5,
 -1.5 / -1.5, -2.0 / 0.5, 1.6 / 0.8, -0.8 / 1.6}
    {
    \filldraw[color=red!60] (\i,\j) circle (1.5pt);
    }
\end{tikzpicture}
\caption{$Nv=0$}
\label{fig:Hodge_locus}
\end{subfigure}
\hfill
\begin{subfigure}[b]{0.45\textwidth}
\centering
\begin{tikzpicture}
    \draw[color=red!60, line width = 0.7mm] (0,0)--(0.33,1);
    \filldraw[black] (0,0) circle (1pt);
    \draw[color=black, dashed, very thick] (0,0) circle (30pt);

    \foreach \i\j in {1.5 / 1.5, -0.7 / 2.1, 2.3 / -1.2, -1.8 / -1.8, 2.0 / 2.0,
 1.0 / 1.0, -1.2 / 0.8, 1.7 / -0.6, 1.8 / 1.8,
 -0.2 / -2.3, 0.8 / -1.7, -2.0 / 0.0, 0.0 / -2.5,
 -1.0 / 2.0, 2.5 / 0.0, -1.3 / -1.3, 1.5 / -0.5, -0.5 / 1.5,
 1.5 / -1.5, -2.0 / -2.0, 1.5 / 0.5,
 -0.8 / 1.8, 1.5 / -2.0, 1.5 / -1.5,  -1.8 / -0.8,
 1.0 / 1.0, 0.5 / -1.8, -2.0 / 2.0, 
 -1.5 / -1.5, -2.0 / 0.5, 1.6 / 0.8, -0.8 / 1.6}
    {
    \filldraw[color=red!60] (\i,\j) circle (1.5pt);
    }
\end{tikzpicture}
\caption{$Nv\neq 0$}
\label{fig:self-dual_locus}
\end{subfigure}
\caption{Schematic illustration of the distribution of self-dual vacua near a punctured disk (shown in red). Close enough to the singularity, dictated by the tadpole bound (indicated by the black dashed line), either (a) both axion and saxion are unstabilized, or (b) only the axion is stabilized, the latter case corresponding to a radial ray.}
\end{figure}
It is important to stress that for this simple example the only possible vacua close to the boundary have an unstabilized saxion. Indeed, since $\hat{v}$ only has an $\ell=0$ component, the self-duality condition \eqref{eq:self-duality_sl2_one-variable} simply becomes
\begin{equation}
    v = \left(e^{x N}C_\infty e^{-xN}\right)v\,,
\end{equation}
in which the saxion does not appear. Furthermore, one can make the following case distinction
\begin{enumerate}
    \item $Nv = 0$: In this case the self-duality condition reduces further to $v=C
    _\infty v$ and also the axion is unstabilized. The resulting vacuum locus is illustrated in figure \ref{fig:Hodge_locus}. In fact, it turns out that in this case $v$ actually corresponds to a Hodge class, as is explained in Appendix \ref{app:Hodge_locus}.
    \item $Nv\neq 0$: In this case a choice of $v$ uniquely fixes a value for the axion $x$. The vacuum locus therefore corresponds to a single angular ray in the disk. This is illustrated in figure \ref{fig:self-dual_locus}.  
\end{enumerate}
The general lesson of this simple one-variable example is the following: as one approaches the boundary of the moduli space, one is more and more restricted in the allowed fluxes, i.e.~the allowed $\mathrm{sl}(2)$ components of the fluxes, that can possibly satisfy the self-duality condition and the tadpole constraint. Eventually, the restrictions become so severe that one can directly show that there are only finitely many possibilities. As will be explained in the next section, a similar phenomenon happens in the multi-variable case. However, the restrictions on the fluxes become dependent on the sector of the moduli space in which the vacua are located. 

\subsection{Proof: finiteness of self-dual vacua}
\label{subsec:self-dual_proof}
For ease of reference, we repeat the exact theorem we aim to prove.
\setcounter{thm}{3}
\begin{thm}
    Let $t^i(n)\in\mathbb{H}^m$
    be a sequence of points such that $x^i(n)$ is bounded and $y^i(n)\rightarrow\infty$ as $n\rightarrow\infty$. Suppose furthermore that $v(n)\in H_{\mathbb{Z}}(L)$ is a sequence of integral fluxes with bounded self-intersection, such that
    \begin{equation}
    \label{eq:self-duality_nil}
        C_{\mathrm{nil}}(t(n))v(n) = v(n)\,,
    \end{equation}
    for all $n$. Then $v(n)$ can only take on finitely many values. 
\end{thm}
We emphasize again that our proof is restricted to the case where the variation of Hodge structure under consideration is described by a nilpotent orbit. In the one-variable case it is relatively straightforward to reduce to this case from the general setting of an arbitrary variation of Hodge structure by including an appropriate exponentially small correction term to the flux sequence $v(n)$ \cite{Schnellletter}. However, in the multi-variable case it is not clear whether a similar strategy can be applied.

In the following, we will often omit the argument in $y(n)$ and simply write $y$, to avoid cluttering the notation. As has been described a few times already, combining the self-duality condition with the tadpole bound on the self-intersection of $v(n)$ gives the following bound on the Hodge norm
\begin{equation}
\label{eq:bound_v(n)}
    ||v(n)||
    ^2\leq L\,.
\end{equation}
The strategy of the proof will be to show that in fact $v(n)$ is bounded with respect to the boundary Hodge norm $||\cdot||_\infty$. Then the desired finiteness follows from the fact that $v(n)$ is integral. We will divide the proof in several steps. 

\subsubsection*{Step 1: Boundedness of $\mathrm{Sl}(2)$-norm}
For the first step of the proof, we would like to translate the bound \eqref{eq:bound_v(n)} into a more detailed statement on the various $\mathrm{sl}(2)$-components of $v(n)$. Indeed, the fact that the Hodge norm of $v(n)$ is bounded implies that also its $\mathrm{Sl}(2)$-norm is bounded. This is reasonable, since one can view the latter as providing the leading approximation to the full Hodge norm. The crucial point is that the latter is also straightforward to evaluate explicitly and allows one to obtain the following bound (after possibly enlarging $L$)
\begin{equation}
\label{eq:boundedness_sl2}    \sum_{\ell}\left[\prod_{i=1}^m\left(\frac{y_i}{y_{i+1}}\right)^{\ell_i}\right]||\hat{v}_{\ell}(n)||_\infty^2\leq L\,,
\end{equation}
where we have introduced the notation
\begin{equation}
    \hat{v}(n) = \mathrm{exp}\left[-\sum_{i=1}^m x_i N_i\right]v(n)\,,
\end{equation}
and we recall that the $\hat{v}_\ell(n)$ denote the weight-components of $\hat{v}(n)$, as defined in \eqref{eq:sl2_decomp_vector}. Since the left-hand side of \eqref{eq:boundedness_sl2} consists of a sum of positive terms, this bound in fact applies for each $\ell$ separately. Because of this it will be natural to prove boundedness of each individual $\hat{v}_\ell$ component. In other words, we have the following
\begin{equation}
    \text{Goal:}\qquad \text{Show that $||\hat{v}_\ell(n)||_\infty$ is bounded, for each $\ell$.}
\end{equation}
Since the axions are assumed to take values on a bounded interval, this immediately implies that also each $||v_\ell(n)||_\infty$ is bounded. 

In the one-variable case ($m=1$) the relation \eqref{eq:boundedness_sl2} yields a natural separation of weight-components into the classes $\ell<0,\,\ell=0,\,\ell>0$, corresponding to fluxes whose Hodge norm tends to zero, stays constant, or grows as one approaches the boundary of the moduli space. In fact, this is the strategy that was used in \cite{Grimm:2020cda,Schnellletter} to prove finiteness for the one-variable case. However, in the multi-variable case such a separation is not available, as the scaling of the various terms in \eqref{eq:boundedness_sl2} highly depends on the exact hierarchy between $y_1,\ldots, y_m$, which in turn is highly path-dependent. For example, even though we do assume that $y_1>y_2>\cdots >y_m$, recall the growth sector \eqref{eq:growth_sector}, it is not necessarily the case that also $y_1>y_2^2$. Indeed, this is one of the main difficulties that were mentioned in section \ref{subsec:finiteness_intro}. In order to tackle the multi-variable case, we proceed in a different way by introducing a \textit{finite} partition of the moduli space $\mathbb{H}^m$ into subsectors on which the scaling behaviour of the various $\ell$-components is under control. Subsequently, the proof will proceed by considering the different types of subsectors individually. 

\subsubsection*{Step 2: Reduction to subsectors}
The key insight is to use the quantization of the fluxes, together with tadpole bound as formulated in \eqref{eq:boundedness_sl2}, to construct the desired partition of the moduli space. First, we note that since the flux $v(n)$ is integral, there exists a constant $\lambda>0$ such that
\begin{equation}
\label{eq:lower_bound_Hodge_infty}
    ||\hat{v}_{\ell}(n)||_\infty^2>\lambda\,,
\end{equation}
for all $n$, whenever $\hat{v}_\ell(n)$ is non-zero. In practice, one can construct $\lambda$ from the smallest eigenvalue of the boundary Hodge norm, and furthermore take $0<\lambda\leq 1$.\footnote{The reason is that, in general, the eigenvalues of the boundary Hodge norm come in pairs $(\lambda_i, \lambda_i^{-1})$, with each $\lambda_i> 0$.} In order to illustrate this rather abstract statement, we have included some basic examples of boundary Hodge norms in Appendix \ref{app:Hodge_norms}, for which one can write down the constant $\lambda$ explicitly. 

We now combine the relation \eqref{eq:lower_bound_Hodge_infty} together with the tadpole bound to define the desired partition of the moduli space into sectors. For each $\ell$, we define a sector in $\mathbb{H}^m$ by
\begin{equation}    R^{\mathrm{heavy}}_{\ell} = \left\{(y_1,\ldots,y_m):\,\prod_{i=1}^m \left(\frac{y_i}{y_{i+1}}\right)^{\ell_i}> \frac{L}{\lambda}\right\}\,,
\end{equation}
and we define another sector $R^{\mathrm{light}}_{\ell}$ as the complement of $R^{\mathrm{heavy}}_{\ell}$. In other words, for each choice of $\ell$, we split up the moduli space into two disjoint pieces. See figure \ref{fig:subsectors} for an illustration of these sectors in the case of two moduli. The motivation for this definition is as follows. Clearly, if the sequence $y(n)$ lies entirely inside the region $R^{\mathrm{heavy}}_{\ell}$ and also $v_\ell(n)$ is non-zero, for some fixed $\ell$, then 
\begin{equation}
    \sum_{\ell'} \left[\prod_{i=1}^m \left(\frac{y_i}{y_{i+1}}\right)^{\ell'_i}\right]||\hat{v}_{\ell'}(n)||^2_\infty\geq \prod_{i=1}^m \left(\frac{y_i}{y_{i+1}}\right)^{\ell_i}||\hat{v}_{\ell}(n)||^2_\infty> \frac{L}{\lambda}\lambda=L\,,
\end{equation}
which is in contradiction with the bound \eqref{eq:boundedness_sl2}. In other words, inside the region $R^{\mathrm{heavy}}_{\ell}$, the weight-$\ell$ component of $v(n)$ would have a Hodge norm that exceeds the tadpole bound. By passing to a subsequence, we may therefore assume that the sequence $y(n)$ lies entirely inside the region $R^{\mathrm{light}}_{\ell}$.

\begin{figure}
    \centering
    \includegraphics[scale=0.6]{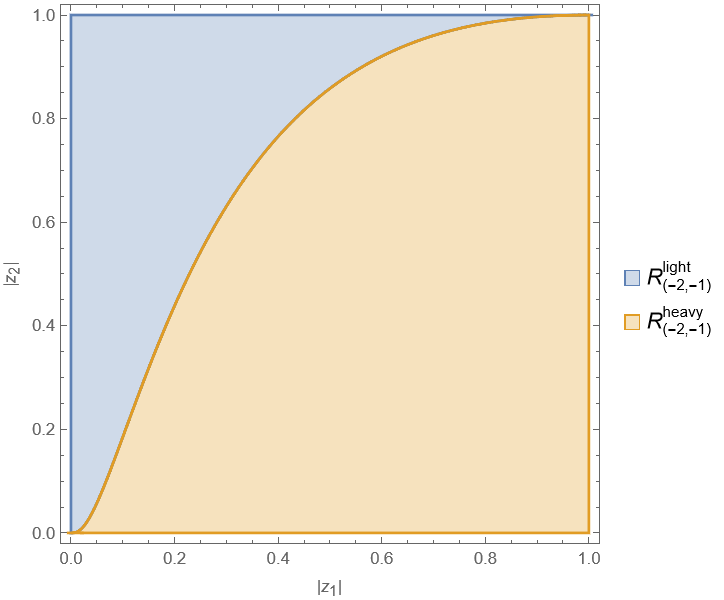}
    \caption{Depiction of the subsectors $R^{\mathrm{light}}_\ell$ (in blue) and $R^{\mathrm{heavy}}_\ell$ (in orange), in terms of the disk coordinates. Here we chose the weight vector to be $\ell=(-2,-1)$ and took $\frac{L}{\lambda}=2$. Note that the growth sector \eqref{eq:growth_sector} translates to the region $|z_1|< |z_2|$.}
    \label{fig:subsectors}
\end{figure}

We now come to the central point of the proof. Since we may assume that $y(n)$ lies entirely inside $R^{\mathrm{light}}_\ell$, we have the upper bound
\begin{equation}
    \prod_{i=1}^m \left(\frac{y_i}{y_{i+1}}\right)^{\ell_i}\leq \frac{L}{\lambda}\,.
\end{equation}
However, since there appears to be no obvious lower bound, the scaling factor that accompanies each $||\hat{v}_\ell(n)||^2_\infty$ in \eqref{eq:boundedness_sl2} could become arbitrarily small. As a result, it appears that some $\hat{v}_\ell(n)$ can be made arbitrarily large, without exceeding the tadpole bound. Note that this is very similar in spirit to the toy example discussed in section \ref{subsec:finiteness_intro}, as well as the example discussed in section \ref{subsec:self-dual_locus_example}. Of course, the missing ingredient that we have not yet exploited is the self-duality condition. To this end, it will be useful to make a further case distinction:
\begin{enumerate}[a.)]
    \item $y(n)\in R^{\mathrm{light}}_{-\ell}$.
    \item $y(n)\in R^{\mathrm{heavy}}_{-\ell}$.
\end{enumerate}
Clearly this covers all possible cases. The motivation for this additional case distinction is that, in order to address the fate of a $\hat{v}_{\ell}$ component, it is actually necessary to consider the $\hat{v}_{-\ell}$ component as well. This is because these two components are related via the self-duality condition. The first case is easy hence we will discuss it first. 

\subsubsection*{Step 3a: Case $y(n)\in R^{\mathrm{light}}_{-\ell}$}
Since $y(n)\in R_{-\ell}^{\mathrm{light}}$ we have that
    \begin{equation}
        \prod_{i=1}^m \left(\frac{y_i}{y_{i+1}}\right)^{-\ell_i}\leq \frac{L}{\lambda}\quad\text{or, equivalently} \quad \prod_{i=1}^m \left(\frac{y_i}{y_{i+1}}\right)^{\ell_i}\geq  \frac{\lambda}{L}\,.
    \end{equation}
In other words, inside the region $R^{\mathrm{light}}_{-\ell}$ we do get a lower bound on the scaling factors. In turn, this provides an upper bound on $||\hat{v}_\ell(n)||_\infty^2$, as is seen explicitly as follows
    \begin{equation}
        L\geq \sum_{\ell'} \prod_{i=1}^m \left(\frac{y_i}{y_{i+1}}\right)^{\ell'_i}||\hat{v}_{\ell'}(n)||^2_\infty\geq \prod_{i=1}^m \left(\frac{y_i}{y_{i+1}}\right)^{\ell_i}||\hat{v}_{\ell}(n)||^2_\infty\geq \frac{\lambda}{L}||v_{\ell}(n)||_\infty^2\,.
    \end{equation}
Hence, one finds the upper bound
    \begin{equation}
    \label{eq:bound_Hodge_case_a}
        ||\hat{v}_{\ell}(n)||_\infty^2 \leq \frac{L^2}{\lambda}\,.
    \end{equation}
Therefore, within the region $R^{\mathrm{light}}_{-\ell}$, we immediately obtain the desired bound on the boundary Hodge norm of $\hat{v}_\ell(n)$, in terms of an `effective tadpole bound' given by the combination $L^2/\lambda$. We stress that this bound depends on both the original tadpole bound $L$, as well as the constant $\lambda$ introduced in \eqref{eq:lower_bound_Hodge_infty}, where the latter reflects the quantization condition of the fluxes and depends on the properties of the boundary Hodge structure. It is also important to note that, in the derivation of the bound \eqref{eq:bound_Hodge_case_a} it has not been necessary to use the self-duality condition. Indeed, in this particular case the resulting finiteness of the fluxes should not be viewed as a property of only the vacuum locus. In particular, it is not obvious whether one could use the refined tadpole bound $L^2/\lambda$ to obtain an accurate estimate for the number of flux vacua in this region of the moduli space. 

\subsubsection*{Step 3b: Case $y(n)\in R^{\mathrm{heavy}}_{-\ell}$ }
This case is more involved and comprises the most difficult part of the whole proof. This is because in this case there is no obvious lower bound for the scaling factor. Instead, the only information at our disposal is that $y(n)$ lies inside $R^{\mathrm{heavy}}_{-\ell}$, which implies that $v_{-\ell}(n)=0$ by the reasoning in step 2. The strategy will be to combine this fact together with an explicit evaluation of the self-duality condition using the nilpotent orbit expansion \eqref{eq:hinv}. Indeed, recalling the definition of the map $h(y_1,\ldots, y_m)$, and using the fact that the boundary Weil operator $C_\infty$ exchanges the $+\ell$ and $-\ell$ eigenspaces, the self-duality condition \eqref{eq:self-duality_nil} can be written as 
\begin{equation}
\label{eq:duality_weights}
    \left(h^{-1}\hat{v}(n)\right)_{\ell}=C_\infty\left(h^{-1}\hat{v}(n)\right)_{-\ell}\,,
\end{equation}
The strategy, then, will be to first evaluate $\left(h^{-1}\hat{v}(n)\right)_{-\ell}$ explicitly to derive its scaling with the moduli and then use the relation \eqref{eq:duality_weights} to infer information about $\left(h^{-1}\hat{v}(n)\right)_\ell$ and subsequently $\hat{v}_{\ell}(n)$ itself. To proceed, we therefore apply the result for $h^{-1}$ stated in equation \eqref{eq:hinv} to find
\begin{equation}
\label{eq:hinv_b}
\left(h^{-1}\hat{v}(n)\right)_{-\ell}=\sum_{k_1,\ldots, k_m=0}^\infty \sum_{s^1,\ldots, s^m}\prod_{i=1}^m \left[\left(\frac{y_i}{y_{i+1}}\right)^{-k_i+\frac{1}{2}s_i^{(m)}}f^{s^i}_{i,k_i}\right]\left(e(y)\hat{v}(n)\right)_{-\ell-s^{(m)}}\,.
\end{equation}
where we recall the notation
\begin{equation}    s^{(m)}=\left(s_1^{(m)},\ldots, s_m^{(m)}\right)\,,\qquad s_i^{(m)}=s_i^i+\cdots +s_i^m\,,
\end{equation}
and note that the second sum in \eqref{eq:hinv_b} runs over all possible values of the $m$ weight vectors $s^j$, see also \eqref{eq:sl2_decomp_operator}. Due to the particular weight properties of the expansion coefficients $f_{i,k_i}$, the sum only runs over $s_i^i\leq k_i-1$ and $s_j^i=0$ for $j>i$. The leading contribution to \eqref{eq:hinv_b} is given by the $k_1,\ldots, k_m=0$ term, and is proportional to $\hat{v}(n)_{-\ell}$. The crucial point is that, because we are considering the case where $y(n)\in R^{\mathrm{heavy}}_{-\ell}$, this leading contribution vanishes. In other words, in order to properly assess the scaling of $\left(h^{-1}\hat{v}(n)\right)_{-\ell}$ it is necessary to understand the scaling of the correction coefficients $f_{i,k_i}$. In particular, we will make use of the result \eqref{eq:f_bounds_new}.

We now come to the main technical computation, namely the estimation of the scaling of the term in brackets in \eqref{eq:hinv_b}. To this end we make two additional observations. 
\begin{enumerate}
    \item First, note that
\begin{equation}
    \left(\frac{y_i}{y_{i+1}}\right)^{-k_i+\frac{1}{2}s_i^i}\prec \left(\frac{y_i}{y_{i+1}}\right)^{-\frac{1}{2}s_i^i}\cdot\begin{cases}
    \left(\frac{y_i}{y_{i+1}}\right)^{-1}\,,&k_i\neq 0\,,\\
    1\,, & k_i=0\,.
    \end{cases}
\end{equation}
This follows from the fact that for $k_i\neq 0$, there is the restriction $s_i^i\leq k_i-1$, while for $k_i=0$ one automatically has $s_i^i=0$. 
\item Second, we note that $\hat{v}_{-\ell-s^{(m)}}(n)$ is only non-zero if the sequence $y(n)$ lies in $R^{\mathrm{light}}_{-\ell-s^{(m)}}$, or, in other words, when the factor
\begin{equation*}    \prod_{i=1}^m\left[\left(\frac{y_i}{y_{i+1}}\right)^{-\left(\ell_i+s_i^{(m)}\right)}\right]
\end{equation*}
is bounded by a constant. In particular, we may apply this to all the terms appearing in \eqref{eq:hinv_b}.
\end{enumerate}

Now suppose, for the moment, that all $k_i$ are non-zero, then combining these observations with the bounds stated in \eqref{eq:f_bounds_new}, we find the following estimate
\begin{align*}    \prod_{i=1}^m\left[\left(\frac{y_i}{y_{i+1}}\right)^{-k_i+\frac{1}{2}s_i^{(m)}}f^{s^i}_{i,k_i}\right]&\stackrel{\text{(a)}}{\prec}  \prod_{i=1}^m\left[\left(\frac{y_i}{y_{i+1}}\right)^{-1-\frac{1}{2}s_i^{i}+\frac{1}{2}(s_i^{i+1}+\cdots +s_i^m)}\prod_{j=1}^{i-1} \left(\frac{y_j}{y_{j+1}}\right)^{-s_j^i}\right]\\
    &\stackrel{\text{(b)}}{=} \prod_{i=1}^m\left[\left(\frac{y_i}{y_{i+1}}\right)^{-1-\frac{1}{2}s_i^{(m)}}\right]\\    &\stackrel{\text{(c)}}{=}\prod_{i=1}^m\left[\left(\frac{y_i}{y_{i+1}}\right)^{\frac{1}{2}\ell_i}\right]\cdot \underbrace{\prod_{i=1}^m\left[\left(\frac{y_i}{y_{i+1}}\right)^{-\frac{1}{2}\left(\ell_i+s_i^{(m)}\right)}\right]}_{\prec 1}\cdot \prod_{i=1}^m\left[\left(\frac{y_i}{y_{i+1}}\right)^{-1}\right]\\
    &\stackrel{\text{(d)}}{\prec} \prod_{i=1}^m\left[\left(\frac{y_i}{y_{i+1}}\right)^{\frac{1}{2}\ell_i}\right]\cdot y_1^{-1}\,.
\end{align*}
To be clear, in step (a) we used \eqref{eq:f_bounds_new} and applied the first observation, in step (b) we simply collected all the terms, in step (c) we expanded the product to uncover the middle term and in step (d) we applied the second observation stating that the middle term is bounded. 

If, in contrast, $k_{i}=0$ for some $i$, the only difference is that the corresponding factor of $\left(y_{i}/y_{i+1}\right)^{-1}$ will not be present, see again the first observation. For example, if $k_1=0$ but all other $k_i$ are non-zero, one will instead get
\begin{equation}
    k_1=0:\qquad \prod_{i=1}^m\left[\left(\frac{y_i}{y_{i+1}}\right)^{-k_i+\frac{1}{2}s_i^{(m)}}f^{s^i}_{i,k_i}\right]\prec \prod_{i=1}^m\left[\left(\frac{y_i}{y_{i+1}}\right)^{\frac{1}{2}\ell_i}\right]\cdot y_2^{-1}\,.
\end{equation}
In particular, the factor of $y_1^{-1}$ is now replaced by a factor of $y_2^{-1}$. A similar thing happens when multiple $k_i$'s are equal to zero. The important point is that one always ends up with some rational factor which goes to zero as all $y_i\rightarrow\infty$.\footnote{Here it is important to recall that sequence $y(n)$ is restricted to lie inside the growth sector \eqref{eq:growth_sector}.} It remains to consider the term in which all $k_i$ are zero. However, as said before, this leading term vanishes since we have assumed $y(n)\in R^{\mathrm{heavy}}_{-\ell}$. To summarize, we have argued that
\begin{equation}
    ||(h^{-1}\hat{v}(n))_{-\ell}||_\infty \prec \prod_{i=1}^m\left[\left(\frac{y_i}{y_{i+1}}\right)^{\frac{1}{2}\ell_i}\right]\cdot\alpha(y_1,\ldots,y_m)\,,
\end{equation}
where $\alpha(y_1.\ldots,y_m)$ is a rational function of $y_1,\ldots, y_m$ that goes to zero as $n\rightarrow\infty$. To complete the argument, we now apply the duality condition \eqref{eq:duality_weights} to find
\begin{equation}
    ||(h^{-1}\hat{v}(n))_{\ell}||_\infty \prec \prod_{i=1}^m\left[\left(\frac{y_i}{y_{i+1}}\right)^{\frac{1}{2}\ell_i}\right]\cdot\alpha(y_1,\ldots,y_m)\,,
\end{equation}
Moving the term in square brackets to the left-hand side and noting that $e(y)h^{-1}\sim 1$, we find the result
\begin{equation}
    ||v_{\ell}(n)||_\infty\prec \alpha(y_1,\ldots,y_m)\,,
\end{equation}
where we have again used the fact that the axions $x_i(n)$ are bounded to remove the hat. In particular, we have shown that the sequence $v_{\ell}(n)$ is bounded with respect to the boundary Hodge norm. In fact, since $v(n)$ is integral, it cannot become arbitrarily small, hence after some finite $n$ we must in fact have that $v_{\ell}(n)=0$.

\subsubsection*{Step 4: Finishing the proof}
Let us collect the results so far. For a fixed $\ell$, we have effectively shown that
\begin{itemize}
    \item $y(n)\in R^{\mathrm{heavy}}_{\ell}$: $v_{\ell}(n) =0$. 
    \item  $y(n)\in R^{\mathrm{light}}_{\ell}\cap R^{\mathrm{light}}_{-\ell}$: $||v_{\pm \ell}(n)||_\infty$ is bounded. More precisely,
    \begin{equation}
        ||\hat{v}_{\pm\ell}(n)||_\infty^2\leq\frac{L^2}{\lambda}\,.
    \end{equation}
    \item  $y(n)\in R^{\mathrm{light}}_{\ell}\cap R^{\mathrm{heavy}}_{-\ell}$: $v_{-\ell}(n)=0$. Furthermore, there exists an $n'$ such that for all $n>n'$ we have that  $v_{\ell}(n)$ vanishes.
\end{itemize}
\begin{figure}
    \centering
    \includegraphics[scale=0.6]{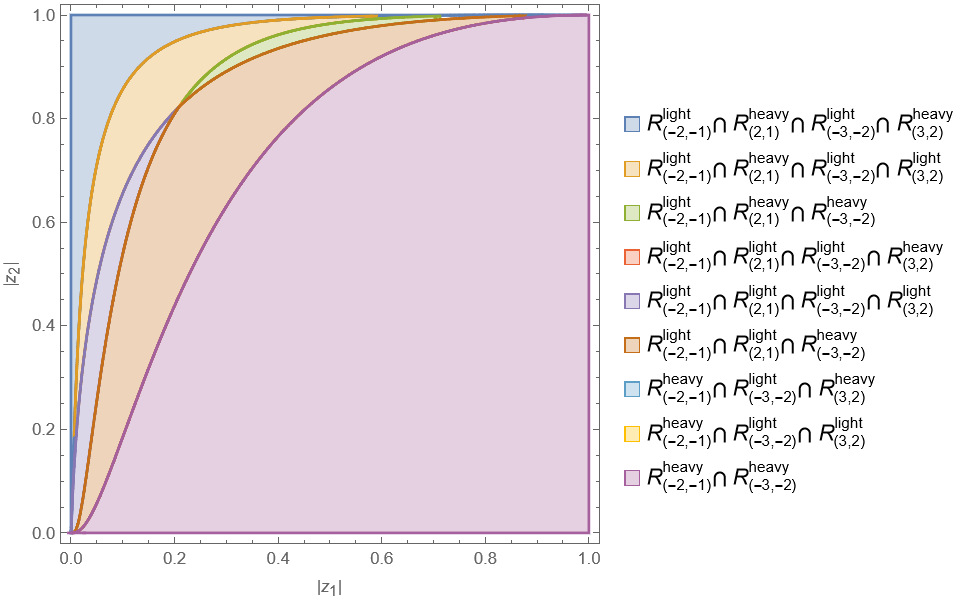}
    \caption{Depiction of the nine disjoint subsectors that arise from taking intersections between $R^{\mathrm{light}}_\ell$ and $R^{\mathrm{heavy}}_\ell$ for various values of $\ell$. Again we have taken $\frac{L}{\lambda}=2$. Note that some of the possible intersections, such as the one in red, cover only a region near small values of $|z_1|$ and $|z_2|$, and are therefore not visible due to the limited resolution. }
    \label{fig:subsectors2}
\end{figure}
This covers all possibilities. Therefore, the sequence $v_{\ell}(n)$ is bounded with respect to the boundary Hodge norm in all sectors. Furthermore, for $n$ sufficiently large, the only way in which it can attain non-zero values is if $y(n)\in R^{\mathrm{light}}_{\ell}\cap R^{\mathrm{light}}_{-\ell}$. One may now simply apply this argument for all possible values of $\ell$, by considering subsequences $y(n)$ lying in all possible intersections of subsectors. For example, one might start with the weight $\ell=(-2,-1)$  and consider the three spaces
\begin{equation}
    R^{\mathrm{light}}_{(-2,-1)}\cap R^{\mathrm{heavy}}_{(2,1)}\,,\qquad R^{\mathrm{light}}_{(-2,-1)}\cap R^{\mathrm{light}}_{(2,1)}\,,\qquad R^{\mathrm{heavy}}_{(-2,-1)}\,,
\end{equation}
which exactly cover the three cases listed above. Then one considers the same three spaces but for $\ell=(-3,-2)$, and constructs all nine pairwise intersections between these spaces. This is illustrated in figure \ref{fig:subsectors2}. One then continues this process ranging over the total number $\#$ of possible values of $\ell$, yielding at most $3^{\#}$ disjoint sectors.\footnote{Note that some intersections may be empty.} Importantly, this procedure always results in a finite partition. Therefore, it suffices to consider a finite number of subsequences of $y(n)$, each lying in a fixed intersection. In this way we conclude that $||v_{\ell}(n)||_\infty$ is bounded for all $\ell$ throughout all sectors. Combining this with the fact that $v(n)$ is integral completes the proof. 

\section{Conclusions}
In this paper we have studied various features of the vacuum locus of low-energy effective theories that originate from compactifications of string theory, focusing on the specific setting of type IIB/F-theory flux compactifications that give rise to four-dimensional $\mathcal{N}=1$ supergravity theories. In this regard, the purpose of our work has been three-fold. First, we have presented a number of known finiteness theorems for the landscape of flux vacua, which apply to the classes of self-dual flux vacua and Hodge vacua. Second, we have provided new insights into the structure of asymptotic Hodge inner products using the methods of asymptotic Hodge theory, and applied these results to obtain an alternative proof for the finiteness of self-dual flux vacua in the multi-variable case, when restricting to the nilpotent orbit approximation. Finally, we have proposed three mathematical conjectures which address finer details of the locus of flux vacua, such as its expected number of connected components, geometric complexity and dimensionality. 

Let us summarize the main technical results of the paper. The first result is a general formula for the nilpotent orbit expansion of the (inverse) period map, see equation \eqref{eq:hinv}, together with a characterization of the scaling of the expansion coefficients $f_{i,k_i}$ in terms of their $\mathfrak{sl}(2)$-weights, see equations \eqref{eq:expansion_coeff_weights} and \eqref{eq:f_bounds_new}. Notably, these results are valid for an arbitrary number of complex structure moduli, and apply to any asymptotic region in the moduli space. Indeed, the expansion is entirely determined by a simple set of boundary data \eqref{eq:boundary_data}, which can be assigned to any such boundary, via a multi-variable generalization of the bulk reconstruction procedure of \cite{CKS}, see also \cite{Grimm:2021ikg}. In this work we have mainly focused on the results that follow from this procedure, and refer the reader to the upcoming doctoral thesis of the second author for the computational details on how one can concretely perform the procedure. As a direct application, we compute the resulting asymptotic expansion for the central charge of D3-particles in type IIB compactifications, see equation \eqref{eq:central_charge_nilpotent}. Our result generalizes the expressions used in \cite{Grimm:2018ohb,Grimm:2018cpv,Bastian:2020egp} beyond the strict asymptotic regime, and is additionally applicable when the asymptotic coupling of the D3-particle to the graviphoton vanishes, see also \cite{Palti:2021ubp} for a related discussion. Another possible future application of the asymptotic expansion of Hodge inner products is a refined classification of possible F-theory scalar potentials, building on the work \cite{Grimm:2019ixq}. In particular, it would be interesting to revisit the analysis of \cite{Grimm:2019ixq} on the asymptotic de Sitter conjecture, as well as the recent work \cite{Calderon-Infante:2022nxb} on the possibility of asymptotic accelerated expansion. More generally, our results present an opportunity to revisit and refine earlier works which have so far only employed the more restrictive $\mathrm{Sl}(2)$-orbit approximation, such as \cite{Grimm:2019wtx,Calderon-Infante:2020dhm,Lanza:2020qmt,Castellano:2021yye,Grana:2022dfw,Grimm:2022ajb}.

The second technical result of this paper is an alternative proof of the finiteness of self-dual flux vacua, in the nilpotent orbit approximation. More precisely, we address the finiteness from a local point of view by showing that no infinite tails of vacua can accumulate in the asymptotic regions of the moduli space. We refer the reader to Theorem \ref{thm:finiteness_selfdual_nilpotent} for the exact statement. The proof relies on the full machinery of the multi-variable bulk reconstruction procedure and can be viewed as a generalization of the one-variable proof of \cite{Grimm:2020cda,Schnellletter}. One of the new complicating features of the multi-variable case is the issue of path dependence, i.e.~the fact that there can be many different hierarchies between the saxions that approach the boundary of the moduli space. We employ a novel strategy to deal with this issue in the nilpotent orbit approximation, by combining the quantization of the fluxes together with the tadpole bound to partition the boundary region of the moduli space into finitely many special subsectors. Together with our results on the nilpotent orbit expansion of generic Hodge inner products, this allows us to control the Hodge norm of the flux in each region separately. It would be interesting to see if this strategy can be applied to study other questions in which the path-dependence plays a crucial role. 

Besides these technical results, we have also proposed three mathematical conjectures in which we formalize some (un)expected features of the locus of self-dual flux vacua. Let us briefly touch upon these conjectures and refer the reader to section \ref{sec:future_questions} for a more in-depth discussion. Conjecture \ref{conjecture:tadpole} deals with the dimensionality of the vacuum locus and is a modified version of the well-known tadpole conjecture of \cite{Bena:2020xrh}, which we believe to be more amenable to a Hodge-theoretic proof. Here we restrict to the class of Hodge vacua and propose that when the dimension of the moduli space is larger than the self-intersection of the flux, up to an $\mathcal{O}(1)$ number, the minimal dimensionality of the locus of Hodge classes becomes nonzero. In other words, there will be at least one unstabilized modulus. It should be stressed, however, that even in this weaker form the conjecture is rather speculative, and it is not at all obvious whether it is true in general.

We additionally propose two new conjectures, both of which address the counting of flux vacua from complimentary perspectives. In Conjecture \ref{conjecture:scaling_Hodge_vacua}, we have proposed that the number of F-theory Hodge vacua, i.e.~self-dual vacua which additionally satisfy $W_{\mathrm{flux}}=0$ exactly, grows subpolynomially in the tadpole bound $L$, when restricting to isolated points in the locus of Hodge tensors. We stress that this is drastically different from the expected polynomial scaling of generic self-dual flux vacua. The conjecture is motivated by the fact that finding such a Hodge vacuum requires one to solve a highly transcendental equation over the integers, for which solutions are expected to be rare. As a first step towards a potential proof of the conjecture, it would be insightful to test it by performing a detailed counting of Hodge vacua in an explicit F-theory compactification, building on the recent works \cite{Cole:2019enn,Tsagkaris:2022apo,Dubey:2023dvu,Ebelt:2023clh,Krippendorf:2023idy,Plauschinn:2023hjw} in the type IIB setting. We hope to address this in the future. Our second Conjecture \ref{conjecture:complexity} can be viewed as a refinement of the result of Theorem \ref{thm:finiteness_self_dual}, which states that the locus of bounded self-dual flux vacua is definable in the o-minimal structure $\mathbb{R}_{\mathrm{an},\mathrm{exp}}$. Indeed, we propose that this locus is actually definable in a much smaller o-minimal structure, which is in fact sharply o-minimal. The latter would imply that the locus has a well-defined notion of geometric complexity as introduced in \cite{binyamini2022}. Additionally, we have made a concrete proposal for how this geometric complexity depends on the number of moduli and the tadpole bound. Using the properties of sharply o-minimal structures, this opens up another path to derive mathematically precise bounds on the counting of flux vacua which we believe to be worthy of future study. It would also be interesting to study the computational complexity of the flux landscape using the properties of sharply o-minimal structures, and to compare with earlier work in this direction \cite{Denef:2006ad,Denef:2017cxt}. As is emphasized by all three of the conjectures, it appears that the locus of self-dual flux vacua is remarkably constrained from a mathematical point of view, especially when focusing on the class of Hodge vacua. It would be fascinating to see if such constraints are reflected in potential phenomenological models. 

\subsubsection*{Acknowledgements}

We would like to thank Michael Douglas, Damian van de Heisteeg, Aroldo Kaplan, Bruno Klingler, Erik Plauschinn, David Prieto, Lorenz Schlechter, Christian Schnell, and Mick van Vliet for useful discussions and comments. This research is supported, in part, by the Dutch Research Council (NWO) via a
Start-Up grant and a Vici grant.

\appendix

\section{Basics of Hodge tensors}\label{app:Hodge-tensors}
In this section we introduce some basic concepts regarding so-called Hodge tensors, as well as the corresponding locus of Hodge tensors, see e.g.~\cite{Moonen:2004,Filippini_2015} for further details and references. We recall that these notions play an important role in the formulation of conjecture \ref{conjecture:scaling_Hodge_vacua}.

Let $H_{\mathbb{Z}}$ be a free abelian group of finite rank, and suppose that 
\begin{equation}\label{eq:Hodge_structure}
    H_{\mathbb{C}} = \bigoplus_{p+q=D} H^{p,q}\,,
\end{equation}
is a Hodge structure of weight $D$ on $H_{\mathbb{Z}}$. There are three basic operations one can perform to construct new Hodge structures from \eqref{eq:Hodge_structure}.
    \begin{itemize}
    \item \textbf{Direct sum:}\\
    Given two Hodge structures $\left(H_{\mathbb{Z}},H^{p,q}\right)$ and $\left(H'_{\mathbb{Z}}, H'^{p,q}\right)$ of the same weight $D$, the direct sum of the two lattices
    \begin{equation}
        H_{\mathbb{Z}}\oplus H'_{\mathbb{Z}}\,,
    \end{equation}
    carries a Hodge structure of weight $D$ given by
    \begin{equation}
        \left(H\oplus H'\right)^{p,q}:=H^{p,q}\oplus H'^{p,q}\,,\qquad p+q=D\,.
    \end{equation}
    \item \textbf{Dual:}\\
    The dual lattice 
    \begin{equation}        H^\vee_{\mathbb{Z}}:=\mathrm{Hom}\left(H_{\mathbb{Z}},\mathbb{Z}\right)\,,
    \end{equation}
    consisting of homomorphisms from the lattice $H_{\mathbb{Z}}$ to $\mathbb{Z}$, carries a Hodge structure of weight $-D$ given by
    \begin{equation}
        \left(H^\vee\right)^{p,q}:= \left(H^{-p,-q}\right)^\vee = \mathrm{Hom}\left(H^{-p,-q},\mathbb{C}\right)\,,\qquad p+q=-D\,.
    \end{equation}
    \item \textbf{Tensor product:}\\
    Given two Hodge structures $\left(H_{\mathbb{Z}},H^{p,q}\right)$ and $\left(H'_{\mathbb{Z}}, H'^{p,q}\right)$ of weights $D$ and $D'$, respectively, the tensor product
    \begin{equation}
        H_{\mathbb{Z}}\otimes H '_{\mathbb{Z}}\,,
    \end{equation}
    naturally carries a Hodge structure of weight $D+D'$ given by
    \begin{equation}
        \left(H\otimes H'\right)^{p'',q''}= \bigoplus_{\substack{p+p'=p''\\ q+q'=q''}} H^{p,q}\otimes H'^{p',q'}\,,\qquad p''+q''=D+D'\,.
    \end{equation}
\end{itemize}
Combining the last two operations one finds that the space 
\begin{equation*}
    H^{\otimes a}\otimes \left(H^\vee\right)^{\otimes b}:=\underbrace{H_{\mathbb{Z}}\otimes\cdots\otimes H_{\mathbb{Z}}}_{\text{$a$ copies}}\otimes \underbrace{H^\vee_{\mathbb{Z}}\otimes\cdots\otimes H^\vee_{\mathbb{Z}}}_{\text{$b$ copies}}
\end{equation*}
of $(a,b)$-tensors on $H_{\mathbb{Z}}$ carries a Hodge structure of weight $(a-b)D$. Note that different values of $a,b$ can give rise to a Hodge structure of the same weight. In particular, for each $k\in\mathbb{Z}$, we can apply the first operation and collect all these Hodge structures into one direct sum
\begin{equation}
    \bigoplus_{\substack{a,b\in\mathbb{N}\\a-b=k}}H^{\otimes a}\otimes \left(H^\vee\right)^{\otimes b}\,,
\end{equation}
which carries a Hodge structure of weight $kD$. Finally, one can consider the formal direct product of all these Hodge structures by summing over $k$, which results in the space
\begin{equation}\label{eq:complete_tensor_product}
    H^{\otimes} := \bigoplus_{k=-\infty}^\infty \bigoplus_{\substack{a,b\in\mathbb{N}\\a-b=k}}H^{\otimes a}\otimes \left(H^\vee\right)^{\otimes b}\,.
\end{equation}
Finally, we come to central objects we wish to study: Hodge tensors. Loosely speaking, a \textbf{Hodge tensor} is a Hodge class in $H^{\otimes}$. More precisely, a type $(p,p)$ Hodge tensor is a Hodge class in the component of $H^{\otimes}$ that carries a Hodge structure of weight $2p$, i.e.~in the component
    \begin{equation}\label{eq:Hodge_tensor_component}
        \bigoplus_{\substack{a,b\in\mathbb{N}\\(a-b)D=2p}} H^{\otimes a}\otimes \left(H^\vee\right)^{\otimes b}\,.
    \end{equation}

\subsubsection*{Examples}
\begin{itemize}
    \item \textbf{Type $\left(\frac{D}{2},\frac{D}{2}\right)$ Hodge tensors}\\
    In the following we assume that $D$ is even. Setting $p=D/2$ in \eqref{eq:Hodge_tensor_component} and working through the definitions, we are searching for $(p,p)$ classes in
    \begin{align}
        \bigoplus_{\substack{a,b\in\mathbb{N}\\a-b=1}} H^{\otimes a}\otimes \left(H^\vee\right)^{\otimes b}&= H\oplus\left[ H\otimes H\otimes H^\vee\right]\oplus\cdots\,.
    \end{align}
    Focusing on the first term on the right-hand side, we see that an example of a $\left(\frac{D}{2},\frac{D}{2}\right)$ tensor is simply a Hodge class in the original Hodge structure $H^{p,q}$. However, the notion of a $\left(\frac{D}{2},\frac{D}{2}\right)$ tensor is more general, as it can also arise from the other summands. As an example, one can also construct a $(p,p)$ tensor via
    \begin{equation}
        \left[ H\otimes H\otimes H^\vee\right]^{p,p} = \left[H^{p+1,p-1}\otimes H^{p-2,p+2}\otimes \left(H^\vee\right)^{-p+1,-p-1}\right]\oplus\cdots\,,
    \end{equation}
    where we have just chosen one of the terms that could appear to illustrate the resulting structure. 
    \item \textbf{Type (0,0) Hodge tensors}\\
    Another illuminating example is given by considering type $(0,0)$ Hodge tensors. Setting $p=0$ in \eqref{eq:Hodge_tensor_component} and working through the definitions, we are searching for $(0,0)$ classes in
     \begin{align}
        \bigoplus_{\substack{a,b\in\mathbb{N}\\a-b=0}} H^{\otimes a}\otimes \left(H^\vee\right)^{\otimes b}&= \left[H\otimes H^\vee\right]\oplus\left[ H\otimes H\otimes H^\vee\otimes H^\vee\right]\oplus\cdots\,.
    \end{align}
    Let us focus on the first term, for which we find
    \begin{align}
        \left(H\otimes H^\vee\right)^{0,0}& = \bigoplus_{\substack{p+p^\vee=0\\q+q^\vee=0}} H^{p,q}\otimes \left(H^\vee\right)^{p^\vee, q^\vee}\,,\\
        &=\bigoplus_{p,q} H^{p,q}\otimes \left(H^{p,q}\right)^{\vee}\,,\\
        &\cong \bigoplus_{p,q} \mathrm{End}\left(H^{p,q}, H^{p,q}\right)\,.
    \end{align}
    In other words, such a type $(0,0)$ Hodge tensor can be interpreted as an endomorphism of the various $H^{p,q}$ spaces, i.e.~a map which preserves the original Hodge structure we started with.
\end{itemize}

\subsubsection*{The locus of Hodge tensors}
The above considerations naturally generalize to the setting of variations of Hodge structure, where the Hodge decomposition varies over a moduli space $\mathcal{M}$. Recall that, given an integral class $v\in H_{\mathbb{Z}}$, it is a non-trivial condition on the moduli whether $v$ is a Hodge class (which may or may not have a solution). Similarly, one might ask which points in the moduli space admit Hodge tensors. To be precise, one should consider the locus of points $z\in\mathcal{M}$ for which the Hodge structure admits more Hodge tensors than the general fibre, see for example \cite{Moonen:2004} for further details. This locus will be referred to as the locus of Hodge tensors. Note that, by the first example discussed above, the locus of Hodge tensors contains the locus of Hodge classes, recall also \eqref{eq:Hodge-locus}.

\section{Properties of $f_{i,k_i}$}
\label{app:additional_proofs}
In this section we elaborate on the properties of the expansion functions $f_{i,k_i}$ appearing in the nilpotent orbit expansion, see for example equation \eqref{eq:hinv}. We focus on the property
\begin{equation}   
\label{eq:f_bounds_app}
f_{i,k_i}^{s^i}\prec \prod_{j=1}^{i-1} \left(\frac{y_j}{y_{j+1}}\right)^{-s_j^i}\,,\qquad 2\leq i \leq m\,,
\end{equation}
see also equation \eqref{eq:f_bounds_new} and the surrounding discussion. We present the proof of \eqref{eq:f_bounds_app} for the case $m=2$ to give the general idea. For arbitrary $m$, the argument will be very similar but becomes more cumbersome to write down. Setting $m=2$, we need to show that
    \begin{equation}
        f^{(s_1^2,s_2^2)}_{2,k_2}\prec \left(\frac{y_1}{y_2}\right)^{-s_1^2}\,.
    \end{equation}
It suffices to consider the case $s_1^2>0$, since it is certainly the case that $f_{2,k_2}\prec 1$. It follows from the general mechanism of the multi-variable bulk reconstruction that $f_2$ is a Lie polynomial of the form
    \begin{equation}
        f_2\left(\frac{y_1}{y_2}\right) = P\left(\mathrm{Ad}\left(h_1\left(\frac{y_1}{y_2}\right)\right)N_{(2)}^+, \delta_{(2)}\right)\,.
    \end{equation}
Importantly, the $y_1/y_2$-dependence of $f_2$ is due entirely to $h_1N^+_{(2)}h_1^{-1}$. It is therefore crucial to understand the scaling of this object. To this end, we recall that, according to the $\mathrm{Sl}(2)$-orbit theorem, we have the relation (see Lemma 4.37 of \cite{CKS})
\begin{equation}
    h_1 N^-_{(2)}h_1^{-1} = \frac{y_1}{y_2}N_1+N_2\,.
\end{equation}
In particular, this implies that $N_{(2)}^-$ can only have weights $0$ and $-2$ with respect to $N^0_{(1)}$. Therefore, since $N^0_{(2)}$ has weight $0$ with respect to $N^0_{(1)}$ (the two commute) it must be that $N_{(2)}^+$ has weights $0$ and $+2$ with respect to $N_{(1)}^0$. In other words, we have the decomposition
\begin{equation}
    h_1 N_{(2)}^+h_1^{-1} = \left(\frac{y_1}{y_2}\right)^{-1}\left(N_{(2)}^+\right)^{(2,2)}+\left(N_{(2)}^+\right)^{(0,2)}+\text{subleading terms}\,.
\end{equation}
Note that the subleading terms can also only have weights $(2,2)$ and $(0,2)$. 
To continue the proof, we recall that the phase operator $\delta_{(2)}$ satisfies 
\begin{equation}
    \delta_{(2)} = \sum_{\tilde{s}_1^ 2}\sum_{\tilde{s}_2^2\leq -2} \delta_{(2)}^{(\tilde{s}_1^2,\tilde{s}_2^2)}\,\qquad [N_{(1)}^-,\delta_{(2)}]=[N_{(2)}^-,\delta_{(2)}]=0\,.
\end{equation}
In words, $\delta_{(2)}$ is a lowest-weight operator, with weight $\tilde{s}_2^2$ at most $-2$. By the enhancement rules for limiting mixed Hodge structures, see e.g.~\cite{Grimm:2018cpv}, this implies that also $\tilde{s}_1^2\leq 0$. Importantly, since $\delta_{(2)}$ is lowest weight, we have that
\begin{equation}
    \mathrm{ad}\left(h_1 N_{(2)}^+h_1^{-1}\right)^s\delta_{(2)}^{(\tilde{s}_1^2,\tilde{s}_2^2)}=0\,,
\end{equation}
whenever $s>\mathrm{min}(\tilde{s}_1^2,\tilde{s}_2^2)$. The argument now proceeds as follows. In order for $f_2$ to have a component with positive weight $s_1^2$, this can only happen due to $\mathrm{ad}\left(N^+_{(2)}\right)^{(2,2)}$ acting on $\delta_{(2)}$, thereby raising the $N^0_{(1)}$ weight by two. However, this goes at the expense of an additional factor $\left(y_1/y_2\right)^{-1}$. Moreover, since $\delta_{(2)}$ is of lowest weight, the most conservative way to get an $s_1^2$ component is by acting exactly $s_1^2$ times with $\mathrm{ad}\left(N^+_{(2)}\right)^{(2,2)}$ on the $-s_1^2$ component of $\delta_{(2)}$ (if it is present). This therefore comes with an additional factor of $\left(y_1/y_2\right)^{-s_1^2}$. This concludes the proof of property \eqref{eq:f_bounds_app} for the case $m=2$. 

\section{Examples of Boundary Hodge Norms}
\label{app:Hodge_norms}
In this section we present two examples of boundary Hodge norms. For simplicity, we restrict to the case of Calabi--Yau threefolds having a single complex structure modulus, and consider the large complex structure point and the conifold point. In particular, we aim to illustrate the relation \eqref{eq:bound_v(n)}. To this end, it is important to express the asymptotic periods in an integral symplectic basis. We choose to follow the conventions of \cite{Bastian:2023shf} and refer the reader to this work for further details. 

\subsubsection*{Type $\mathrm{IV}_1$: LCS point}
Near the large complex structure point, corresponding to a Type $\mathrm{IV}_1$ singularity, the period vector may written as
\begin{equation}
    \mathbf{\Pi} = \begin{pmatrix}
        1\\
        t\\
        \frac{\kappa}{2}t^2-\sigma t+\frac{c_2}{24}\\
        -\frac{\kappa}{6}t^3+\frac{c_2}{24}t-i\chi
    \end{pmatrix}\,,
\end{equation}
where
\begin{itemize}
    \item $\kappa$ is the triple intersection number of the mirror Calabi--Yau,
    \item $c_2$ is the integrated second Chern class,
    \item $\sigma=\frac{\kappa}{2}\,\mathrm{mod}\,1$.
    \end{itemize}
Going through the standard procedure, as described e.g.~in \cite{Grimm:2021ckh}, one can write down the boundary Hodge norm associated to the limiting mixed Hodge structure at $t\rightarrow i\infty$. Explicitly, one finds
\begin{equation}
    S\cdot C_\infty = \left(
\begin{array}{cccc}
 \frac{2 \hat{c}_2^2}{\kappa }+\frac{\kappa }{6} & -\frac{2 \hat{c}_2 \sigma }{\kappa } & -\frac{2 \hat{c}_2}{\kappa } & 0 \\
 -\frac{2 \hat{c}_2 \sigma }{\kappa } &\frac{12 \hat{c}_2^2+\kappa ^2+4 \sigma ^2}{2 \kappa } & \frac{2 \sigma }{\kappa } & -\frac{6
   \hat{c}_2}{\kappa } \\
 -\frac{2 \hat{c}_2}{\kappa } & \frac{2 \sigma }{\kappa } & \frac{2}{\kappa } & 0 \\
 0 & -\frac{6 \hat{c}_2}{\kappa } & 0 & \frac{6}{\kappa } \\
\end{array}
\right)\,,\qquad \hat{c}_2=\frac{c_2}{24}\,.
\end{equation}
Here $S$ denotes the symplectic pairing and $C_\infty$ is the Weil operator associated to the boundary Hodge structure. It is instructive to compute the boundary Hodge norm of the various weight-components of an integral flux. Let $G_3=(g_1,g_2,g_3,g_4)\in H_{\mathbb{Z}}$. Since we are working in an integral basis, this means that $g_1,\ldots, g_4\in\mathbb{Z}$. For the $\mathrm{IV}_1$ singularity, the only possible weights are $\ell=3,1,-1,-3$. By projecting $G_3$ on the various weight eigenspaces, one straightforwardly computes
\begin{align}
    ||(G_3)_{3}||^2_\infty &= \frac{\kappa}{6}g_1^2\,,\\
    ||(G_3)_{1}||^2_\infty &= \frac{\kappa}{2}g_2^2\,,\\
    ||(G_3)_{-1}||^2_\infty &= \frac{2
    }{\kappa}\left(\hat{c}_2 g_1-g_3+\sigma g_2\right)^2\,,\\
    ||(G_3)_{-3}||^2_\infty &= \frac{6}{\kappa}\left(\hat{c}_2 g_2-g_4\right)^2\,.
\end{align}
One sees that, for example, the boundary Hodge norm of a non-zero $(G_3)_3$ is bounded from below by $\kappa/6$. On the other hand, for the $\ell=-1,-3$ components the exact bound will depend on the values of the coefficients $c_2$ and $\sigma$.

\subsubsection*{Type $\mathrm{I}_1$: conifold point}
Around the conifold point, corresponding to a Type $\mathrm{I}_1$ singularity, the period vector may be expressed as
\begin{equation}
    \mathbf{\Pi} = \begin{pmatrix}
        1\\
        0\\
        \delta-\gamma\tau\\
        \tau
    \end{pmatrix}
    +A_1 e^{2\pi i t}\begin{pmatrix}
        \gamma\\
        1\\
        t+\frac{k}{2\pi i}\\
        \delta
    \end{pmatrix}
    +e^{4\pi i t}\begin{pmatrix}
        \gamma A_2 +\frac{k}{8\pi \tau_2}A_1^2\\
        A_2\\
        A_2\left(t+\frac{k}{4\pi i}\right)+\left(\delta-\gamma \Bar{\tau}\right)A_1^2\\
        \delta A_2 +\frac{k\Bar{\tau}}{8\pi \tau_2}A_1^2
    \end{pmatrix}\,,
\end{equation}
where
\begin{itemize}
    \item $\tau=\tau_1+i\tau_2$, with $\tau_2>0$, corresponds to a rigid period,
    \item $\gamma,\delta\in\mathbb{R}$ are the extension data,
    \item $k>0$ corresponds to the order of a finite subgroup that quotients the three-sphere $S^3$,
    \item $A_1,A_2\in\mathbb{C}$ are coefficients that parametrize the exponential corrections, these will not play a role in the following discussion. 
\end{itemize}
As for the case of the LCS point, one can write down the boundary Hodge norm associated to the limiting mixed Hodge structure at $t\rightarrow i\infty$. Explicitly, one finds
\begin{equation}
    S\cdot C_\infty = \left(
\begin{array}{cccc}
 \frac{|\tau|^2}{\tau_2}+\frac{\delta ^2}{k} & \frac{\delta  \tau_1-\gamma  |\tau|^2}{\tau_2} & -\frac{\delta }{k} &
   -\frac{\tau_1}{\tau_2}-\frac{\gamma  \delta }{k} \\
 \frac{\delta \tau_1-\gamma  |\tau|^2}{\tau_2} & \frac{\gamma ^2 |\tau|^2+\delta ^2+k\tau_2-2 \gamma  \delta 
   \tau_1}{\tau_2} & 0 & \frac{\gamma  \tau_1-\delta }{\tau_2} \\
 -\frac{\delta }{k} & 0 & \frac{1}{k} & \frac{\gamma }{k} \\
 -\frac{\tau_1}{\tau_2}-\frac{\gamma  \delta }{k} & \frac{\gamma \tau_1-\delta }{\tau_2} & \frac{\gamma }{k} &
   \frac{1}{\tau_2}+\frac{\gamma ^2}{k} \\
\end{array}
\right)
\end{equation}
Let $G_3=(g_1,g_2,g_3,g_4)\in H_{\mathbb{Z}}$ again be an integral flux, with $g_1,\ldots, g_4\in\mathbb{Z}$. For the $\mathrm{I}_1$ singularity, the only possible weights are $\ell=1,0,-1$. By projecting $G_3$ on the various weight eigenspaces, one straightforwardly computes
\begin{align}
    ||(G_3)_{1}||^2_\infty &= k g_2^2\,,\\
    ||(G_3)_{0}||^2_\infty &=\frac{1}{\tau_2}\left|g_4-\delta g_2 -(g_1-\gamma g_2)\tau\right|^2\,,\\
    ||(G_3)_{-1}||^2_\infty &=\frac{1}{k}\left(g_3+\gamma g_4-\delta g_1\right)^2
\end{align}
For example, we see that if the $\ell=1$ component of $G_3$ is non-zero, then its boundary Hodge norm is bounded by $k$ (since $g_2^2\geq 1$). For the other components the exact bound will depend on the details of the geometry, namely the values of $\gamma,\delta,\tau$. To avoid possible confusion, we stress that for example
\begin{equation}
    (G_3)_{-1} = (0,0,g_3+\gamma g_4-\delta g_1,0)\,,
\end{equation}
so that indeed $||(G_3)_{-1}||_\infty=0$ if and only if $(G_3)_{-1}=0$. 

\section{The Hodge Locus}
\label{app:Hodge_locus}

In this section we will describe the proof of Theorem \ref{thm:finiteness_Hodge_loci_local} in some detail. In contrast to the proof of Theorem \ref{thm:finiteness_selfdual_nilpotent}, which relied heavily on the nilpotent orbit expansion, the proof discussed here relies more on understanding in which space a sequence of Hodge classes ends up when approaching the boundary. To describe this properly, we first recall some basics of the theory of mixed Hodge structures in section \ref{app:MHS}. Here our discussion will be brief, and is mostly intended to set the notation. For a more detailed introduction we refer the reader to \cite{Grimm:2018cpv,Grimm:2021ckh}. Subsequently, the proof of Theorem \ref{thm:finiteness_Hodge_loci_local} will be presented in section \ref{app:finiteness_Hodge}.

\subsection{Mixed Hodge structures}
\label{app:MHS}
We start by recalling the result of the nilpotent orbit theorem, which states that for sufficiently large $\mathrm{Im}\,t^i$, for $1\leq i \leq m$, the variation of Hodge structure in question can be approximation by a nilpotent orbit
\begin{equation}
    F^p\approx F_{\mathrm{nil}}^p = e^{t^i N_i}F_{(m)}^p\,.
\end{equation}
Here we have changed notation slightly by denoting the limiting filtration by $F^p_{(m)}$. It is important to note that this limiting filtration generically does not define a Hodge filtration. However, together with the log-monodromy operators $N_i$ it does define a so-called mixed Hodge structure, which we now introduce.

\subsubsection*{Weight filtration}
For each $1\leq i \leq m$, we denote
\begin{equation}
    N_{(i)} = N_1+\cdots +N_{i}\,,
\end{equation}
and define the monodromy weight filtration $W^{(i)}_\ell$ by
\begin{equation}
\label{eq:def_monodromy_weight_filtration}
    W^{(i)}_{\ell} = \sum_{j\leq \mathrm{max}(-1,\ell-D)} \mathrm{Ker}\,N_{(i)}^{j+1}\cap \mathrm{Im}\,N_{(i)}^{j-\ell+D}\,.
\end{equation}
This defines an increasing filtration satisfying the following two properties
\begin{enumerate}
    \item $N_{(i)}W^{(i)}_{\ell}\subseteq W^{(i)}_{\ell-2}$\,,
    \item $N_{(i)}^\ell : \mathrm{Gr}_{D+\ell}^{(i)}\rightarrow \mathrm{Gr}_{D-\ell}^{(i)}$ is an isomorphism for all $\ell\geq 0$, where we have defined the graded pieces $\mathrm{Gr}_{\ell}^{(i)}$ as the following quotients
    \begin{equation}
        \mathrm{Gr}_{\ell}^{(i)}:= W^{(i)}_\ell / W^{(i)}_{\ell-1}\,.
    \end{equation}
\end{enumerate}

\subsubsection*{Mixed Hodge structure and Deligne splitting}
It turns out that the weight filtration $W_\ell^{(m)}$ and the limiting filtration $F^p_{(m)}$ are such that each graded piece $\mathrm{Gr}^{(m)}_{\ell}$ admits a Hodge structure of weight $\ell$, for each $\ell\geq 0$. More precisely, the decomposition is given by
\begin{equation}
    \mathrm{Gr}_{\ell}^{(m)} = \bigoplus_{p+q=\ell} \left[\mathrm{Gr}_{\ell}^{(m)}\right]^{p,q}\,,
\end{equation}
where the individual $(p,q)$ pieces can be computed explicitly from appropriate intersections and quotients of $F^p_{(m)}$ and $W^{(m)}_\ell$, see for example \cite{Grimm:2018cpv}. For this reason, it is said that the pair $(W^{(m)},F_{(m)})$ defines a \textit{mixed Hodge structure}. 

Another useful way to collect the various constituents of a mixed Hodge structure is through the so-called Deligne splitting. It is constructed out of the weight filtration $W^{(m)}$ and limiting filtration $F_{(m)}$ as
\begin{equation}
\label{eq:def_Deligne}
    I^{p,q}_{(m)} := F_{(m)}^p\cap W^{(m)}_{p+q}\cap \big(\bar{F}_{(m)}^q\cap W^{(m)}_{p+q}+\sum_{j\geq 1}\bar{F}_{(m)}^{q-j}\cap W^{(m)}_{p+q-j-1} \big)\,,
\end{equation}
for $p,q=0,\ldots, D$. Conversely, one can recover the weight filtration as well as the limiting filtration from the Deligne splitting via the relations
\begin{equation}    W^{(m)}_{\ell}=\bigoplus_{p+q=\ell}I_{(m)}^{p,q}\,,\qquad   F^p_{(m)} = \bigoplus_{r\geq p}\bigoplus_s I_{(m)}^{r,s}\,.
\end{equation}

\subsubsection*{$\mathbb{R}$-split MHS and $\mathfrak{sl}(2,\mathbb{R})$}
Generically, the Deligne splitting $I^{p,q}_{(m)}$ does not behave nicely under complex conjugation. Rather, it satisfies the relation
\begin{equation}
    \overline{I^{p,q}_{(m)}} = I^{q,p}_{(m)}\quad\mathrm{mod}\quad \bigoplus_{r<q,s<p}I^{r,s}_{(m)}\,.
\end{equation}
When the splitting satisfies the much simpler relation $\overline{I^{p,q}_{(m)}} = I^{q,p}_{(m)}$ we say that the mixed Hodge structure is $\mathbb{R}$-split. It was shown by Deligne that, given a mixed Hodge structure $(W^{(m)}, F_{(m)})$, it is always possible to find two $\mathfrak{g}_{\mathbb{R}}$-valued operators $\zeta_{(m)},\delta_{(m)}$ such that the filtration
\begin{equation}
    \tilde{F}^p_{(m)}:= e^{\zeta_{(m)}}e^{i\delta_{(m)}} F^p_{(m)}
\end{equation}
defines a mixed Hodge structure $(W^{(m)}, \tilde{F}_{(m)})$ which is $\mathbb{R}$-split, and we denote its associated Deligne splitting by $\tilde{I}^{p,q}_{(m)}$. It will be useful to note that, due to the $\mathbb{R}$-split property, the relation \eqref{eq:def_Deligne} simplifies to
\begin{equation}
\label{eq:relation_Deligne_Rsplit}
    \tilde{I}^{p,q}_{(m)} = \tilde{F}^p_{(m)}\cap \overline{\tilde{F}^q_{(m)}}\cap W^{(m)}_{p+q}\,.
\end{equation}
A crucial property of $\mathbb{R}$-split mixed Hodge structures is that one can naturally define a real $\mathfrak{sl}(2,\mathbb{R})$-triple associated to them in terms of the Deligne splitting. Indeed, the weight operator $N^0_{(m)}$ is defined by the relation
\begin{equation}
\label{eq:def_weight_operator}
    N^0_{(m)}v = (p+q-D)\,,\qquad v\in \tilde{I}^{p,q}_{(m)}\,,
\end{equation}
while the lowering operator $N_{(m)}^-$ is obtained by an appropriate projection of the log-monodromy matrix $N_{(m)}$ as explained in \cite{Grimm:2021ckh}. Lastly, the raising operator $N_{(m)}^+$ is uniquely determined by solving the commutation relations. Note that there is a straightforward relation between the weight components of a vector $v$, recall the decomposition \eqref{eq:sl2_decomp_vector}, and its position in the weight filtration. To be precise, If $v$ has highest $N^0_{(m)}$ weight $\ell_{m}$, then it lies inside the $W^{(m)}_{D+\ell_m}$ component of the weight filtration.

\subsubsection*{Constructing the rest of the boundary data}
The operator $\delta_{(m)}$ and the $\mathfrak{sl}(2,\mathbb{R})$-triple $\{N^+_{(m)}, N^0_{(m)}, N^-_{(m)}\}$ comprise  part of the boundary data \eqref{eq:boundary_data}. 
In order to construct the rest of the boundary data, one effectively performs the above procedure inductively as follows. First, one can show that the filtration
\begin{equation}
    F^p_{(m-1)}:=e^{iN_m}\tilde{F}^p_{(m)}\,,
\end{equation}
together with the monodromy weight filtration $W_\ell^{(m-1)}$ again defines a mixed Hodge structure, with an associated Deligne splitting $I^{p,q}_{(m-1)}$ defined by the same relation \eqref{eq:def_Deligne} but with $m$ replaced by $m-1$. If this mixed Hodge structure is not $\mathbb{R}$-split, then one may again find operators $\zeta_{(m-1)}$ and $\delta_{(m-1)}$ such that the filtration 
\begin{equation}
    \tilde{F}^p_{(m-1)}:=e^{\zeta_{(m-1)}}e^{i\delta_{(m-1)}}F^p_{(m-1)}\,,
\end{equation}
together with $W_\ell^{(m-1)}$ defines an $\mathbb{R}$-split mixed Hodge structure $\tilde{I}^{p,q}_{(m-1)}$. This, in turn, also has an associated $\mathfrak{sl}(2,\mathbb{R})$ triple. Proceeding inductively, for each $1\leq k\leq m$ one defines 
\begin{equation}
\label{eq:filtrations_recursive}
    F^p_{(k-1)} = e^{i N_k}\tilde{F}^p_{(k)}\,,\qquad \tilde{F}^p_{(k)} = e^{\zeta_{(k)}}e^{i\delta_{(k)}}F^p_{(k)}\,,
\end{equation}
such that each pair $(W^{(i)},\tilde{F}_{(i)})$ defines an $\mathbb{R}$-split mixed Hodge structure, with associated Deligne splitting $\tilde{I}^{p,q}_{(i)}$ and corresponding $\mathfrak{sl}(2,\mathbb{R})$-triple. To summarize, this provides us with the following data
\begin{itemize}
    \item phase operators: $\delta_{(i)}$,
    \item $\mathfrak{sl}(2,\mathbb{R})$-triples: $\{N^+_{(i)}, N^0_{(i)}, N^+_{(i)}\}$,
    \item $\mathbb{R}$-split Deligne splittings: $\tilde{I}^{p,q}_{(i)}$,
    \item graded spaces: $\mathrm{Gr}^{(i)}_{\ell}$,
\end{itemize}
for each $1\leq i \leq m$. 

\subsubsection*{Boundary Hodge structure}
Finally, setting $k=1$ in \eqref{eq:filtrations_recursive} one obtains the boundary Hodge filtration
\begin{equation}
    F^p_{\infty}:=F^p_{(0)}\,,
\end{equation}
which is fact turns out to define a polarized Hodge structure. This comprises the last piece of the boundary data \eqref{eq:boundary_data}.

\subsection{Proof: finiteness of Hodge vacua}
\label{app:finiteness_Hodge}
Having discussed the additional necessary ingredients of mixed Hodge theory, we now turn to a description of the proof of finiteness of Hodge vacua, recall Theorem \ref{thm:finiteness_Hodge_loci_local}. Hence, we let
\begin{equation}
        \left(t^i(n), v(n)\right)\in \mathbb{H}^m\times H_{\mathbb{Z}}(L)\,,
\end{equation}
be a sequence of points such that $\mathrm{Re}\,t^i(n)$ is bounded and $\mathrm{Im}\,t^i(n)\rightarrow\infty$ as $n\rightarrow\infty$.
Furthermore, we assume that
\begin{equation}
        v(n)\in  H^{k,k}\,,\qquad D=2k\,,
    \end{equation}
is a sequence of Hodge classes. Our goal is to show that $v(n)$ can only take on finitely many values. To this end, it suffices to show that there exists a constant subsequence of $v(n)$. For this reason, we may and will freely pass to a subsequence of $v(n)$ whenever possible, without changing the notation to avoid unnecessary cluttering. To ease the reader into the proof, we start by considering the simpler one-variable case ($m=1$). Afterward, we explain how one may inductively apply the one-variable result to obtain the general result. 

\subsubsection{One variable}
\label{subsubsec:proof_Hodge_locus_single}
The proof will be divided into five main steps. 

\subsubsection*{Step 1: Incorporating exponential corrections}
It will be helpful to effectively reduce the problem to the case where the variation of Hodge structure in question is a nilpotent orbit. This can be done as follows. As explained in \cite{CDK,schnell2014extended}, one may assume, without loss of generality, that the flux is parametrized as
\begin{equation}
    v(n) = v_{\mathrm{nil}}(n)+v_{\mathrm{inst}}(n)\,,
\end{equation}
where $v_{\mathrm{nil}}(n)\in F^k_{\mathrm{nil}}$ is a sequence of Hodge classes with respect to a nilpotent orbit $F_\mathrm{nil}$, while $v_{\mathrm{inst}}(n)$ is a series of exponentially small corrections, satisfying
\begin{equation}
    \frac{||v_{\mathrm{inst}}(n)||_\infty}{||v(n)||_\infty}\sim e^{-\alpha y(n)}\,,
\end{equation}
for some constant $\alpha$. 

\subsubsection*{Step 2: Boundedness and $\mathfrak{sl}(2)$-weights}
The start of the proof is identical to the discussion in section \ref{subsec:self-dual_proof}. Indeed, we recall the important result that the self-duality condition and the tadpole cancellation condition together imply that the Hodge norm $||v(n)||_\infty$ is bounded. Consequently, the relation \eqref{eq:Fsharp_F_relation_limit} implies that also the boundary Hodge norm $||e(n)v(n)||$ is bounded. In the one-variable case, this means that
\begin{equation}
    \sum_{\ell}y^\ell \,||\hat{v}_\ell(n)||^2_\infty < L\,,\qquad \hat{v}(n) = e^{-x(n)N}v(n)\,.
\end{equation}
As discussed in section \ref{subsec:self-dual_locus_example}, this implies that $\hat{v}(n)$, for $n$ sufficiently large, can only have non-zero weight components for $\ell\leq 0$. This means that, after passing to a subsequence, the sequence $\hat{v}(n)$ lies inside the $W_{2k}$ component of the weight filtration induced by the single monodromy operator $N$ via \eqref{eq:def_monodromy_weight_filtration}. Additionally, note that the $\ell=0$ component $\hat{v}_0(n)$ can only take finitely many values, due to the quantization condition. Therefore, after passing to another subsequence, we may write
\begin{equation}
\label{eq:hatv_limit}
    \hat{v}(n) \equiv \hat{v}_0\quad \mathrm{mod}\,W_{2k-1}\,,
\end{equation}
where $\hat{v}_0$ is a constant. Intuitively, one can think of the $\hat{v}_0$ component of $\hat{v}(n)$ as the part of $\hat{v}(n)$ that remains in the limit $n\rightarrow\infty$. The next step of the proof amounts to showing that this component is rather restricted, owing to the fact that $\hat{v}(n)$ is a sequence of Hodge classes.

\subsubsection*{Step 3: Show that $N\hat{v}_0=0$}
For the moment, let us denote
\begin{equation}
    w = \lim_{n\rightarrow\infty} e(n)v(n)=\lim_{n\rightarrow\infty} e(n)v_{\mathrm{nil}}(n)\,,
\end{equation}
where the second relation follows from the fact that the contribution from $v_{\mathrm{inst}}(n)$ is sub-leading. Recalling the relation \eqref{eq:Fsharp_F_relation_limit} and using the fact that $v_{\mathrm{nil}}(n)\in F_{\mathrm{nil}}^k$, one finds that $w\in F^k_\infty$. At the same time, since $v(n)\in W_{2k}$ and $v(n)$ is real, the same is true for $w$. This is because $e(n)$ is a real operator that does not change the weights. In other words, we have
\begin{equation}
    w\in F^k_\infty\cap W_{2k}\cap H_{\mathbb{R}}\,.
\end{equation}
Such elements are very restricted, as captured by the following
\begin{lem}[{{\cite[Lemma 4.4]{CDK}}}]
\label{lem:Fsharp_W0}
    If $w\in F^k_{\infty}\cap W_{2k}(N)\cap H_{\mathbb{R}}$, then $N^0 w = Nw=0$.
\end{lem}
\begin{proof}
We follow the proof of Schnell, see Lemma 12.4 of \cite{schnell2014extended}. We proceed in two steps:
\begin{enumerate}
\item First, we show that $Nw=0$. Since $N$ acts on the weight filtration as $N W_{2k}\subseteq W_{2k-2}$ and $F_\infty = e^{iN}\tilde{F}$, with $\tilde{F}$ the limiting filtration, one finds that
\begin{equation}
    e^{-iN}w\in \tilde{F}^k\cap W_{2k}(N)\,.
\end{equation}
In particular, it follows that both $w$ and $Nw$ lie inside $\tilde{F}^{k}\cap W_{2k}(N)$. At the same time, however, noting that $N \tilde{F}^{k}\subseteq \tilde{F}^{k-1}$, it follows that $Nw\in \tilde{F}^{k-1}\cap W_{2k-2}(N)$. Combining these results, together with the fact that $w$ and $N$ are real, gives the condition
\begin{equation}
    Nw \in \tilde{F}^{k}\cap W_{2k-2}(N)\cap H_{\mathbb{R}}\,.
\end{equation}
Using the fact that $(W,\tilde{F})$ defines a mixed Hodge structure, the space on the right-hand side is empty, hence $Nw=0$. To give some feeling for this property, we have illustrated the relevant spaces in the case of a weight $D=4$ limiting mixed Hodge structure (so $k=2$) in figure \ref{fig:Deligne}.
\item Next, we show that $N^0w=0$. This follows straightforwardly from the fact that
\begin{equation}
    W_{2k}\cap \tilde{F}^k\cap \overline{\tilde{F}^k}= \tilde{I}^{k,k}\,,
\end{equation}
where $\tilde{I}^{k,k}$ denotes the $(k,k)$-component of the Deligne splitting associated to the $\mathbb{R}$-split mixed Hodge structure $(W,\tilde{F})$, recall also \eqref{eq:relation_Deligne_Rsplit} and the surrounding discussion. In particular, using the earlier result that $Nw=0$, one finds that $w\in \tilde{I}^{k,k}$. Using the fact that $N^0$ acts on $\tilde{I}^{p,q}$ as multiplication by $p+q-D$, recall equation \eqref{eq:def_weight_operator}, the result follows. 
\end{enumerate}
\end{proof}
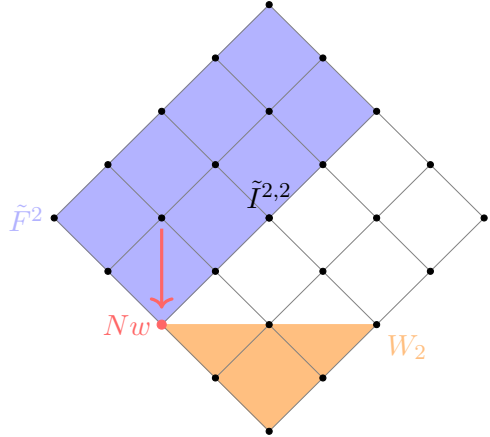
\begin{figure}[h!]
\centering
\begin{tikzpicture}[baseline={([yshift=-.5ex]current bounding box.center)},scale=1,cm={cos(45),sin(45),-sin(45),cos(45),(15,0)}]

\filldraw[color=orange!50] (2,0) node[below right]{$W_{2}$} -- (0,2) -- (0,0);
\filldraw[color=blue!30] (4,4) -- (0,4) node[left]{$\tilde{F}^2$}-- (0,2) -- (4,2);

\draw[step = 1, gray, ultra thin] (0, 0) grid (4, 4);

  \foreach \i\j in {0/0, 0/1, 0/2, 0/3, 0/4,
1/0, 1/1, 1/2, 1/3, 1/4,
2/0, 2/1, 2/2, 2/3, 2/4,
3/0, 3/1, 3/2, 3/3, 3/4,
4/0, 4/1, 4/2, 4/3, 4/4}
{
  \draw[fill] (\i,\j) circle[radius=0.04] ;
}

\draw[fill] (2,2) circle[radius=0.04] node[above]{$\tilde{I}^{2,2}$};

\draw[fill, color=red!60] (0,2) circle[radius=0.06] node[left]{$Nw$};
\draw[->,color=red!60, very thick] (0.9,2.9) -- (0.15,2.15); 

\end{tikzpicture}
\caption{Arrangement of the Deligne splitting $\tilde{I}^{p,q}$ for a weight four limiting mixed Hodge structure $(W,\tilde{F})$. In blue: $\tilde{F}^2$ component of the limiting Hodge filtration. In orange: $W_{2}$ component of the weight filtration. In red: the potential location of $Nw$, with the arrow denoting the action of the log-monodromy matrix $N$. Since complex conjugation acts on the Deligne splitting as reflection in the vertical axis, the vector $Nw$ cannot be real unless it is zero. Following the proof of Lemma \ref{lem:Fsharp_W0}, the only possible location for $w$ is in the space $\tilde{I}^{2,2}$.
\label{fig:Deligne}}
\end{figure}

As a result of Lemma \ref{lem:Fsharp_W0}, we find that 
\begin{equation}
    N^0w = 0\,,\qquad Nw = 0\,.
\end{equation}
One can summarize this result in the statement that $w$ should be a singlet under the $\mathfrak{sl}(2)$ action. Returning to our original sequence $\hat{v}(n)$, recall equation \eqref{eq:hatv_limit}, and projecting the congruence
\begin{equation}
    e(n)\hat{v}(n)\equiv \hat{v}_0\quad \mathrm{mod}\,W_{2k-1} 
\end{equation}
to the eigenspace $E_0$, one finds that indeed $w=\hat{v}_0$ and hence $N\hat{v}_0=0$. Additionally, as mentioned in the proof of Lemma \ref{lem:Fsharp_W0}, the limiting element $\hat{v}_0$ lies inside the space $\tilde{I}^{k,k}$. In other words, it lies exactly in the center of the Deligne diamond. 

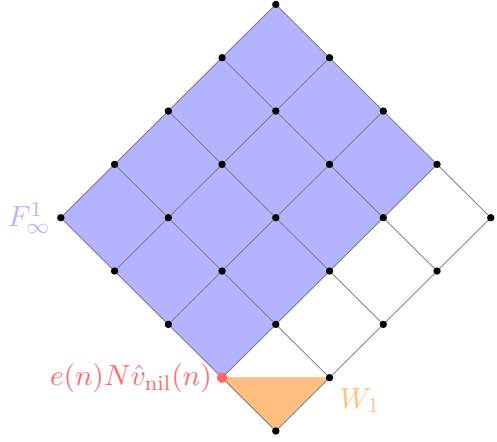
\begin{figure}[t!]
\centering
\begin{tikzpicture}[baseline={([yshift=-.5ex]current bounding box.center)},scale=1,cm={cos(45),sin(45),-sin(45),cos(45),(15,0)}]

\filldraw[color=orange!50] (1,0) node[below right]{$W_{1}$} -- (0,1) -- (0,0);
\filldraw[color=blue!30] (4,4) -- (0,4) node[left]{$F_\infty^1$}-- (0,1) -- (4,1);

\draw[step = 1, gray, ultra thin] (0, 0) grid (4, 4);

  \foreach \i\j in {0/0, 0/1, 0/2, 0/3, 0/4,
1/0, 1/1, 1/2, 1/3, 1/4,
2/0, 2/1, 2/2, 2/3, 2/4,
3/0, 3/1, 3/2, 3/3, 3/4,
4/0, 4/1, 4/2, 4/3, 4/4}
{
  \draw[fill] (\i,\j) circle[radius=0.04] ;
}

\draw[fill, color=red!60] (0,1) circle[radius=0.06] node[left]{$e(n)N\hat{v}_{\mathrm{nil}}(n)$};

\end{tikzpicture}
\caption{Arrangement of the Deligne splitting $I^{p,q}$ for a weight four limiting mixed Hodge structure $(W,F_\infty)$. In blue: $F_\infty^1$ component of the boundary Hodge filtration. In orange: $W_{1}$ component of the weight filtration. In red: the potential location of the limit of the sequence $e(n)N\hat{v}_{\mathrm{nil}}(n)$ as $n\rightarrow\infty$. Since complex conjugation acts on the Deligne splitting as reflection in the vertical axis, the vector $e(n)N\hat{v}_{\mathrm{nil}}(n)$ cannot be real unless it is zero. 
\label{fig:Deligne2}}
\end{figure}
\subsubsection*{Step 4: Show that $N\hat{v}(n)=0$}
Having shown that $N\hat{v}_0=0$, the next step is to show that in fact this relation applies to the whole sequence $\hat{v}(n)$, after passing to a subsequence. Indeed, since $N\hat{v}_0=0$ and $N W_{2k-1}\subseteq W_{2k-3}$, it follows that 
\begin{equation}
    N\hat{v}(n)\in W_{2k-3}\,.
\end{equation}
Now, let us write
\begin{equation}
    e(n)N\hat{v}(n)=e(n)N\hat{v}_{\mathrm{nil}}(n)+e(n)N\hat{v}_{\mathrm{inst}}(n)\,.
\end{equation}
Our goal will be to first show that both terms on the right-hand side are in fact exponentially small. Indeed, suppose not, then the ratio
\begin{equation}
    \frac{||e(n)N\hat{v}_{\mathrm{inst}}(n)||_\infty}{||e(n)\hat{v}(n)||_\infty}
\end{equation}
would go to zero as $n\rightarrow\infty$. Consequently, one would find that
\begin{equation}
    \lim_{n\rightarrow\infty}\frac{e(n)N\hat{v}(n)}{||e(n)N\hat{v}(n)||_\infty}=\lim_{n\rightarrow\infty}\frac{e(n)N\hat{v}_{\mathrm{nil}}(n)}{||e(n)N\hat{v}_{\mathrm{nil}}(n)||_\infty}\,,
\end{equation}
and furthermore the resulting limit would give a unit vector in the space $NF^k_\infty\cap W_{2k-3} \cap H_{\mathbb{R}}$. To see the latter, one uses the fact that $e(n)N = y(n)^{-1}Ne(n)$ to swap the order of $e(n)$ and $N$ (note that these are vector space identities, so an overall rescaling is irrelevant) and again applies the relation \eqref{eq:Fsharp_F_relation_limit}. The important observation is that $NF^k_\infty\cap W_{2k-3} \cap H_{\mathbb{R}}=\{0\}$, due to the fact that $(W,F_\infty)$ defines a mixed Hodge structure (recall also that $N F^k_\infty\subseteq F^{k-1}_\infty$). This is again illustrated in figure \ref{fig:Deligne2} for the case $D=4$, or $k=2$. We have therefore arrived at a contradiction and must conclude that in fact $||e(n)\hat{v}(n)||_\infty$ is bounded by a constant multiple of $||e(n)N\hat{v}_{\mathrm{inst}}(n)||_\infty$, in particular it is exponentially small. Since $e(n)$ grows at most polynomially, we conclude that $N\hat{v}(n)$ itself becomes exponentially small as $n\rightarrow\infty$. Since $\hat{v}(n)$ is quantized, this means that for sufficiently large $n$, we must have $N\hat{v}(n)=0$, as desired. 

\subsubsection*{Step 5: $v(n)$ is bounded}
Since $N\hat{v}(n)=0$, it immediately follows that also $Nv(n)=0$. This is a huge simplification, because it effectively allows us to remove the moduli-dependence from the problem. Indeed, recall that the original nilpotent orbit is of the form
\begin{equation}
    F^p_{\mathrm{nil}}(t) = e^{t(n)N}F^p_0\,.
\end{equation}
Now choose some fixed $t_*$ with imaginary part large enough, and note that
\begin{equation}
    v(n) = e^{t_*N-t(n)N}v(n)\in F^k_\mathrm{nil}(t_*)\,.
\end{equation}
In particular, due to the tadpole condition, $v(n)$ is bounded in Hodge norm with respect to the fixed filtration $F^k_{\mathrm{nil}}(t_*)$. Together with the quantization condition this implies that $v(n)$ can take on only finitely many values. This finishes the proof in the one-variable case.

\subsubsection{Multiple variables}
\label{subsubsec:proof_Hodge_locus_multi}
In this section we present the general multi-variable proof of finiteness of Hodge classes, based on the original work of Cattani, Deligne, and Kaplan \cite{CDK}. We also draw heavily from the formulation of the proof in \cite{schnell2014extended}. The proof is based on an inductive application of the one-variable result discussed in the previous section. To start, we set up the induction and briefly describe in which stage of the proof the induction step is used.  

\subsubsection*{Step 0: Setting up the induction}
The major complicating factor of the multi-variable proof is the fact that there are many possible hierarchies between the saxions $y_i$ that become large as one approaches a boundary of the moduli space. The strategy of \cite{CDK} is to inductively iterate over all such possible hierarchies. Indeed, it is argued that, after passing to a subsequence, one can always parametrize the sequence of saxions as
\begin{equation}
\label{eq:parametrization_saxions}
    y_k(n) = i \sum_{l=1}^d A_{kl}\,\lambda_l(n)+b_k(n)\,,
\end{equation}
where the $A_{kl}$ are real and non-negative constants, comprising the entries of an $m\times d$ matrix, and $b(n)\in\mathbb{R}^m$ is a bounded sequence. Furthermore, the sequence $\lambda(n)\in\mathbb{R}^d$ has the property that
\begin{equation}
    \frac{\lambda_k(n)}{\lambda_{k+1}(n)}\rightarrow\infty\,,\qquad \text{as $n\rightarrow\infty$}\,,
\end{equation}
for all $1\leq k\leq d$, where we put $\lambda_{d+1}(m)=1$. Next, we note that \eqref{eq:parametrization_saxions} leads to the following relation
\begin{equation}
    i\sum_{j=1}^m y_j(n)N_i = i\sum_{k=1}^d \lambda_k(n) M_k+\sum_{j=1}^m b_j(n)N_j\,,
\end{equation}
where the new monodromy matrices $M_k$ are related to the original ones by $M_k = A_{kj}N_j$. Effectively, the integer $1\leq d\leq m$ parametrizes the number of different hierarchies between the saxions. Consequently, the proof will proceed by induction on $d$. To be precise, we will show that the sequence $v(n)$ has a constant subsequence, which we denote by $v$, such that 
\begin{equation}
    M_k v=0\,,\qquad 1\leq k\leq d\,.
\end{equation}
Finally, to avoid cluttering the notation, we will set the axions $x_i$ to zero for the rest of the proof. As in the one-variable case, the axion-dependence can be straightforwardly incorporated using the monodromy matrices.

\subsubsection*{Step 1: Incorporating exponential corrections}
As in the one-variable case, it will be useful to write
\begin{equation}
    v(n) = v_{\mathrm{nil}}(n)+v_{\mathrm{inst}}(n)\,,
\end{equation}
where $v_{\mathrm{nil}}\in F^k_{\mathrm{nil}}$ is a sequence of Hodge classes with respect to a nilpotent orbit $F_{\mathrm{nil}}$, and the sequence $v_{\mathrm{inst}}(n)$ satisfies
\begin{equation}
    \frac{||v_{\mathrm{inst}}(n)||_\infty}{||v(n)||_\infty}\sim e^{-\alpha\, \mathrm{sup}(y_i)}\,,
\end{equation}
for some constant $\alpha$.

\subsubsection*{Step 2: Boundedness and $\mathfrak{sl}(2)$-weights}

Again, as in the one-variable case, the starting point of the proof is the statement that the Hodge norm $||v(n)||$ is bounded, and therefore the boundary Hodge norm $||e(n)v(n)||_\infty$ is bounded. In the one-variable case, one could immediately conclude from the latter that $v(n)$ lies inside $W_{2k}$, meaning that it only has weight components $\ell\leq 0$. In the multi-variable case, it similarly turns out to be true that $v(n)$ lies inside $W^{(1)}_{2k}$, meaning that its weight-components with respect to the first $\mathfrak{sl}(2)$ grading operator $N^0_{(1)}$ must have $\ell_1\leq 0$. However, the proof of this statement is significantly more involved. Below we present the general idea in a number of steps, but refer the reader to \cite{CDK} for the details of some of the statements.
\begin{enumerate}
    \item Suppose $v(n)\in W^{(1)}_{D+\ell_1}$ and define $\tilde{H}_{\mathbb{C}}=\mathrm{Gr}_{D+\ell_1}^{(1)}$, which supports a polarized variation of Hodge structure of weight $D+\ell_1$. Denote by $\tilde{v}(n)$ the projection of $v(n)$ onto $\tilde{H}_{\mathbb{C}}$. Suppose that $\ell_1\geq 1$, then one can show that
    \begin{equation}
    \label{eq:relation_norms_projection}
        ||\tilde{e}(n)\tilde{v}(n)||_\infty^2\leq \lambda_2(n)^{-\ell_1}||e(n)v(n)||_\infty^2 < \lambda_1(n)^{-\ell_1}||e(n)v(n)||_\infty^2\,,
    \end{equation}
    where similarly $\tilde{e}(n)$ denotes the projection of the operator $e(n)$ onto $\tilde{H}_{\mathbb{C}}$.    Intuitively, the relation \eqref{eq:relation_norms_projection} holds because the projection removes the $\lambda_1(n)$ dependence. 
    \item As a result of \eqref{eq:relation_norms_projection} and the fact that $||e(n)v(n)||_\infty$ is bounded, it follows that the expression 
    \begin{equation}
        \lambda_2(n)^{\ell_1}||\tilde{e}(n)\tilde{v}(n)||^2_\infty
    \end{equation}
    must be bounded as well. We will now argue that the highest weights of $\tilde{v}(n)$ satisfy
    \begin{equation}
    \label{eq:weights_projection}
    \ell_i=0\,,\qquad i>1\,.
    \end{equation} 
    This can be seen as follows. Since the unit vector
    \begin{equation}
        \frac{\tilde{e}(n)\tilde{v}(n)}{||\tilde{e}(n)\tilde{v}(n)||_\infty}
    \end{equation}
    converges to an element in $W^{(2)}_{D+\ell_2}\cap F^k_\infty\cap H_{\mathbb{R}}$ and $(W^{(2)}, F_\infty)$ defines a mixed Hodge structure, one has that $\ell_2\geq 0$. Now one can once more project $\tilde{v}(n)$ onto the graded space $\mathrm{Gr}_{D+\ell_2}^{(2)}$, which carries a polarized variation of Hodge structure of weight $D+\ell_2$. Denoting this projection by $\tilde{v}(n)'$ and noting that
    \begin{align}
        \lambda_3(n)^{\ell_2} ||\tilde{e}(n)' \tilde{v}(n)'||_\infty^2 &\leq \lambda_2(n)^{\ell_2} ||\tilde{e}(n)' \tilde{v}(n)'||_\infty^2\\
        &\leq \lambda_2(n)^{\ell_1}||\tilde{e}(n)\tilde{v}(n)||_\infty^2\\
        &\leq ||e(n)v(n)||^2_\infty\,,
    \end{align}
    one concludes that $\lambda_3(n)^{\ell_2} ||\tilde{e}(n)' \tilde{v}(n)'||_\infty^2$ is bounded. Proceeding by induction, it follows that $\ell_3=\ldots=\ell_d=0$. But then
    \begin{equation}
        \lambda_2(n)^{\ell_1}||\tilde{e}(n)\tilde{v}(n)||_\infty^2\geq \lambda_2(n)^{\ell_2}||\tilde{v}(n)^{(\ell_2,0,\ldots, 0)}||_\infty^2\,,
    \end{equation}
    and boundedness of the left-hand side implies that $\ell_2\leq 0$, hence $\ell_2=0$.     
    \item Returning to the sequence $v(n)$, and applying \eqref{eq:weights_projection}, one finds that
    \begin{equation}
        ||e(n)v(n)||_\infty^2\geq ||e(n)v(n)^{(\ell_1,\ldots, \ell_d)}||_\infty^2 = \left(\frac{\lambda_1(n)}{\lambda_2(n)
        }\right)^{\ell_1}||v(n)^{(\ell_1,0,\ldots, 0)}||_\infty^2\,.
    \end{equation}
    This is in contradiction with the fact that $||e(n)v(n)||_\infty^2$ must remain bounded, hence the assumption that $\ell_1\geq 1$ is false. Therefore, we conclude that $\ell_1\leq 0$.
\end{enumerate}

\subsubsection*{Step 3: Restricting the limit of $v(n)$}
The idea is now to apply the induction hypothesis to projection $\tilde{v}(n)$ of $v(n)$ onto $\tilde{H}_{\mathbb{C}}=\mathrm{Gr}_{D}^{(1)}$, which again supports a variation of Hodge structure of weight $D$. This effectively projects the relation \eqref{eq:parametrization_saxions} to
\begin{equation}
    \tilde{y}_k(n) = \sum_{l=2}^d A_{kl}\tilde{\lambda}_l(n)+\tilde{b}(n)\,.
\end{equation}
This expansion is very similar to \eqref{eq:parametrization_saxions}, but with $d-1$ terms instead of $d$. Therefore, after passing to a subsequence, the induction step implies that $\tilde{v}(n)$ has a constant value $\tilde{h}\in \tilde{H}_{\mathbb{Z}}$ and that $\tilde{M}_k \tilde{h}=0$ for $k=2,\ldots, d$. Lifting this result back to the sequence $v(n)$, we may write
\begin{equation}
    v(n) \equiv v_0\quad\mathrm{mod}\,W_{D-1}^{(1)}\,,
\end{equation}
where 
\begin{equation}
    v_0 = \sum_{\ell_2,\ldots, \ell_d\leq 0} v(n)^{(0,\ell_2,\ldots, \ell_d)}
\end{equation}
is a constant sequence that projects to $\tilde{h}$. The strategy is now similar to the one-variable case, where we have shown that $v_0$ is annihilated by the monodromy operator. In the multi-variable case, the result is not quite as strong, because the element $v_0$ contains multiple $\mathfrak{sl}(2)$-components. Indeed, we will instead show that $M_1 v_0^{(0,\ldots, 0)}=0$. To this end, let us again introduce the limiting vector
\begin{equation}
    w:=\lim_{n\rightarrow\infty}e(n)v(n)\in F^k_\infty\cap W_{2k-1}^{(1)}\cap H_{\mathbb{R}}\,.
\end{equation}
Applying the result of Lemma \ref{lem:Fsharp_W0}, one finds 
\begin{equation}
    N_1^0w = M_1w=0\,.
\end{equation}
Projecting the congruence
\begin{equation}
    e(n)v(n)\equiv e(m)v_0\quad \mathrm{mod}\,W_{2k-1}^{(1)}\,,
\end{equation}
to the weight $\ell_1=0$ eigenspace of $N_{(1)}^0$, we get that indeed $w=v_0^{(0,\ldots, 0)}$ and hence $M_1v_0^{(0,\ldots, 0)}=0$ as desired.

\begin{figure}[t!]
\centering
\begin{tikzpicture}[baseline={([yshift=-.5ex]current bounding box.center)},scale=1,cm={cos(45),sin(45),-sin(45),cos(45),(15,0)}]

\filldraw[color=blue!30] (4,4) -- (0,4) node[left]{$F_\infty^1$}-- (0,1) -- (4,1);
\filldraw[color=orange!50] (2,0) node[below right]{$W^{(1)}_{2}$} -- (0,2) -- (0,0);

\draw[step = 1, gray, ultra thin] (0, 0) grid (4, 4);

  \foreach \i\j in {0/0, 0/1, 0/2, 0/3, 0/4,
1/0, 1/1, 1/2, 1/3, 1/4,
2/0, 2/1, 2/2, 2/3, 2/4,
3/0, 3/1, 3/2, 3/3, 3/4,
4/0, 4/1, 4/2, 4/3, 4/4}
{
  \draw[fill] (\i,\j) circle[radius=0.04] ;
}

\draw[fill] (1,1) circle[radius=0.06] node[above]{$u$};
\end{tikzpicture}
\caption{Arrangement of the Deligne splitting $\tilde{I}^{p,q}_{(1)}$ associated to the mixed Hodge structure $(W^{(1)},\tilde{F}_\infty)$. In blue: $\tilde{F}^1_{(1)}$ component of the $\tilde{F}^p_{(1)}$ filtration. In orange: $W_1^{(1)}$ component of the weight filtration $W_\ell^{(1)}$. The only possible location for the limiting vector $u=\lim_{n\rightarrow\infty} u(n)$ is also indicated. Note that indeed $u$ necessarily has weight $\ell_1=-2$.}
\label{fig:Deligne_diamond_multi_var}
\end{figure}
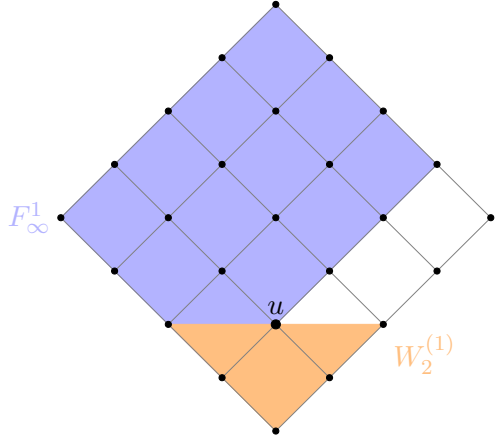

\subsubsection*{Step 4: Showing that $M_1 v(n)=0$}
Having shown that $M_1v_0^{(0,\ldots, 0)}=0$, we next show that in fact $M_1 v(n)=0$, after passing to a subsequence. The argumentation is again analogous to the one-variable case. However, it is important to note that now it is not immediately obvious whether $M_1 v_0=0$. So far we have only shown this for the $v_0^{(0,\ldots, 0)}$ component. As a result, the analysis needs to be slightly modified, but is very similar in spirit. Let us again write
\begin{equation}
    e(n)M_1 v(n)=e(n)M_1 v_{\mathrm{nil}}(n)+e(n)M_1 v_{\mathrm{inst}}(n)\,,
\end{equation}
and suppose that the ratio
\begin{equation}
    \frac{||e(n)M_1 v_{\mathrm{inst}}(n)||_\infty}{||e(n)M_1 v(n)||_\infty}
\end{equation}
goes to zero as $n\rightarrow\infty$. Consequently, one would find that
\begin{equation}
    u(n):=\frac{e(n)M_1v(n)}{||e(n)M_1v(n)||_\infty}\in W^{(1)}_{2k-2}\cap H_{\mathbb{R}}
\end{equation}
defines a sequence of unit vectors (note the appearance of $W_{2k-2}^{(1)}$ as opposed to $W_{2k-3}$ in the one-variable case!), which would converge to a unit vector $u\in F^{k-1}_\infty\cap W^{(1)}_{2k-2}\cap H_{\mathbb{R}}$. Because the index on the weight filtration is increased by one compared to the one-variable case, it is no longer the case that this space is immediately trivial. It is, however, very restricted. Indeed, another application of Lemma \ref{lem:Fsharp_W0} gives that $u$ necessarily has weight $\ell_1=-2$. This is illustrated in figure \ref{fig:Deligne_diamond_multi_var} for the case $D=4$.

At the same time, we may recall that $v(n)\equiv v_0\,\mathrm{mod}\, W^{(1)}_{2k-1}$, hence by projecting onto the $\ell_1=-2$ component, we have
\begin{equation}
    u = \frac{e(n)M_1v_0}{||e(n)M_1v(n)||_\infty}\in W^{(1)}_{2k-2}\cap \cdots \cap W^{(d)}_{2k-2}\,.
\end{equation}
Hence, we find that
\begin{equation}
    u\in F^k_\infty\cap W^{(i)}_{2k-2}\cap H_{\mathbb{R}}\,,
\end{equation}
for every $i=1,\ldots, d$. Yet another application of Lemma \ref{lem:Fsharp_W0} implies that $u$ necessarily has weights $\ell_i=-2$ for all $i=1,\ldots, d$. However, recalling that $M_1$ lowers all the weights by exactly $-2$, it must be that
\begin{equation}
    u(m)^{(-2,\ldots, -2)}=\frac{e(n)M_1 v_0^{(0,\ldots, 0)}}{||e(n)M_1v(n)||_\infty}=0\,.
\end{equation}
Hence, this is in contradiction to the fact that $u$ is a unit vector. Consequently, it must be the case that $||e(n)M_1v(n)||_\infty$ is bounded by a constant multiple of $||e(n)M_1v_{\mathrm{inst}}(n)||_\infty$. In particular, the former is also exponentially small. Since $e(n)$ grows at most polynomially, this implies that $M_1v(n)$ is exponentially small. Due to the quantization condition, we may therefore assume that $M_1v(n)=0$ after passing to a subsequence. 

\subsubsection*{Step 5: Finishing the proof}
Having argued that $M_1 v(n)=0$ the final step of the proof proceeds in the same spirit as the one-variable case. Indeed, in the one-variable case this relation allowed us to completely remove the moduli dependence of the problem. In the multi-variable case, it instead allows us to remove one of the moduli from the problem. Indeed, since the sequence of filtrations
\begin{equation}
    e^{-i\lambda_1(n)M_1} F^p_{\mathrm{nil}}
\end{equation}
no longer depends on $\lambda_1(n)$, we have effectively reduced the value of $d$ to $d-1$. By induction, we may pass to a subsequence of $v(n)$ which is constant and lies in the kernel of $M_2\,\ldots, M_d$. Since we already have $M_1 v(n)=0$ we have indeed shown that $M_k v=0$ for all $1\leq k \leq d$. Finally, when $d=1$ one may simply apply the one-variable proof. This concludes the proof of the multi-variable case. 

\bibliographystyle{JHEP}
\bibliography{references}

\end{document}